\documentclass[11pt,a4paper]{article}

\usepackage{amsthm, amsfonts,amssymb,euscript}
\usepackage{amsmath}
\usepackage{bm}
\usepackage{amssymb}
\usepackage{amsthm}
\usepackage{graphicx}
\usepackage{caption}
\usepackage{subcaption}
\usepackage{amsfonts}
\usepackage{float}
\usepackage{color}
\usepackage{slashed}
\usepackage[affil-it]{authblk}

\usepackage[usenames,dvipsnames,svgnames,table]{xcolor}
\usepackage[colorlinks=true]{hyperref}
\hypersetup{linkcolor=BrickRed, urlcolor=green, citecolor=blue, linktoc=page}

\numberwithin{equation}{section}

\newtheorem*{proposition*}{Proposition}
\newtheorem*{theorem*}{Theorem}
\newtheorem*{conjecture*}{Conjecture}
\newtheorem*{claim*}{Claim}
\newtheorem*{lemma*}{Lemma}
\newtheorem*{corollary*}{Corollary}
\newtheorem{theorem}{Theorem}[section]
\newtheorem{proposition}[theorem]{Proposition}

\newtheorem{lemma}[theorem]{Lemma}
\newtheorem{corollary}[theorem]{Corollary}

\newtheorem*{definition*}{Definition}
\newtheorem{definition}{Definition}[section]
\newtheorem*{assumption*}{\mathcal{A}ssumption}

\newtheorem*{remark*}{Remark}
\newtheorem{remark}{Remark}[section]

\newcommand{\R}{\mathbb{R}}
\newcommand{\s}{\mathbb{S}}

\newcommand{\N}{\mathbb{N}}

\newcommand{\snabla}{\slashed{\nabla}}

\allowdisplaybreaks

\begin{document}

\title{Price's law and precise late-time asymptotics for subextremal Reissner--Nordstr\"{o}m black holes}
\date{February 23, 2021}
\author[1]{Yannis Angelopoulos \thanks {yannis@caltech.edu}}
\author[2]{Stefanos Aretakis\thanks {aretakis@math.toronto.edu}}
\author[3,4]{Dejan Gajic \thanks {D.Gajic@dpmms.cam.ac.uk, d.gajic@ru.nl}}
	\affil[1]{\small The Division of Physics, Mathematics and Astronomy, Caltech,
1200 E California Blvd, Pasadena CA 91125, USA}
	\affil[2]{\small Department of Mathematics, University of Toronto, 40 St George Street, Toronto, ON, Canada}
	\affil[3]{\small Centre for Mathematical Sciences, University of Cambridge, Wilberforce Road, Cambridge CB3 0WB, UK}
	\affil[4]{\small Department of Mathematics, Radboud University, 6525 AJ Nijmegen, The Netherlands}

\normalsize

\maketitle

\begin{abstract}
In this paper, we prove precise late-time asymptotics for solutions to the wave equation supported on angular frequencies greater or equal to $\ell$ on the domain of outer communications of subextremal Reissner--Nordstr\"{o}m spacetimes up to and including the event horizon. Our asymptotics yield, in particular, sharp upper and lower decay rates which are consistent with Price's law on such backgrounds. We present a theory for inverting the time operator and derive an explicit representation of the leading-order asymptotic coefficient in terms of the Newman--Penrose charges at null infinity associated with the time integrals. Our method is based on purely physical space techniques. For each angular frequency $\ell$ we establish a sharp hierarchy of $r$-weighted radially commuted estimates with length $2\ell+5$. We complement this hierarchy with a novel hierarchy of weighted elliptic estimates of length $\ell+1$.

\end{abstract}

\tableofcontents

\section{Introduction}
\label{intro}

\subsection{Introduction and background}
\label{introandback}
Price \cite{Price1972} predicted in 1972 that if $\psi$ solves the wave equation 
\begin{equation}
\label{eq:waveequation}
\square_g\psi=0,
\end{equation}
on a Schwarzschild spacetime and is supported on a fixed angular frequency $\ell$ (in this case we denote the linear wave by $\psi_{\ell}$) then  
\begin{equation}
\psi_{\ell} \sim \frac{1}{t^{2\ell+3}}
\label{pricelaw}
\end{equation}
asymptotically in time $(t\rightarrow \infty)$ along constant  $\{r=r_0\}$ hypersurfaces. In this paper we provide the first rigorous derivation and proof of the precise late-time asymptotics for solutions to the wave equation  \eqref{eq:waveequation} on the domain of outer communications of subextremal Reissner--Nordstr\"{o}m spacetimes for all angular frequencies $\ell\geq 0$ confirming in particular Price's law \eqref{pricelaw}. We obtain the precise asymptotic behavior for the scalar field up to and including the event horizon and for the radiation field along null infinity. Our results yield in particular optimal upper and lower time-decay bounds. Such bounds are important in a wide range of problems in general relativity such as 1) the study of non-zero spin wave equations (see \cite{dhr-teukolsky-kerr}),  2) the black hole stability problem (see \cite{klainerman17}), 3) the strong cosmic censorship conjecture (see \cite{Luk2015}, \cite{Luk2016a}, \cite{gregjan})
and 4) the propagation of gravitational waves. 

The wave equation has been the object of intense study in the past decade. For a summary of the subject we refer the reader to \cite{part3, molog} and for sharp decay results we refer to \cite{other1, metal, Kronthaler2007, paper1}).  The first attempt to resolve Price's conjecture is due to  Donninger, Schlag and Soffer \cite{dssprice} where  the  $2\ell+3$  decay rate is proved for static initial data supported on the $\ell$ frequency (it is worth pointing out that for such data $2\ell + 4$ is the expected optimal decay rate). Hintz \cite{hintzprice} recently derived the $2\ell+3$ upper bound for general initial data supported on the $\ell$ frequency on Schwarzschild backgrounds. Ma \cite{ma1} also derived almost sharp decay rates for $\ell=1$ angular modes of Maxwell fields.

The first rigorous work that derived precise late-time asymptotics is \cite{paper2, paper-bifurcate, logasymptotics}. This work obtained asymptotics for general solutions $\psi$ (without any assumptions on the angular frequency) to the wave equation on sub-extremal Reissner--Nordstr\"{o}m spacetimes. Precise asymptotics were obtained in \cite{hintzprice} for some general asymptotically flat spacetimes that include as a special case the subextremal Kerr family of black hole spacetimes, and in \cite{ma2} for the Dirac equation on Schwarzschild backgrounds. On the other hand, \cite{paper4} derived precise asymptotics for general solution on extremal Reissner--Nordstr\"{o}m spacetimes. The terms appearing in these asymptotics decay much slower than in the sub-extremal case in view of the horizon instability \cite{aretakis1, aretakis2, aretakis4, aretakis2013} and the presence of conserved charges at the event horizon \cite{aretakisglue}. An application of the precise asymptotics in the extremal case was presented in \cite{extremal-prl} where it was shown that the horizon charges can be computed using only the knowledge of the radiation field at null infinity. On the other hand, proving late-time asymptotics for higher angular frequencies on extremal Reissner--Nordstr\"{o}m remains an open problem.

\subsection{Price's law and precise late-time asymptotics}
\label{sec:AFirstVersionOfTheMainTheorems}

 In this section, we provide a brief summary of our main results. 
We consider appropriate hypersurfaces $\Sigma_{\tau}$ and $\mathcal{S}_{\tau}$ in subextremal Reissner--Nordstr\"{o}m that cross the event horizon and terminate at future null infinity as depicted in the figure below
\begin{figure}[H]
\begin{center}
\includegraphics[width=6cm]{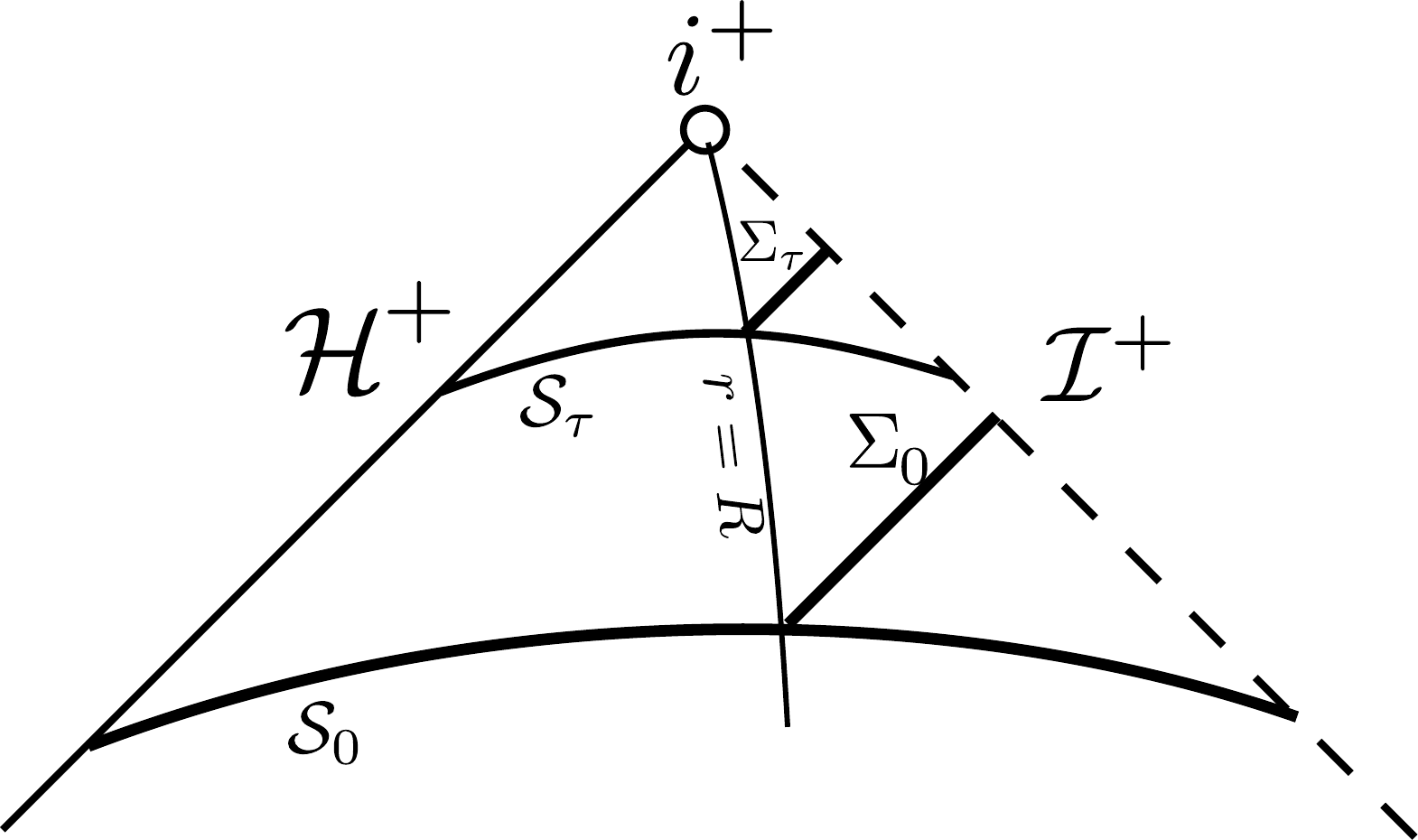}
\vspace{-0.4cm}
\caption{\label{fig:243ka}The hypersurfaces $\Sigma_{\tau}$.}
\end{center}
\end{figure}
\vspace{-0.8cm}
We derive the following precise late-time asymptotics for each angular frequency $\psi_{\ell}$. In fact our asymptotics are summable, in the sense that they hold for solutions $\psi_{\geq \ell}$ to the wave equation which are supported on angular frequencies greater or equal to $\ell$. 

\begin{equation}
 \boxed{\psi_{\geq \ell} (\tau ,r,\theta , \varphi )\sim 
{A_{\ell }^{(1)}\cdot {r^{\ell}} \cdot  I_{\ell}^{(1)} [ \psi ] ( \theta , \varphi )}\cdot \frac{1}{\tau^{2\ell+3}}  }
\label{pricesimplecompact}
\end{equation}
along the hypersurfaces $r=r_0$ for any $r_0\geq r_{+}$. Here $A_{\ell}^{(1)}$ are numerical constants that depend on $\ell$ and  $I_{\ell}^{(1)}[\psi]$ denotes what we call the $\ell^{\text{th}}$ time-inverted Newman--Penrose charge of $\psi$. This charge is equal to the Newman--Penrose charge (see Section \ref{high}) of the time integral (see Section \ref{timeinverse}) of $\psi_{\ell}$. In a region close to infinity, we prove that 
\begin{equation}
 \boxed{\psi_{\geq \ell} (u,v,\theta , \varphi )\sim 
{\bar{A}_{\ell}^{(1)}\cdot {r^{\ell}} \cdot  I_{\ell}^{(1)} [ \psi ] ( \theta , \varphi )}\cdot \frac{1}{u^{\ell+2}\cdot v^{\ell+1}}}  
\label{pricesimpleglobal}
\end{equation}
for a numerical constant $\bar{A}_{\ell}^{(1)}$ depending on $\ell$. In particular the asymptotics for the radiation field $r\psi$ along null infinity are as follows:
\begin{equation}
\boxed{
 r\psi_{\geq \ell} (\tau ,\infty ,\theta , \varphi )\sim 
 \frac{2^{4\ell} I_{\ell}^{(1)} [\psi ] ( \theta , \varphi )}{ (2\ell+1 ) \cdot \dots \cdot ( \ell +2 )} \cdot \frac{1}{\tau^{\ell+2}} . }
\label{pricesimpleradiation}
\end{equation}
We also derive asymptotics for higher-order $T$ and $\partial_\rho$ derivatives of $\psi$, where $T$ is the stationary Killing field and $\partial_{\rho}$ is the radial vector field tangential to $\Sigma_{\tau}$. See Section \ref{maintheorems} for the rigorous statements of the main theorems.

\subsection{Overview of the proof}
\label{sec:OverviewOfTheProof}

In this section we provide a summary of the main ideas of the proof of the asymptotics presented in the previous section. To make our methods clearer, we will present a schematic version of the main estimates in this section, omitting terms that do not play an important role for the structure of the arguments. We also omit difficulties such as capturing the redshift and trapping effects which have been extensively addressed in the literature. We note that all the integrals are taken with respect to the volume form corresponding to the induced metric of the integrating region. 

\subsubsection{Almost-sharp time decay of the energy flux}
\label{sec:AlmostSharpDecayOfTheEnergyFlux}
Recall that the energy-momentum tensor corresponding to a linear wave is given by:
$$ \mathbb{T}_{\mu \nu} [\psi] \doteq \partial_{\mu} \psi \partial_{\nu} \psi - \frac{1}{2} g_{\mu \nu} \partial^{\alpha} \psi \partial_{\alpha} \psi . $$
The energy current $J^V [\psi]$ for a vector field $V$ is given by
$$ J^V_{\mu} [\psi] \doteq \mathbb{T}_{\mu \nu} [\psi ] \cdot V^{\nu} . $$ 

We will first show how to obtain almost sharp decay for the standard energy flux through $\Sigma_{\tau}$:
\begin{equation}
\int_{\Sigma_{\tau}}J^{N}[\psi_{\geq \ell}]  \cdot n_{\tau}  \, d\mu_{\Sigma_{\tau}} .
\label{energyflux}
\end{equation}

Here $N$ is a globally timelike vector field such that $N=T$ away from the event horizon, $d\mu_{\Sigma_{\tau}}$ is the volume form corresponding to the hypersurface $\Sigma_{\tau}$, and $n_{\tau}$ the normal to $\Sigma_{\tau}$ (note that we will also use $n$ to denote the normal to other hypersurfaces in analogous situations without specifying it). First of all, the Dafermos--Rodnianski hierachy \cite{newmethod} schematically reads as:
\begin{equation}
\int_{\tau}\int_{\Sigma_{\tau}}\frac{r^{p-1}}{r^2}\cdot J^{N}\left[\phi_{\geq \ell} \right] \cdot n_{\tau} \, d\mu_{\Sigma_{\tau}} d\tau \lesssim \int_{\Sigma_{\tau_0}}\frac{r^{p}}{r^2}\cdot J^{N}\left[\phi_{\geq \ell} \right] \cdot n_{\tau_0}  \, d\mu_{\Sigma_{\tau_0}}
\label{drhierarchy}
\end{equation}
where $\phi_{\geq \ell}:=r\psi_{\geq \ell}$ and $0< p\leq 2$\footnote{We should note that the $r^p$-hierarchy as stated here is not quite right in the case of $p=2$, as in this case there is no term involving angular derivatives. For the sake of the schematic exposition we will ignore this, one can refer to Section \ref{section:rp} for the precise statements.}, where $A \lesssim B$ (for function $A$ and $B$) denotes $A \leq C B$ where $C$ is a constant, and where $r$ is the radial variable. This hierarchy (combined with an integrated local energy decay estimate), applied with $p=1$ and $p=2$, yields $\tau^{-2}$ decay for the energy flux \eqref{energyflux}. In order to obtain faster decay for the energy flux we need to obtain higher-order versions of \eqref{drhierarchy}. For this reason we introduce the following weighted derivatives
\begin{equation}
\widetilde{\Phi}_{(k)}=(-1)^{k}\cdot \left(r^2\partial_{r}\right)^k\phi_{\geq \ell}.
\label{unmorad}
\end{equation}
with $k\geq 0$. Here $\partial_r$ denotes the outgoing null vector field such that $\partial_rr=1$.   Assuming  that $\psi$ is supported on angular frequencies $\geq \ell$ and its initial data are decaying sufficiently fast (for example, are compactly supported) then we obtain the following schematic hierarchies:
\begin{equation}
\int_{\tau}\int_{\Sigma_{\tau}}\frac{r^{p-1}}{r^2}\cdot J^{N}\left[\widetilde{\Phi}_{(k)} \right] \cdot n_{\tau}  \, d\mu_{\Sigma_{\tau}} d\tau \lesssim \int_{\Sigma_{\tau_0}}\frac{r^{p}}{r^2}\cdot J^{N}\left[\widetilde{\Phi}_{(k)} \right] \cdot n_{\tau_0}  \, d\mu_{\Sigma_{\tau_0}}
\label{schematichier1}
\end{equation}
for 
\[0\leq k\leq \ell, \ \ \ \ \ \ \ \ 0< p\leq 2.\]
The hierarchy \eqref{schematichier1} was previously presented in \cite{paper1} for $k=1$ and $\ell \geq 1$. Here we present the full range of $k$ and $\ell$ and show that the maximum number of commutations with $r^2 \partial_r$ is exactly $\ell$ and which allows us to derive precisely $\ell+1$ hierarchies. 

In order to further extend the top order hierarchy (i.e.~$k=\ell$), we define the \textit{modified} weighted derivatives $\Phi_{(k)}$ in an iterative way as follows:
\begin{equation}
{\Phi}_{(k)}= \widetilde{\Phi}_{(k)} + \sum_{m=1}^k \alpha_{k,m} \Phi_{(k-m)},
\label{morad}
\end{equation}
where $\alpha_{k,m}$ denote numerical constants that depend on $k$ and $m$. These quantities prove useful for extending the $r^p$-weighted hierarchies to their almost sharp range, and are important in the definition of the so-called Newman--Penrose charges (that will be defined later). For the top order hierarchy with $k=\ell$ we show the following improved range for $p$:
\[0\leq p<5. \]
The above hierarchies can be connected to each other via the following Hardy inequality:
\begin{equation}
\int_{\tau} \int_{\Sigma_{\tau}}\frac{r^{p-1}}{r^2}\cdot J^{N}\left[\widetilde{\Phi}_{(k)} \right] \cdot n_{\tau}  \, d\mu_{\Sigma_{\tau}} d\tau \lesssim \int_{\tau} \int_{\Sigma_{\tau}}\frac{r^{p-3}}{r^2}\cdot J^{N}\left[\widetilde{\Phi}_{(k+1)} \right] \cdot n_{\tau_0}  \, d\mu_{\Sigma_{\tau_0}}
\label{schematichier2}
\end{equation}
that holds for all $p\neq -1$ (with appropriate boundary assumptions). We also show that if we replace $\psi_{\geq \ell}$ with $T^m\psi_{\geq \ell}$ then we get additional estimates, which can be thought of as extending the range of $p$ to
$$ 0 < p < 2+2m  $$
in \eqref{schematichier1} for $0 \leq k \leq \ell$, and
$$ 0 < p < 5+2m $$
in \eqref{schematichier2} (for both cases this is not quite correct, the hierarchies can be extended only after commuting $m$ times with $\partial_r$, and then exchanging the $\partial_r$ derivatives with $T$ derivatives).

For all hierarchies with $0\leq k\leq \ell-1$ we have essentially two estimates ($p=1,2$). For the last hierarchy $k=\ell$ we have 5 estimates. Hence, our $r^p$-hierarchy  yields the following decay rates for the energy flux and the conformal flux:
\[ \int_{\Sigma_{\tau}}J^{N}[\psi_{\geq \ell}] \cdot n_{\tau}  \, d\mu_{\Sigma_{\tau}} \lesssim \frac{1}{\tau^{2\ell+5-\epsilon}},\ \ \ \  \ \ \ \ 
 \int_{\Sigma_{\tau}}r^2\cdot J^{N}[\psi_{\geq \ell}] \cdot n_{\tau}  \, d\mu_{\Sigma_{\tau}} \lesssim \frac{1}{\tau^{2\ell+3-\epsilon}} ,\]
for any $\epsilon > 0$.

\subsubsection{Almost sharp decay for the radiation field}
\label{sec:AlmostSharpDecayForTheRadiationField}

We can derive almost sharp decay for the radiation field $r\psi_{\geq \ell}|_{\mathcal{I}^{+}}$ using the fundamental theorem of calculus and a Hardy inequality:
\begin{equation*}
\begin{split}
r\psi_{\geq \ell} \lesssim  \sqrt[4]{\int_{\Sigma_{\tau}}J^{N}[\psi_{\geq \ell}] \cdot n_{\tau}  \, d\mu_{\Sigma_{\tau}} \cdot \int_{\Sigma_{\tau}}r^2 \cdot J^N[\psi_{\geq \ell}] \cdot n_{\tau}  \, d\mu_{\Sigma_{\tau}}} 
\lesssim & \sqrt[4]{\frac{1}{\tau^{2\ell+5}}\cdot \frac{1}{\tau^{2\ell+3}}} =\frac{1}{\tau^{\ell+2-\epsilon}} .
\end{split}
\end{equation*}
This is the optimal rate for $r\psi_{\geq \ell}|_{\mathcal{I}^{+}}$, however it is far from optimal for $\psi_{\geq \ell}|_{r=r_0}$. For $\psi_{\geq \ell}$ itself we make use of 
\begin{align*} 
 -\partial_{\rho}(r^{-2\ell} \psi_{\geq \ell}^2) &=2\ell r^{-2\ell-1}\psi_{\geq \ell}^2+2r^{-2\ell}\psi_{\geq \ell}  \partial_{\rho}\psi_{\geq \ell} \\ & \leq (2\ell+1)r^{-2\ell-1}\psi_{\geq \ell}^2+ r^{-2\ell+1} ( \partial_{\rho}\psi_{\geq \ell})^2 \sim r^{-2\ell-1} r^2J^N[\psi_{\geq \ell}] \cdot n_{\tau}
 \end{align*}
to conclude that
\begin{equation}
r^{-2\ell+\epsilon} \psi_{\geq \ell}^2 \lesssim \int_{\Sigma_{\tau}} r^{-2\ell-1+\epsilon} J^N[\psi_{\geq \ell}] \cdot n_{\tau} \, d\mu_{\Sigma_{\tau}} . 
\label{psipointwiseintro}
\end{equation}
Hence, in order to get the (almost) sharp decay for $r\psi_{\geq \ell}|_{\mathcal{I}^{+}}$ it suffices to obtain the sharp decay for  $r^{-2\ell+\epsilon} \psi_{\geq \ell}^2 $ and hence of the weighted energy flux with decreasing weights $r^p$ with $p<0$. For this we present a new hierarchy of elliptic estimates.

\subsubsection{A novel hierarchy of elliptic estimates}
\label{sec:SchematicFormulationOfTheorem74}
Let us define the following weighted derivative
\[\tilde{\partial}_{\rho}= r\partial_{\rho}, \]
where $\partial_{\rho}$ is the radial tangential vector field on the asymptotically hyperboloidal hypersurfaces $\mathcal{S}_{\tau}$.

The hierarchy of elliptic estimates that we derive schematically take the following form:
\begin{equation}
\int_{\mathcal{S}_{\tau}}\frac{1}{r^{p}}\left[D^2\cdot J^{N}[\tilde{\partial}^{m+1}_{\rho}\psi_{\geq \ell}]+J^{N}[\tilde{\partial}^{m}_{\rho}\psi_{\geq \ell}]\right] \cdot n_{\tau} \, d\mu_{\mathcal{S}_{\tau}} \lesssim  \int_{\mathcal{S}_{\tau}}\frac{1}{r^{p-2}}\sum_{s=0}^{m}J^{N}[T\tilde{\partial}^{s}_{\rho}\psi_{\geq \ell}] \cdot n_{\tau} \, d\mu_{\mathcal{S}_{\tau}}
\label{commutellipticintro}
\end{equation}
with
\[-1<p<2\ell+1, \ \ \ \ \ 0\leq m\leq \ell , \]
with the degenerate factor $D$ given by \eqref{Dterm}, and where $\mathcal{S}_{\tau}$ are asymptotically hyperboloidal hypersurfaces (see Figure \ref{fig:243ka}).
Dropping the degenerate term, the uncommuted version ($m=0$) reads:
\begin{equation}
\int_{\mathcal{S}_{\tau}}\frac{1}{r^{p}}J^{N}[\psi_{\geq \ell}] \cdot n_{\tau} \, d\mu_{\mathcal{S}_{\tau}} \lesssim \int_{\mathcal{S}_{\tau}}\frac{1}{r^{p-2}}J^{N}[T\psi_{\geq \ell}] \cdot n_{\tau} \, d\mu_{\mathcal{S}_{\tau}}
\label{ellipticschematic0}
\end{equation}
with
\[-1<p<2\ell+1.\]

The elliptic estimates show that we can add a $T$ derivative at the expense of increasing the power of the weight $r$ by 2. If $\psi$ is supported on angular frequences $\geq \ell$ then we can commute $\ell$ times with $\tilde{\partial}_{\rho}$ and in this way we can then replace $\psi$ with $T\psi$ in the weighted energy fluxes at the expense of an additional $r^2$ factor. 

The proof relies on the fact that we can integrate by the equation 
$$ \partial_{\rho} ( (Dr^2 ) \partial_{\rho}^2 \psi_{\geq \ell} ) + (Dr^2 )' \partial_{\rho}^2 \psi_{\geq \ell} + 2 \partial_{\rho} \psi_{\geq \ell} = \partial_{\rho} \bar{F}_T , $$
where $F_T$ involves $T$ derivatives of $\psi_{\geq \ell}$. We note that this is possible for the proof of our estimates due to the fact that the right hand side of the last equation defines an elliptic operator. This should be compared and contrasted to the Kerr case where something like that is not possible, see section 7 of \cite{aagkerr}.

\subsubsection{Decay of energy flux with negative $r$ weights}
\label{sec:DecayOfEnergyFluxWithNegativeRWeights}

We will use the above elliptic estimates to derive decay for the weighted flux
\[\int_{\mathcal{S}_{\tau}} \frac{1}{r^{2\ell+1-\epsilon}}J^{N}[\psi_{\geq \ell}] \cdot n_{\tau} \, d\mu_{\mathcal{S}_{\tau}}, \]
where

We apply \eqref{ellipticschematic0} successively for $p=2\ell+1-\epsilon, 2\ell-1-\epsilon,...,3-\epsilon,1-\epsilon$ and we get
\begin{equation*}
\begin{split}
\int_{\mathcal{S}_{\tau}} \frac{1}{r^{2\ell+1-\epsilon}}J^{N}[\psi_{\geq \ell}]  \cdot n_{\tau} \, d\mu_{\mathcal{S}_{\tau}}
& \lesssim  \int_{\mathcal{S}_{\tau}} \frac{1}{r^{2\ell-1-\epsilon}}J^{N}[T\psi_{\geq \ell}] \cdot n_{\tau} \, d\mu_{\mathcal{S}_{\tau}} \\
& \lesssim  \int_{\mathcal{S}_{\tau}} \frac{1}{r^{2\ell-3-\epsilon}}J^{N}[T^2\psi_{\geq \ell}] \cdot n_{\tau} \, d\mu_{\mathcal{S}_{\tau}} \\
& \cdots \\
& \lesssim  \int_{\mathcal{S}_{\tau}} \frac{1}{r^{3-\epsilon}}J^{N}[T^{\ell-1}\psi_{\geq \ell}] \cdot n_{\tau} \, d\mu_{\mathcal{S}_{\tau}} \\
& \lesssim  \int_{\mathcal{S}_{\tau}} \frac{1}{r^{1-\epsilon}}J^{N}[T^{\ell}\psi_{\geq \ell}] \cdot n_{\tau} \, d\mu_{\mathcal{S}_{\tau}} \\
& \lesssim  \int_{\mathcal{S}_{\tau}} r^{1+\epsilon}J^{N}[T^{\ell+1}\psi_{\geq \ell}] \cdot n_{\tau} \, d\mu_{\mathcal{S}_{\tau}}\\
\end{split}
\end{equation*}
The top estimate follows by taking $p=2\ell+1-\epsilon$ and the bottom estimate by taking $p=1-\epsilon$ which is the admissible range for $p$. Since $  \displaystyle\int_{\mathcal{S}} r^{1+\epsilon}J^{N}[\psi_{\geq \ell}] \cdot n_{\tau} \, d\mu_{\mathcal{S}_{\tau}} $ decays like $\displaystyle\frac{1}{\tau^{2\ell+4-\epsilon}}$ we have that $\displaystyle\int_{\mathcal{S}} r^{1+\epsilon}J^{N}[T^{\ell+1}\psi_{\geq \ell}] \cdot n_{\tau} \, d\mu_{\mathcal{S}_{\tau}}$ decays like $\displaystyle\frac{1}{\tau^{2\ell+2+2(\ell+1)-\epsilon}}=\frac{1}{\tau^{2(2\ell+3)-\epsilon}}$ which via \eqref{psipointwiseintro} yields the required result for $\psi_{\geq \ell}$.

\subsubsection{Decay of higher order radial derivatives}
\label{sec:DecayOfPartialRKPsiDerivatives}

Using the commuted elliptic estimates \eqref{commutellipticintro} and the above idea we also obtain:
\begin{equation}
\frac{1}{r^{\ell-k-\epsilon}}{\partial^k_{\rho}\psi_{\geq \ell}} \lesssim \frac{1}{\tau^{2\ell+3-\epsilon}}
\label{decayderiv}
\end{equation}
for all $0\leq k\leq \ell$.  Moreover, we have
\begin{equation}
\widetilde{\Phi}_{(k)} \lesssim \frac{1}{\tau^{(\ell-k)+\epsilon}}
\label{decayrad}
\end{equation}
for all $0\leq k\leq \ell$ where $\widetilde{\Phi}_{(k)}$ is as defined in \eqref{unmorad}. Commuting with the time derivative $T^m$ increases the decay rate by $2m$ with $m\geq 0$.

\subsubsection{Newman--Penrose charges}
\label{sec:NPConstants}

In the remaining sections of the introduction, we will summarize our method that allows us to go beyond almost-sharp time-decay estimates and obtain the precise late-time asymptotics for angular modes $\psi_{\ell}$. For this reason we will assume that $\psi=\psi_{\ell}$ in the remainder of this introduction. 

We first need to introduce the Newman--Penrose charges and constants. The Newman--Penrose constants have been presented previously in the context of the Maxwell and the Einstein equations (see \cite{NP1, np2} for the case of Maxwell and Einstein equations on Minkowski, and also the recent \cite{ma1} for the Maxwell equations on Schwarzschild), but in this paper we derive them rigorously in the context of the wave equation for all angular frequencies on subextremal Reissner--Nordstr\"{o}m.

Recall that the modified derivative $\Phi_{(\ell+1)}$ was defined in \eqref{morad}. Then the limiting function on null infinity
\begin{equation}\label{npintro}
I_{\ell} [ \psi ] ( \theta , \varphi ) \doteq \lim_{r \rightarrow \infty}\Phi_{(\ell+1)} (u, r, \theta , \varphi ) ,
\end{equation}
 is independent of $u$. By further decomposing the function $I_{\ell} [ \psi ] ( \theta , \varphi )$ relative to spherical harmonics we obtain
\begin{equation}\label{np:shintro}
I_{\ell} [\psi ] (\theta, \varphi) \doteq \sum_{m=-\ell}^{\ell} I_{m , \ell} [\psi ] Y_{m , \ell} (\theta , \varphi ) , 
\end{equation}
and define the constants $I_{m , \ell}$ to be the Newman--Penrose constants of $\psi$.

\subsubsection{Near-infinity asymptotics for $I_{\ell}\neq 0$}
\label{sec:AsymptoticsForIEllNeq0}

We first derive asymptotics in the case where the initial data satisfies $I_{\ell}\neq 0$; such data are certainly not compactly supported, nonetheless our methods still provide almost sharp decay in this case. 

Let $(u,v)$ denote the standard null coordinates covering the black hole exterior. 
Working in the near-infinity region $\mathcal{B}_{\gamma_{\alpha}}$ that lies to the right of the curve $\gamma_{\alpha}=\{ v-u = v^{\alpha} \}$  for some $\alpha \in (0,1)$ (where $\alpha=\alpha({\ell})$ is a constant depending on the angular frequency $\ell$) as depicted below.
\begin{figure}[H]
\begin{center}
\includegraphics[width=6cm]{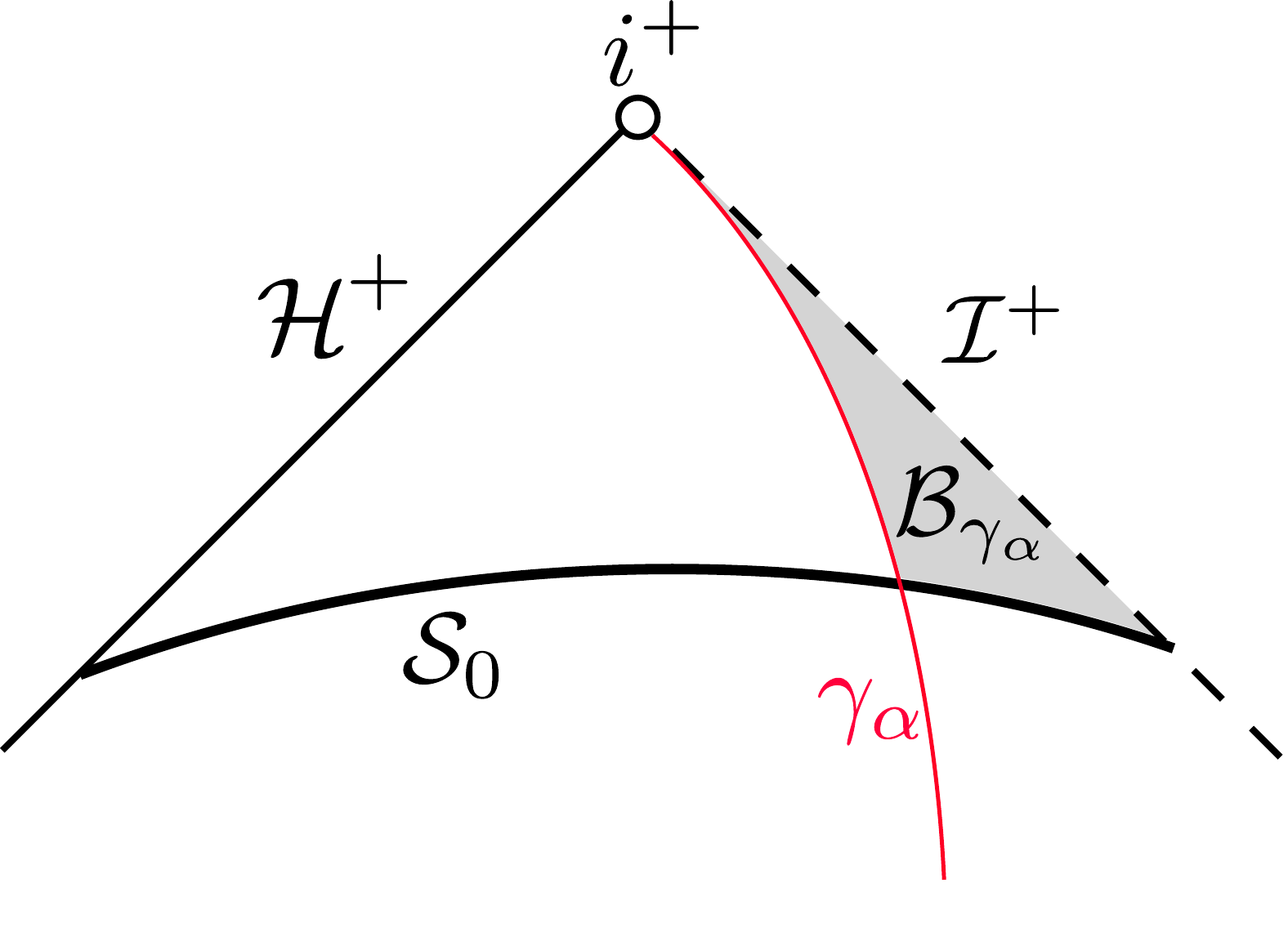}
\vspace{-0.4cm}
\caption{\label{fig:243k2}The curve $\gamma_{\alpha}$ and the region $\mathcal{B}_{\gamma_{\alpha}}$.}
\end{center}
\end{figure}
\vspace{-0.8cm} 
By successively integrating in $v$ we obtain asymptotics for radial derivatives of $\psi$ in $\mathcal{B}_{\gamma_{\alpha}}$ in the following order:
\[v^2\partial_\rho\Phi_{(\ell)} \rightarrow \partial_{\rho}\Phi_{(\ell)}\rightarrow \Phi_{(\ell)} 
\rightarrow \widetilde{\Phi}_{(\ell-1)}\rightarrow \widetilde{\Phi}_{(\ell-2)}\rightarrow \cdots \rightarrow \widetilde{\Phi}_{(1)}\rightarrow \phi \]
This also yields precise asymptotics for the radial derivatives $\frac{\partial_r^k \psi}{r^{\ell - k}}$ for $0 \leq k \leq \ell$ in the same region:
\begin{equation}
\partial_{\rho}^k\psi \sim \widetilde{A}_{\ell , k} r^{\ell-k}\cdot \frac{1}{\tau^{2\ell+3}} \label{radialasympt}
\end{equation}
for $0\leq k\leq \ell$ and numerical constants $\widetilde{A}_{\ell , k}$.

\subsubsection{Improved decay of $\partial_{\rho}^{\ell+1}\psi$ and global asymptotics}
\label{sec:MagicOfPartialEll1Psi}

One way to propagate the asymptotics for a quantity, say $Q$, from $\mathcal{B}_{\gamma_{\alpha}}$ to the rest of the black hole exterior region, all the way up to the event horizon, is to derive faster decay for the radial derivative $\partial_{\rho}Q$ and the use the fundamental theorem of calculus for $Q$. However, as is evident from \eqref{radialasympt}, the (sharp) decay rate for all the radial derivatives up to the $\ell^{\text{th}}$-order is the same. For this reason we turn to the radial derivative $\partial_{\rho}^{\ell+1}\psi$. 

Using the wave equation we have for all $k\geq 0$ that:
\begin{equation}
\partial^{k+1}_{\rho}\psi+\frac{1}{r}2k\cdot \partial^k_{\rho}\psi=\frac{1}{r^2} [ -\Delta_{\mathbb{S}^2} - k ( k+1) ] \partial^{k-1}_{\rho}\psi+... 
\label{eq3}
\end{equation} 
The omitted terms are $\partial_{\rho}$ derivatives of $T\psi$, they decay sufficiently fast and can be thought of as lower-order terms. 

If we consider $\psi = \psi_{\ell}$ we have that coefficient of the term $\partial^{k-1}_{\rho} \psi$ does not vanish for $k\leq l$, and so no result can be obtained for $\partial^{k+1}_{\rho} \psi$ since it is non-trivially coupled with the lower order derivatives. On the other hand when $k=\ell+1$, the coefficient of $\partial^{k-1}_{\rho} \psi = \partial^{\ell}_{\rho} \psi$ \textbf{vanishes}, and this yields
\begin{equation}
\partial^{\ell+2}_{\rho}\psi+\frac{1}{r}2(\ell+1)\cdot \partial^{\ell+1}_{\rho}\psi=... 
\label{eq4}
\end{equation} 
Multiplying the above equation with $r^{2(\ell+1)}$ yields
\begin{equation}
r^{2(\ell+1)}\partial^{\ell+2}_{\rho}\psi+2(\ell+1)\cdot r^{2\ell+1}\partial^{\ell+1}_{\rho}\psi= \partial_{\rho}\left(r^{2(\ell+1)}\cdot\partial^{\ell+1}_{\rho}\psi\right)=...
\label{eq2}
\end{equation} 
Integrating in $r$ the above yields decay for $\partial^{\ell+1}_{\rho} \psi$ \emph{faster} than $\tau^{-2\ell-3}$ which can then be used to propagate the precise asymptotics of $\partial^\ell_{\rho}\psi$ everywhere in the exterior region. We can then inductively propagate the asymptotics of $\partial^k_{\rho}\psi$ everywhere from $\mathcal{B}_{\alpha}$ in the black hole exterior region for all $k= \ell-1,\ell-2, ..., 2,1,0$.

Note that in the previous computation this is the only part of our proof where we use the precise form of the Reissner--Nordstr\"{o}m metric, and more specifically we use that $(Dr^2)''=2$.

\subsubsection{A general construction of the time integral}
\label{sec:AGeneralConstructionOfTheTimeIntegral}

If we assume that for an angular mode $\psi = \psi_{\ell}$ the Newman--Penrose charge $I_{\ell}[\psi]$ vanishes then we can construct the time integral of $\psi$, namely a suitably regular solution $\tilde{\psi}$ to the wave equation such that $T\tilde{\psi}=\psi$. We define the time-inverted Newman--Penrose charge of $\psi$ to be 
\[I_{\ell}^{(1)}[\psi]= I_{\ell}[\tilde{\psi}]. \]
We note that the construction of $\tilde{\psi}$  requires repeated integration in $r$ (the number of integrations depending on $\ell$), however one can still obtain explicit, yet complicated, formulas for $I_{\ell}^{(1)}[\psi]$ in terms of the initial data of $\psi$. 

We also note that in our case we can still construct the time-inverse by solving an ODE, similarly to the spherically symmetric case, and in contrast to the Kerr case where such an approach is not possible and where this step involves the inversion of an operator that is elliptic outside the ergoregion, see section 9 of \cite{aagkerr} for details.

\subsubsection{Asymptotics for $I_{\ell}=0$}
\label{sec:AsymptoticsForIEll0}

We now have all the tools needed to prove the precise asymptotics \eqref{pricesimplecompact}, \eqref{pricesimpleglobal} and \eqref{pricesimpleradiation} in the case where $I_{\ell} [\psi ] =0$ (again for an angular mode $\psi = \psi_{\ell}$). First we construct the time integral $\tilde{\psi}$ of $\psi$ and show that generically its Newman--Penrose charge is non-vanishing. We can then apply the methods of Sections \ref{sec:MagicOfPartialEll1Psi} and
\ref{sec:AsymptoticsForIEllNeq0} to obtain the asymptotics for $\tilde{\psi}$. Finally, commuting the derived estimates with $T$ yields the asymptotics for $\psi$.

Moreover, we note that if the time-inverted $\ell$-th Newman--Penrose charge happens to vanish as well, we can repeat the process, and prove asymptotics that decay one power faster. Using our results, one can also try to recover expressions similar to the ones given for $I_0^{(1)}$ in \cite{paper-bifurcate}, for compactly supproted data. In this case, for time-symmetric data (i.e. $T\psi|_{\Sigma_0} = 0$) that also vanish at the event horizon, the time-inverted $\ell$-th Newman--Penrose charge will vanish as well, hence in that situation we can show the polynomial late-time law of $\tau^{-2\ell-4}$ (this should be contrasted with the suboptimal rates obtained for such data in \cite{other1} and \cite{dssprice}).

As a final remark, let us also note that one can use the results of the present paper and the methods developed in \cite{logasymptotics} to obtain second order asymptotics.

\subsection{Acknowledgements}
\label{sec:Acknowledgements}

The second author (S.A.) acknowledges support through the NSERC grant 502581 and the Ontario Early Researcher Award.

\section{The main theorems}
\label{maintheorems}
In this section we state the main results that we obtain in this paper. 

We state first the precise late-time asymptotics obtained for a linear wave restricted to angular frequencies $\geq \ell$ for some $\ell \geq 1$ (we denote this spectral projection by $P_{\geq \ell}$, and by $P_{\ell}$ we denote the spectral projection to frequency $\ell$) with vanishing $\ell$-th Newman--Penrose charge, but non-vanishing time-inverted $\ell$-th Newman--Penrose charge (such data are considered to be the most naturally physical ones and include ).

\begin{theorem}\label{thm:main2}
Let $\psi$ be a solution of the wave equation \eqref{eq:waveequation} in the domain of outer communications of a subextremal Reissner--Nordstr\"{o}m spacetime up to and including the horizon. We assume that our data are smooth and compactly supported, which implies that its Newman--Penrose charges are vanishing. We additionally assume that for some $\ell \geq 1$ its time-inverted $\ell$-th Newman--Penrose charge is non-vanishing, i.e.
$$ I^{(1)}_{\ell} [\psi ] \neq 0 . $$

Let $m \in \mathbb{N}$. For $\Phi_{(k)}$ the quantities that are given in Proposition \ref{prop:np}, we assume that
\begin{align*}
& E^{(1)} \doteq  \sum_{k \leq 2\ell+1} \Big( \int_{\mathcal{N}_{u_0}} r^{5-\delta} ( L ( P_{\ell} T^{m+k} \Phi_{(\ell)} ) )^2 \,d\omega dv + \int_{\mathcal{N}_{u_0}} r^2 ( L ( P_{\ell+1} T^{m+k} \Phi_{(\ell)} ) )^2 \,d\omega dv  \\ & + \int_{\mathcal{N}_{u_0}} r^2 ( L ( P_{\geq \ell+2} T^{m+k} \Phi_{(\ell)} ) )^2 \,d\omega dv \Big) + \sum_{s=0}^{\ell-1} \sum_{|\alpha | \leq m+\ell+1  , k \leq m+3\ell+1} \int_{\mathcal{N}_{u_0}} r^{2-\delta} | \slashed{\nabla}^{\alpha}_{\mathbb{S}^2} (L ( P_{\geq \ell+2} T^k \Phi_{(s)} ) ) |^2 \,d\omega dv\\ &+\sum_{0 \leq k+l \leq 3\ell + m +3} \int_{\Sigma_{u_0}} J^N [ N^k T^l \psi ] \cdot n_{u_0} \, d\mu_{\Sigma_{u_0}} < \infty . 
\end{align*}

Then we have for $(u,v,\theta , \varphi ) \in \{ r \geq R \}$ for some $R > r_+$ that
\begin{equation}\label{asym:rf0}
\left| \frac{P_{\geq \ell} ( T^m \psi ) (u,v,\theta , \varphi )}{r^{\ell}} - A^{(1)}_{\ell , m} I^{(1)}_{\ell} [ \psi ] ( \theta , \varphi )T^m \left( \frac{1}{u^{\ell+2} v^{\ell+1}} \right) \right| \leq C E^{(1)} u^{-\ell-2-m-\eta'_1}{v^{-\ell-1}} , 
\end{equation}
for some $\eta'_1 > 0$, for $C = C ( D , R  , m , \ell, \eta'_1 ) > 0$,  and for a numerical constant $A^{(1)}_{\ell , m}$ that depends on $\ell$ and $m$.

For $(\tau , r ,\theta , \varphi ) \in \{ r \leq R \}$ we have that
\begin{equation}\label{asym:w0}
\left| \frac{P_{\geq \ell} ( T^m \psi ) (\tau ,r,\theta , \varphi )}{r^{\ell}} - \frac{A^{(1)}_{\ell , m} I^{(1)}_{\ell} [ \psi ] ( \theta , \varphi )}{\tau^{2\ell+3+m}}  \right| \leq C E^{(1)} \tau^{-2\ell -3 -m-\eta'_2} , 
\end{equation}
for some $\eta'_2 > 0$ and for $C = C ( D , R , m , \ell , \eta'_2) > 0$.
\end{theorem}

We now state the precise asymptotics obtained for a linear wave localized at angular frequencies $\geq \ell$ for some $\ell \geq 1$ with non-vanishing $\ell$-th Newman--Penrose constant.

\begin{theorem}\label{thm:main1}
Let $\psi$ be a solution of the wave equation \eqref{eq:waveequation} in the domain of outer communications of a subextremal Reissner--Nordstr\"{o}m spacetime up to and including the horizon. Assume that its $\ell$-th Newman--Penrose constant is non-vanishing for some fixed $\ell \geq 1$, i.e.
$$ I_{\ell} [\psi ] \neq 0 . $$
Let $m \in \mathbb{N}$. Moreover we assume that our data are smooth and that:
\begin{equation*}
\sup_r \sum_{s\leq 2} \sum_{l \leq m}\left(\int_{\s^2}|\snabla_{\s^2}^{s+1} ( T^l \Phi_{(\ell)} )|^2\,d\omega\right)(u_0,r)<\infty ,
\end{equation*}
and
\begin{equation*}
\sup_r \sum_{s\leq 2} \sum_{l \leq m} \sum_{0\leq k\leq \ell-1} \sum_{0\leq j\leq \ell-k} \left(\int_{\s^2}|\snabla_{\s^2}^{s+1} \slashed{\Delta}_{\s^2}^{j} ( T^l \Phi_{(k)} ) |^2\,d\omega\right)(u_0,r)<\infty,
\end{equation*}
for all $r$, where the $\Phi_{(k)}$'s are given in Proposition \ref{prop:np}, and that
\begin{align*}
& E \doteq   \sum_{k \leq 2\ell+1} \Big( \int_{\mathcal{N}_{u_0}} r^{3-\delta} ( L ( P_{\ell} T^{m+k} \Phi_{(\ell)} ) )^2 \,d\omega dv + \int_{\mathcal{N}_{u_0}} r^2 ( L ( P_{\ell+1} T^{m+k} \Phi_{(\ell)} ) )^2 \,d\omega dv  \\ & + \int_{\mathcal{N}_{u_0}} r^2 ( L ( P_{\geq \ell+2} T^{m+k} \Phi_{(\ell)} ) )^2 \,d\omega dv \Big) + \sum_{s=0}^{\ell-1} \sum_{|\alpha | \leq m+\ell+1  , k \leq m+3\ell+1} \int_{\mathcal{N}_{u_0}} r^{2-\delta} | \slashed{\nabla}_{\mathbb{S}^2}^{\alpha} (L ( P_{\geq \ell+2} T^k \Phi_{(s)} ) ) |^2 \,d\omega dv\\ &+\sum_{0 \leq k+l \leq 3\ell + m +3} \int_{\Sigma_{u_0}} J^N [ N^k T^l \psi ] \cdot n_{u_0} \, d\mu_{\Sigma_{u_0}} < \infty . 
\end{align*}

Then we have for $(u,v,\theta , \varphi ) \in \{ r \geq R \}$ for some $R > r_+$ that
\begin{equation}\label{asym:rfn0}
\left| \frac{P_{\geq \ell} ( T^m \psi ) (u,v,\theta , \varphi )}{r^{\ell}} - A_{\ell,m} I_{\ell} [ \psi ] ( \theta , \varphi ) T^m \left( \frac{1}{u^{\ell+1} v^{\ell+1}} \right) \right| \leq C E u^{-\ell-1-m-\eta_1}{v^{-1-\ell}} , 
\end{equation}
for some $\eta_1 > 0$, for $C = C ( D,R, m,\ell,\eta_1 )> 0$, and for a quantity $A_{\ell,m}$ that depends on $\ell$ and $m$.

For $(\tau , r ,\theta , \varphi ) \in \{ r \leq R \}$ we have that
\begin{equation}\label{asym:wn0}
\left| \frac{P_{\geq \ell} (T^m \psi ) (\tau ,r,\theta , \varphi )}{r^{\ell}} - \frac{A_{\ell , m} I_{\ell} [ \psi ] ( \theta , \varphi )}{\tau^{2\ell+2+m}}  \right| \leq C E \tau^{-2\ell -2-m -\eta_2} , 
\end{equation}
for some $\eta_2 > 0$ and for $C = C(D,R,m,\ell,\eta_2) > 0$.
\end{theorem}

Finally we state the precise asymptotics obtained for the \textit{radiation field} of a linear wave localized at frequencies $\geq \ell$ for some $\ell \geq 1$ with vanishing $\ell$-th Newman--Penrose charge, but non-vanishing time-inverted $\ell$-th Newman--Penrose charge, at future null infinity $\mathcal{I}^+$.
\begin{theorem}
Under the assumptions of Theorem \ref{thm:main2} we have that along null infinity the following estimate holds true:
\begin{equation}\label{asym:inf1}
\left| T^m \phi_{\geq \ell} (u,\infty,\theta , \varphi )  + \frac{2^{4\ell} I_{\ell}^{(1)} [\psi ] ( \theta , \varphi )}{ (2\ell+1 ) \cdot \dots \cdot ( \ell +2 )} T^m ( u^{-2-\ell} ) \right| \leq C E^{(1)} u^{-2-\ell-m-\eta^{(1)}} ,
\end{equation}
for some $\eta^{(1)} > 0$, $E^{(1)}$ as in Theorem \ref{thm:main2}, and $C = C(D,R,m,\ell,\eta^{(1)})$.

On the other hand under the assumptions of Theorem \ref{thm:main1} we have the following estimate along null infinity:
\begin{equation}\label{asym:inf2}
\left| T^m \phi_{\geq \ell} (u,\infty,\theta,\varphi) - \frac{2^{4\ell}I_{\ell} [ \psi ] (\theta,\varphi)}{(2\ell+1)\cdot \ldots \cdot ( \ell +1 )} T^m ( u^{-1-\ell} ) \right| \leq C E u^{-1-\ell-m-\eta} ,
\end{equation}
for some $\eta > 0$, $E$ as in Theorem \ref{thm:main1}, and $C = C(D,R,m,\ell,\eta)$.

\end{theorem}

\section{Preliminaries}

\subsection{The subextremal Reissner--Nordstr\"{o}m spacetimes}
The Reissner--Nordstr\"{o}m spacetimes are the unique spherically symmetric and asymptotically flat 2-parameter family of solutions of the Einstein--Maxwell equations, the two parameters being the mass $M > 0$ and the electromagnetic charge $e$. We consider a subextremal Reissner--Nordstr\"{o}m black hole spacetime $(\mathcal{M} , g)$ where $g$ in ingoing Eddington--Finkelstein coordinates $(v,r,\theta , \varphi)$ has the form:
\begin{equation*}
g=-Ddv^2+2dvdr+r^2(d\theta^2+\sin^2\theta d\varphi^2),
\end{equation*}
with 
\begin{equation}
D=1-\frac{2M}{r}+ \frac{e^2}{r^2}
\label{Dterm}
\end{equation} where $0 \leq | e | < M$. 

Note that $D$ has two roots $0 < r_{-} < r_{+}$. We will work work in the domain of outer communications $\mathcal{D}$ of $\mathcal{M}$ up to and including the future event horizon, that is we will always have $r \geq r_{+}$, the future event horizon being the hypersurface $\mathcal{H}^{+} := \{ (v,r,\theta , \varphi ) | r = r_{+} \}$.

In Bondi coordinates $(u,r,\theta , \varphi)$ (that are valid outside the event horizon) $g$ has the form:
\begin{equation*}
g=-Ddu^2-2dudr+r^2(d\theta^2+\sin^2\theta d\varphi^2).
\end{equation*}
In $\mathcal{D}$ we have that $u \in \mathbb{R}$, $r \in [r_+ , \infty)$, $\theta \in ( 0 , \pi )$ and $\varphi \in (0, 2 \pi )$. By using the tortoise coordinate $r^{*} (r) = r + \frac{1}{2\kappa_+} \ln \left| \frac{r -r_+}{r_+} \right|  + \frac{1}{2\kappa_-} \ln \left| \frac{r -r_-}{r_-} \right|  + C$, for a constant $C$ (that we will fix shortly), for $\kappa_{\pm} =  \left. \frac{1}{2} \frac{dD}{dr} \right|_{r = r_{\pm}}$, we define $v = u+2r^{*}$ and in double coordinates $(u,v,\theta , \varphi)$ the metric $g$ takes the form
$$ g = -Ddudv + r^2(d\theta^2+\sin^2\theta d\varphi^2), $$
and in these coordinates we can define \emph{future null infinity} $\mathcal{I}^+$ which is the limiting null hypersurface foliated by 2-spheres where the null hypersurfaces $\{ u = u' \}$ terminate as $v \rightarrow \infty$ (i.e. it is the limiting hypersurface that is formed by the limit points as $r \rightarrow \infty$ of future null geodesics). We additionally note that on $r_{ph} = \frac{3M}{2} \left( 1 + \sqrt{1 - \frac{8e^2}{9M^2}} \right)$ (where we choose $C$ such that $r^{*}(r_{ph}) = 0$) we have the \emph{photon sphere} which imposes a derivative loss in the Morawetz estimates that we present in the next section.

The vector field $T = \partial_u$ (in $(u,r,\theta , \varphi)$) is Killing. By $\snabla_{\mathbb{S}^2}$ we denote the covariant derivative on $\mathbb{S}^2$ with respect to its standard metric, and by $\slashed{\Delta}_{\mathbb{S}^2}$ the Laplacian on $\mathbb{S}^2$. Moreover by $\Omega_i$, $i \in \{1,2,3\}$ we denote the three Killing vector fields associated to $\mathbb{S}^2$ which can be expressed as:
$$ \Omega_1 := \sin \varphi \partial_{\theta} + \cot \theta \cos \varphi \partial_{\varphi} , $$
$$ \Omega_2 := - \cos \varphi \partial_{\theta} + \cot \theta \sin \varphi \partial_{\varphi} , $$
$$ \Omega_3 := \partial_{\varphi} , $$
and using the above we can also define the $\Omega^{\alpha}$ vector fields by $\Omega^{\alpha} := \Omega_1^{\alpha_1} \Omega_2^{\alpha_2} \Omega_3^{\alpha_3}$ where $( \alpha_1 , \alpha_2 , \alpha_3 ) \in \mathbb{N}^3$ and $\alpha = \sum_{i=1}^3 \alpha_i$. We also consider the vector field
$$ \partial_{\rho} := \partial_r + h(r) T , $$
for a smooth function $h : [ r_+ , \infty) \rightarrow \mathbb{R}$ such that
$$ \frac{1}{\max_{r \in [ r_+ , R]} D(r)} \leq h(r) < \frac{2}{D(r)} \mbox{  for $r \leq R$,  } $$ $$ 0 < \frac{2}{D(r)} - h(r) = O (r^{-1-\eta} ) \mbox{  for $r > R$, for some $\eta > 0$, for some $R > r_+$} . $$ 

We will consider two foliations of $\mathcal{D}$. The first foliation consists of spacelike-null hypersurfaces $\Sigma_{\tau}$ that are constructed as follows: consider a spacelike and asymptotically flat hypersurface $\mathbf{\Sigma}$ that intersects $\mathcal{H}^+$ on a 2-sphere, a null hypersurface $\mathcal{N}_{\tau} := \{ (u=\tau , r , \theta , \varphi) | r \geq R$ for some $R > r_+$, take some $u_0$ such that $( \mathbf{\Sigma} \cap \{ r = R \} ) = ( \{ u = u_0 \} \cap \{ r = R \} )$, and consider
$$ \Sigma_{u_0} := ( \mathbf{\Sigma} \cap \{ r \leq R \} ) \cup \mathcal{N}_{u_0} . $$
Then define $\Sigma_{\tau} : = f_{\tau} ( \Sigma_{u_0} )$ for $f$ the flow of $T$. The second foliation consists of hyperboloidal hypersurfaces $\mathcal{S}_{\tau}$, where 
$$\mathcal{S} := \{ ( v,r,\theta , \varphi ) | v = v_0 - \int_r^R h (r' ) \, dr' , r \geq r_+ \} , $$
and $v_0$ is large enough and depends on $h$ and $R$, and now we can define $\mathcal{S}_{\tau} := f_{\tau} ( \mathcal{S} )$. Finally in the region $\{ r \geq R \}$ we define also $\mathcal{I}_{\tau} ( \tau_1 , \tau_2 ) := \{ v = \tau | u_0 \leq \tau_1 \leq u \leq \tau_2 \} \cap \{ r \geq R \}$ in double null coordinates.  

Let us introduce the projection operators $P_{\geq \ell},P_{\ell}: L^2(\s^2)\to L^2(\s^2)$ by noticing that any smooth function $f$ can be written as
$$ f = \sum_{\ell = 0}^{\infty} P_{\ell} f , $$
where 
$$ P_{\ell} f (u,r,\theta , \varphi ) = \sum_{m=-\ell}^{\ell} f_{m , \ell} (u,r) Y_{m, \ell} (\theta , \varphi ) , $$
for certain functions $f_{m , \ell}$ and where $Y_{m , \ell}$ are the \emph{spherical harmonics} which form a complete basis of eigenfunctions of $\slashed{\Delta}_{\mathbb{S}^2}$ on $L^2 (\mathbb{S}^2 )$. Note that
$$ \slashed{\Delta}_{\mathbb{S}^2} ( P_{\ell} f ) = - \ell ( \ell +1) f . $$
We also define
$$ P_{\geq \ell} f = \sum_{l \geq \ell} P_l f . $$ 

Finally by $f = O (x)$ we mean that there exist a constant $K$ such that $|f(x) | \leq K x$ for all $x$.

We note that the linear wave equation 
$$ \Box_g \psi = 0 , $$
is globally well-posed in $\mathcal{D}$ for smooth data $(\psi , n_{\Sigma_{u_0}} \psi ) |_{\Sigma_{u_0}\cap \{r\leq R\}}$ and $(\psi , \partial_v \psi) |_{\mathcal{N}_{u_0}}$, or smooth data $(\psi , n_{\mathcal{S}_{u_0}} \psi ) |_{\mathcal{S}_{u_0}}$ (see \cite{alinhacbook}). Note that we can consider data either on $\Sigma_{u_0}$ or $\mathcal{S}_{u_0}$, as determining data on either hypersurface can always give data on a later $\mathcal{S}_{u'}$ or $\Sigma_{u'}$ hypersurface respectively by solving a local problem.

\textbf{Note that in the rest of the paper we will always assume that we work with a smooth solution $\psi$ of the linear wave equation on the domain of outer communications up to and including the event horizon of a subextremal Reissner--Nordstr\"{o}m black hole spacetime.} In particular we will consider smooth data with respect to the conformal compactification of $\Sigma_{u_0}$, which implies that
\begin{equation}\label{basic:as}
\sup_r \left. \left(\int_{\s^2}|\snabla_{\s^2}\slashed{\Delta}_{\s^2}^{j}(r^2 \partial_r )^k T^m ( r\psi ) |^2\,d\omega\right) \right|_{\Sigma_{u_0}} (r) < \infty,
\end{equation}
for any $j$, $k$, $m \in \mathbb{N}$. Note that we can always assume less (usually we will work with $P_{\ell} \psi$ or $P_{\geq \ell} \psi$ and we need to assume the previous estimate only for finite $j$ and $k$, their range dependent on $\ell$) but we make the aforementioned assumption to slightly simplify the statements of our results. 
\subsection{Basic estimates for linear waves on subextremal Reissner--Nordstr\"{o}m spacetimes}
Let $N$ be a vector field defined as follows (in ingoing $(v,r,\theta , \varphi)$ coordinates):
$$ N := T - \partial_r \mbox{  for $r \in [r_+ , r_1 ]$,} $$
$$ N := T \mbox{  for $r \geq r_2$,} $$
for $r_+ \leq r_1 < r_2$. Assuming that
 \begin{equation*}
 \int_{\Sigma_{u_0} }J^N[\psi]\cdot n_{\Sigma_{u_0}}\,d\mu_{\Sigma_{u_0}}<\infty,
 \end{equation*}
there exists a uniform constant $C>0$, such that for all $v$
\begin{equation}
\label{est:enb}
\int_{\Sigma_{\tau}}J^N[\psi]\cdot n_{\tau}\,d\mu_{\Sigma_{\tau}}+\int_{\mathcal{I}_v(u_0,\tau)}J^N[\psi]\cdot \underline{L}\: r^2d\omega du\leq C \int_{\Sigma_{u_0}}J^N[\psi]\cdot n_{\Sigma_{u_0}}\,d\mu_{\Sigma_{u_0}}.
\end{equation}

In the region $\{ r \geq R \}$ (where $R$ was chosen in the definition of the $\Sigma$ hypersurfaces) we have the following estimate (which is a combination of a Morawetz estimate and a \textit{redshift} estimate, see \cite{lecturesMD}): there exists a uniform constant $C>0$, such that for all $u_0<\tau_1<\tau_2$
\begin{equation}
\label{morawetz}
\int_{\tau_1}^{\tau_2}\left(\int_{\Sigma_{\tau}\setminus \mathcal{N}_{\tau}}J^N[ \psi]\cdot n_{\tau}\,d\mu_{\Sigma_{\tau}}\right)\,d\tau\leq C \int_{\Sigma_{\tau_1}}( J^N[\psi]\cdot n_{\tau_1}+J^N[T\psi]\cdot n_{\tau_1} )\,d\mu_{\Sigma_{\tau_1}}.
\end{equation}
For a proof of the aforementioned estimate see the lecture notes \cite{lecturesMD}.

Close to the horizon in the region $\mathcal{A}$ (that is away from the photon sphere) we have the following local Morawetz estimates $\psi$ that does not lose derivatives:
\begin{equation}
\label{redshift:as}
\int_{\tau_1}^{\tau_2}\left(\int_{\Sigma_{\tau}\cap \mathcal{C}\}}J^T[ \partial^{\alpha}\psi]\cdot n_{\tau}\,d\mu_{\Sigma_{\tau}}\right)\,d\tau\leq C_{\alpha} \sum_{k\leq |\alpha|}\int_{{\Sigma}_{\tau_1}}J^N[T^{k}\psi]\cdot n_{\tau_1}d\mu_{\Sigma_{\tau_1}},
\end{equation}
\begin{equation}
\label{redshift:ash}
\int_{\tau_1}^{\tau_2}\left(\int_{\mathcal{S}_{\tau}\cap \mathcal{C}\}}J^T[ \partial^{\alpha}\psi]\cdot n_{\tau}\,d\mu_{\Sigma_{\tau}}\right)\,d\tau\leq C_{\alpha} \sum_{k\leq |\alpha|}\int_{{\mathcal{S}}_{\tau_1}}J^N[T^{k}\psi]\cdot n_{\tau_1}d\mu_{\Sigma_{\tau_1}},
\end{equation}

for$\mathcal{C} \cap ( r_{ph} - \delta , r_{ph} + \delta) = \emptyset$ for some $\delta > 0$ (so the spacetime region of integration is taken to be away from the photon sphere), where $C_{\alpha}>0$ depends on $\mathcal{C}$ and the choice of $\alpha$ (see \cite{redshift}).

\section{Higher-order radiation fields and Newman--Penrose charges}
\label{high}
We denote with
\begin{equation*}
\phi:=r\cdot \psi.
\end{equation*}
the \emph{zeroth order radiation field} corresponding to a solution $\psi$ to \eqref{eq:waveequation}. Using that $\psi$ is a solution to \eqref{eq:waveequation} it follows that $\phi$ satisfies the equation:
\begin{equation}
\label{eq:radfieldequation}
2\partial_u\partial_r\phi=\partial_r(D\partial_r\phi)-D'r^{-1}\phi+r^{-2}\slashed{\Delta}_{\s^2}\phi.
\end{equation}
See for example Appendix of \cite{paper1} for a derivation. Setting $x = \frac{1}{r}$, in $(u,x,\theta , \varphi )$ coordinates, equation \eqref{eq:radfieldequation} turns into: 
\begin{equation}\label{eq:radfieldequationx}
2\partial_u\partial_x\phi+\partial_x(Dr^{-2}\partial_x\phi)+\slashed{\Delta}_{\s^2}\phi-(2Mx-2e^2 x^2)\phi=0,
\end{equation}
as $\partial_x = -r^2 \partial_r$. Note that
\begin{equation*}
Dr^{-2}=x^2-2Mx^3+e^2x^4.
\end{equation*}
Then, if we define $\widetilde{\Phi}_{(\ell)}=\partial_x^{\ell}\phi$ (note that $\widetilde{\Phi}_{(\ell)}=(-1)^{\ell} (r^2\partial_r)^{\ell}\phi$), we obtain inductively the following equation :
\begin{equation}
\label{eq:horadfields}
\begin{split}
0=&\:2\partial_u\partial_x\widetilde{\Phi}_{(\ell)}+\partial_x(Dr^{-2}\partial_x\widetilde{\Phi}_{(\ell)})+2\ell(x-3Mx^2+2e^2 x^3)\widetilde{\Phi}_{(\ell+1)}+\slashed{\Delta}_{\s^2}\widetilde{\Phi}_{(\ell)}\\
&+[\ell(\ell +1)(1-6Mx+6e^2x^2)-(2Mx-2e^2x^2)]\widetilde{\Phi}_{(\ell)}\\
&+[-(\ell+1)\ell(\ell-1)(2M-4e^2x)-\ell(2M-4e^2x)]\widetilde{\Phi}_{(\ell-1)}\\
&+[(\ell+1)\ell(\ell-1)(\ell-2)e^2+2\ell(\ell-1) e^2]\widetilde{\Phi}_ {(\ell-2)}.
\end{split}
\end{equation}
For a derivation see Section 6.1 of \cite{gajwar19a}. Using the previous equation we show the existence of conserved quantities at infinity in Proposition \ref{prop:np} below. Note that a similar derivation has been presented also by Ma in \cite{ma1} in the context of the Maxwell equations.
\begin{proposition}\label{prop:np}
Let $\ell \in \N_0$. Let ${\Phi}_{(\ell)}$ satisfy the following inductive definition:
\begin{align*}
\Phi_{(0)}:=&\:\phi,\\
\Phi_{(\ell)}:=&\:\widetilde{\Phi}_{(\ell)}+ \sum_{k=1}^{\ell} \alpha_{\ell,k}\Phi_{(\ell-k)} ,
\end{align*}
for appropriate constants $\{ \alpha_{\ell,k} \}_{k=1}^{\ell}$. Then
\begin{equation}
\label{eq:NPquantity}
\begin{split}
\:2\partial_u\partial_x\Phi_{(\ell)}+ & \partial_x(Dr^{-2}\partial_x\Phi_{(\ell)})+2\ell(x-3Mx^2+2e^2 x^3)\Phi_{(\ell)} \\ & +[\ell(\ell+1)+\slashed{\Delta}_{\s^2}]\Phi_{(\ell)} + +\sum_{k=0}^{\ell} O (x) \widetilde{\Phi}_{(k)} = 0 .
\end{split}
\end{equation}

\end{proposition}
\begin{proof}
We will establish \eqref{eq:NPquantity} by strong induction. Note that \eqref{eq:NPquantity} holds for $n=0$. Now suppose \eqref{eq:NPquantity} holds for $n\leq N$, $N\in \N_0$. Then by combining \eqref{eq:horadfields} with $n=N+1$ and \eqref{eq:NPquantity} with $n=N$ and $n=N-1$, we obtain the following equation for $\Phi_{(N+1)}$:
\begin{equation*}
\begin{split}
0&=\:2\partial_u\partial_x\Phi_{(N+1)}+\partial_x(Dr^{-2}\partial_x\Phi_{(N+1)})+2(N+1)(x-3Mx^2+2e^2 x^3)\Phi_{(N+1)}\\
&+[(N+1)(N+2)+\slashed{\Delta}_{\s^2}]\widetilde{\Phi}_{(N+1)}\\
&+ \sum_{k=0}^N [ c_k + \alpha_{N+1 , N+1 - k} ( N+1 -k)  ] \widetilde{\Phi}_{(k)} \\
&+\sum_{k=0}^{N+1}O (x) \widetilde{\Phi}_{(k)} ,
\end{split}
\end{equation*}
where the $c_k$'s are constants that depend on the constants $\{ \alpha_{n,k} \}_{k=0}^n$ that have already been determined due to the induction hypothesis. In order to conclude that \eqref{eq:NPquantity} holds also for $n=N+1$, we must therefore have that
\begin{equation*}
c_k + \alpha_{N+1 , N+1-k} (N+1-k) = 0 \mbox{  for every $k \in \{0, \ldots , N \}$} ,
\end{equation*}
which of course we can always do as each one of these equations is a linear algebraic equation in $\alpha_{N+1 , m}$, $m \in \{0, \ldots , N\}$.
\end{proof}

We will refer to $\Phi_{(n)}$ and $\widetilde{\Phi}_{(n)}$ as the \emph{$n$-th order radiation fields}. In outgoing $(u,r,\theta, \varphi)$ coordinates let
$$ L = \frac{1}{2} D \partial_r , \quad \underline{L} = \partial_u - \frac{1}{2} D \partial_r , $$
and we have the following equations for the $\widetilde{\Phi}_{(\ell)}$'s and the $\Phi_{(\ell)}$'s:
\begin{corollary}
\label{cor:cancel}
Fix $\ell\geq 1$. Then $\widetilde{\Phi}_{(\ell)}$ satisfies:
  \begin{equation}
 \label{eq:maincommeq}
 \begin{split}
    2\partial_r\partial_u \widetilde{\Phi}_{(\ell)}=&\:\partial_r(D\partial_r \widetilde{\Phi}_{(\ell)})+r^{-2}\slashed{\Delta}_{\s^2}\widetilde{\Phi}_{(\ell)}+[-2\ell r^{-1}+O(r^{-2})] \partial_r\widetilde{\Phi}_{(\ell)}+ \left[\ell(\ell+1)r^{-2} +O(r^{-3})\right]\widetilde{\Phi}_{(\ell)}\\
    &+\sum_{k=0}^{\ell-1} O(r^{-2}) \widetilde{\Phi}_{(k)}
    \end{split} 
     \end{equation}
     or equivalently,
       \begin{equation}
 \label{eq:maincommeqv2}
 \begin{split}
   4L\underline{L}\widetilde{\Phi}_{(\ell)}=&\:Dr^{-2}\slashed{\Delta}_{\s^2}\widetilde{\Phi}_{(\ell)}+[-4\ell r^{-1}+O(r^{-2})] L\widetilde{\Phi}_{(\ell)}+ D\left[\ell(\ell+1)r^{-2} +O(r^{-3})\right]\widetilde{\Phi}_{(\ell)}\\
    &+\sum_{k=0}^{\ell-1} O(r^{-2}) \widetilde{\Phi}_{(k)}.
    \end{split} 
     \end{equation}
Furthermore, for $\Phi_{(\ell)}$ we have that $P_{\ell}\Phi_{(\ell)}$ satisfies:
  \begin{equation}
    \label{eq:maineqNpquantfixedl}
    \begin{split}
    2\partial_r\partial_u (P_{\ell}\Phi_{(\ell)})=&\:\partial_r(D\partial_r (P_{\ell}\Phi_{(\ell)}))+(-2\ell r^{-1}+O(r^{-2}))\partial_r(P_{\ell}\Phi_{(\ell)})\\
    &+ O(r^{-3})P_{\ell}\Phi_{(\ell)}+\sum_{k=0}^{\ell-1}O(r^{-3})P_{\ell}\Phi_{(k)},
    \end{split} 
    \end{equation}
\end{corollary}
or equivalently,
  \begin{equation}
 \label{eq:maineqNpquantfixedlv2}
 \begin{split}
   4L\underline{L}(P_{\ell}\Phi_{(\ell)})=&\:[-4\ell r^{-1}+O(r^{-2})] L(P_{\ell}\Phi_{(\ell)})+\sum_{k=0}^{\ell} O(r^{-3})P_{\ell} \Phi_{(k)}.
    \end{split} 
     \end{equation}


Finally we note that from equation \eqref{eq:maineqNpquantfixedlv2} it follows that the $\ell$-th Newman--Penrose quantities are finite along $\mathcal{I}^+$. 
\begin{proposition}
\label{prop:finitenesshoradfields}
Fix $\ell \geq 1$, and assume that
$$ \int_{\Sigma} J^T [\psi ] \cdot n_{\Sigma} \, d\mu_{\Sigma} < \infty . $$

For any $u\geq u_0$ if
\begin{equation*}
\lim_{r\to \infty} \slashed{\Delta}_{\s^2}^j\widetilde{\Phi}_{(k)}(u_0,r,\theta,\varphi)<\infty,
\end{equation*}
with $0\leq j\leq \ell-k$ and $0\leq k\leq \ell$, then
\begin{equation*}
\lim_{r\to \infty} \widetilde{\Phi}_{(\ell)}(u,r,\theta,\varphi)<\infty.
\end{equation*}

\end{proposition}
The proof of the above Proposition can be done inductively in the same as the proof of Proposition 6.2 of \cite{paper4} (see also section 3.2 of \cite{paper1}). 

Using the aforementioned Proposition \ref{prop:finitenesshoradfields} and Proposition \ref{prop:np} we get the following important result.

\begin{theorem}\label{thm:np}
Fix $\ell \geq 1$ and assume that
\begin{equation*}
\lim_{r\to \infty} \slashed{\Delta}_{\s^2}^j\widetilde{\Phi}_{(k)}(u_0,r,\theta,\varphi)<\infty,
\end{equation*}
for all $0\leq j\leq \ell-k$ and $0\leq k\leq \ell$. Then $\partial_x ( P_{\ell} \Phi_{\ell} )$ is conserved along $\mathcal{I}^+$.
\end{theorem}

\begin{definition}\label{def:np}
Let
\begin{equation}\label{np}
I_{\ell} [ \psi ] ( \theta , \varphi ) \doteq \lim_{r \rightarrow \infty} (-1)^{\ell+1} \partial_x ( P_{\ell} \Phi_{(\ell) } ) (u, r, \theta , \varphi ) ,
\end{equation}
which is independent of $u$ as shown in Proposition \ref{prop:np}. We will refer to the quantities $I_{\ell} [\psi]$ (usually omitting the dependence on $\theta$ and $\varphi$) as the \textit{$\ell$-th Newman--Penrose charges} of a linear wave $\psi$. Note that we can also have a spherical harmonics decomposition of each Newman--Penrose charge as
\begin{equation}\label{np:sh}
I_{\ell} [\psi ] (\theta, \varphi) \doteq \sum_{m=-\ell}^{\ell} I_{m , \ell} [\psi ] Y_{m , \ell} (\theta , \varphi ) , 
\end{equation}
and the Newman--Penrose constants $I_{m , \ell}$ determine the value of the Newman--Penrose charge. 

We will also refer to the quantities $\partial_x \Phi_{(\ell )}$ as the \textit{$\ell$-th Newman--Penrose quantities} (note that there is no angular frequency localization in the latest quantities).
\end{definition}

\section{Hierarchies of $r^p$-weighted estimates}\label{section:rp}

\subsection{$r^p$-weighted estimates for $\Phi_{(\ell)}$}
\begin{proposition}
\label{prop:generalrpest}
Fix $\ell\in \N$ and consider a smooth solution $\psi$ to \eqref{eq:waveequation} satisfying \eqref{basic:as}. Then there exists $R>0$ sufficiently large, such that for $p\in (-4\ell,2]$ and for all $u_0 \leq u_1\leq u_2$:
\begin{equation*}
\begin{split}
\int_{\mathcal{N}_{u_2}}& r^p(L(P_{\geq \ell}\widetilde{\Phi}_{(\ell)}))^2\,d\omega dv+ \frac{1}{2}\int_{u_1}^{u_2} \int_{\mathcal{N}_u}(p+4\ell)r^{p-1}(L(P_{\geq \ell}\widetilde{\Phi}_{(\ell)}))^2\,d\omega dv du\\
&+\frac{1}{8}\int_{u_1}^{u_2} \int_{\mathcal{N}_u}(2-p)r^{p-3}D\left(|\snabla_{\s^2}(P_{\geq \ell}\widetilde{\Phi}_{(\ell)})|^2-\ell(\ell+1)(P_{\geq \ell}\widetilde{\Phi}_{(\ell)})^2\right)\,d\omega dv du\\
\leq&\: C\int_{\mathcal{N}_{u_1}}r^p(L(P_{\geq \ell}\widetilde{\Phi}_{(\ell)}))^2\,d\omega dv+ C\sum_{k\leq \ell}\int_{\Sigma_{u_1}} J^T[T^kP_{\geq \ell}\psi]\cdot n_{u_1}\,d\mu_{\Sigma_{u_1}},
\end{split}
\end{equation*}
where $C=C(D,R,\ell)>0$ and we can take $R=(p+4\ell)^{-1}R_0(\ell,D)>0$. 
\end{proposition}
The proof of the aforementioned Proposition follows along the same lines as the proof of Proposition 6.5 of \cite{paper4} by using equation \eqref{eq:maincommwaveq}.

Now we present an extended hierarchy for a linear wave localized at frequency $\ell$ with non-zero $\ell$-th Newman--Penrose charge.
\begin{proposition}
\label{prop:fixedlrpest}
Fix $\ell\in \N$ and consider a smooth solution $\psi$ to \eqref{eq:waveequation} satisfying \eqref{basic:as}. Then there exists $R>0$ sufficiently large, such that for $p\in (-4\ell,4)$ and for all $u_0\leq u_1\leq u_2$:
\begin{equation*}
\begin{split}
\int_{\mathcal{N}_{u_2}}& r^p(L(P_{\ell}\Phi_{(\ell)}))^2\,d\omega dv+ \frac{1}{2}\int_{u_1}^{u_2} \int_{\mathcal{N}_u}(p+4\ell)r^{p-1}(L(P_{\ell}\Phi_{(\ell)}))^2\,d\omega dv du\\
\leq&\: C\int_{\mathcal{N}_{u_1}}r^p(L(P_{\ell}\Phi_{(\ell)}))^2\,d\omega dv+ C\sum_{k\leq \ell}\int_{\Sigma_{u_1}} J^T[T^kP_{\ell}\psi]\cdot n_{u_1}\,d\mu_{\Sigma_{u_1}},
\end{split}
\end{equation*}
where $C=C(D,R,\ell)>0$ and we can take $R=\max\{(p+4\ell)^{-1},(p-4)^{-2}\}R_0(\ell,D)>0$.
\end{proposition}

\begin{proof}
Let us use $\psi$ to denote $P_{\ell}\psi$, for convenience. We then proceed as in the proof of Proposition \ref{prop:generalrpest}, using instead \eqref{eq:maineqNpquantfixedl}, and then we use that for $P_{\ell}\psi$ the Poincar\'e \emph{inequality} becomes an \emph{equality}:
\begin{equation*}
\int_{\s^2}|\snabla_{\s^2}\Phi_{(\ell)}|^2\,d\omega=\ell(\ell+1)\int_{\s^2}\Phi_{(\ell)}^2\,d\omega,
\end{equation*}
which allows us to write
\begin{equation}
\label{eq:maineqrpestfixedl}
\begin{split}
\int_{\mathcal{N}_{u_2}}& r^p(L(\chi\Phi_{(\ell)}))^2\,d\omega dv+ \frac{1}{2}\int_{u_1}^{u_2} \int_{\mathcal{N}_u}[ (p+4\ell)r^{p-1} + O (r^{p-2}) ](L(\chi\Phi_{(\ell)}))^2\,d\omega dv du\\
=&\:\int_{\mathcal{N}_{u_1}} r^p(L(\chi\Phi_{(\ell)}))^2\,d\omega dv+ J_1+ \sum_{|\alpha|\leq 1} \int_{u_1}^{u_2} \int_{\mathcal{N}_u}r^{p-2}\underline{L} (\chi\Phi_{(\ell)}) \cdot \mathcal{R}_{\chi}[\partial^{\alpha}\Phi_{(\ell)}]\,d\omega dvdu,
\end{split}
\end{equation}
with
\begin{align*}
J_1:=&\:\sum_{k=0}^{\ell}\int_{u_1}^{u_2} \int_{\mathcal{N}_u} O(r^{p-3}) \chi \Phi_{(k)}\cdot L(\chi \Phi_{(\ell)})\,d\omega du dv,
\end{align*}
We apply a weighted Young's inequality to estimate
\begin{equation*}
|J_1|\leq \int_{u_1}^{u_2} \int_{\mathcal{N}_u} \epsilon (p+4\ell) r^{p-1}(L(\chi \Phi_{(\ell)}))^2\,d\omega du dv+ C_{\epsilon}(p+4\ell)^{-1}\sum_{k=0}^{\ell}\int_{u_1}^{u_2}\int_{\mathcal{N}_u}r^{p-5}\chi^2\Phi_{(k)}^2\,d\omega du dv.
\end{equation*}
As in the proof of Proposition \ref{prop:generalrpest}, the right-hand side of the above equation can be absorbed into the left-hand side of \eqref{eq:maineqrpestfixedl} if $p<3$ by applying Hardy's inequality \eqref{hardy1} successively to the second integral, where the $R_{\chi}$ terms that appear can be estimated using \eqref{morawetz}.
\end{proof}

Next we show how the above hierarchy for a wave localized at frequency $\ell$ can be extended even further by assuming that its $\ell$-th Newman--Penrose constant vanishes. 

\begin{proposition}
\label{prop:fixedlrpest2}
Fix $\ell \in \N$ and consider a smooth solution $\psi$ to \eqref{eq:waveequation} satisfying \eqref{basic:as}. Then there exists $R>0$ sufficiently large, such that for $p \in (-4\ell, 4)$ and for all $u_0\leq u_1\leq u_2$ we have that:
\begin{equation*}
\begin{split}
\int_{\mathcal{N}_{u_2}}& r^p(L(P_{\ell}\Phi_{(\ell)}))^2\,d\omega dv+ \frac{1}{2}\int_{u_1}^{u_2} \int_{\mathcal{N}_u}(p+4\ell)r^{p-1}(L(P_{\ell}\Phi_{(\ell)}))^2\,d\omega dv du\\
\leq&\: C\int_{\mathcal{N}_{u_1}}r^p(L(P_{\ell}\Phi_{(\ell)}))^2\,d\omega dv+ C\sum_{k\leq \ell}\int_{\Sigma_{u_1}} J^T[T^kP_{\ell}\psi]\cdot n_{u_1}\,d\mu_{\Sigma_{u_1}},
\end{split}
\end{equation*}
where $C=C(D,R,\ell)>0$ and we can take $R=\max\{(p+4\ell)^{-1},(p-4)^{-2}\}R_0(\ell,D)>0$.

Furthermore there exists $R>0$ sufficiently large, such that for $p\in (-4\ell,5)$ and for all $u_0\leq u_1\leq u_2$ we have that:
\begin{equation*}
\begin{split}
\int_{\mathcal{N}_{u_2}}& r^p(L( P_{\ell} \Phi_{(\ell)} ))^2\,d\omega dv+ \frac{1}{2}\int_{u_1}^{u_2} \int_{\mathcal{N}_u}(p+4\ell)r^{p-1}(L( P_{\ell} \Phi_{(\ell)} ))^2\,d\omega dv du\\
\leq&\: C\int_{\mathcal{N}_{u_1}}r^p(L( P_{\ell} \Phi_{(\ell)} ))^2\,d\omega dv+ C\sum_{k\leq \ell}\int_{\Sigma_{u_1}} J^T[T^k P_{\ell} \psi]\cdot n_{u_1}\,d\mu_{\Sigma_{u_1}}+C(p-5)^{-1}E_{\textnormal{aux},\ell},
\end{split}
\end{equation*}
where $C=C(D,R,\ell)>0$ and we can take $R=(p+4\ell)^{-2}R_0(\ell,D)>0$ and
\begin{equation*}
E_{\textnormal{aux},\ell}=\int_{\mathcal{N}_0} r^{4-\epsilon}(L(P_{\ell}\Phi_{(k)}))^2\,d\omega dv+\sum_{j=0}^{\ell}E_1[T^jP_{\ell}\psi] ,
\end{equation*}
where
\begin{equation*}
\begin{split}
E_1[\psi]:=&\:\sum_{m\leq 4}\int_{\Sigma_0} J^N[T^m\psi]\cdot n_0\,d\mu_{\Sigma_0}+\sum_{m\leq 2} \int_{\mathcal{N}_0}r^2(LT^m\phi)^2+r(LT^{m+1}\phi)^2\,d\omega dv\\
&+ \int_{\mathcal{N}_0}r^{2-\epsilon}(L\widetilde{\Phi}_{(1)})^2+r^{1-\epsilon}(LT\widetilde{\Phi}_{(1)})^2\,d\omega dv.
\end{split}
\end{equation*}
\end{proposition}
\begin{proof}
For the proof of the first estimate in the range $(-4\ell , 4)$ we follow the exact same proof as in Proposition \ref{prop:fixedlrpest}. Note that the difference in the two cases comes from the assumption on $I_{\ell} [\psi]$. In the current case the right-hand side would be infinite for $p \in [3,4)$ if $I_{\ell} [\psi]$ was non-vanishing.

In the remaining range $p \in [4,5)$ we proceed again as in the proof of Proposition \ref{prop:fixedlrpest} and we rewrite \eqref{eq:maineqrpestfixedl} by taking a supremum on the right-hand side:
\begin{equation}
\label{eq:maineqrpestfixedl2}
\begin{split}
\sup_{u_1\leq u \leq u_2}\int_{\mathcal{N}_{u}}& r^p(L(\chi\Phi_{(\ell)}))^2\,d\omega dv+ \frac{1}{2}\int_{u_1}^{u_2} \int_{\mathcal{N}_u}(p+4\ell)r^{p-1}(L(\chi\Phi_{(\ell)}))^2\,d\omega dv du\\
=&\: \int_{\mathcal{N}_{u_1}} r^p(L(\chi\Phi_{(\ell)}))^2\,d\omega dv + J_1+ \sum_{|\alpha|\leq 1} \int_{u_1}^{u_2} \int_{\mathcal{N}_u}r^{p-2}\underline{L} (\chi\Phi_{(\ell)}) \cdot \mathcal{R}_{\chi}[\partial^{\alpha}\Phi_{(\ell)}]\,d\omega dudv,
\end{split}
\end{equation}
Now we will estimate $|J_1|$ by absorbing it into the (first) \emph{flux} integral on the left-hand side of \eqref{eq:maineqrpestfixedl2} rather than the (second) spacetime integral. We apply a Cauchy--Schwarz inequality with weights in $r$ \emph{and} $u$ to estimate:
\begin{equation*}
r^{p-3}\chi\Phi_{(k)}\cdot L(\chi \Phi_{(\ell)})\leq \epsilon (u+1)^{-1-\eta}r^{p}L(\chi \Phi_{(\ell)})+C_{\epsilon}(u+1)^{1+\eta} r^{p-6}\Phi_{(k)}^2.
\end{equation*}
We can estimate
\begin{equation*}
\epsilon \int_{u_1}^{u_2} \int_{\mathcal{N}_u}(u+1)^{-1-\eta}r^{p}L(\chi \Phi_{(\ell)})\,d\omega dv\leq \epsilon (u_1+1)^{-\eta}\cdot \sup_{u_1\leq u \leq u_2}\int_{\mathcal{N}_{u}}r^p(L(\chi\Phi_{(\ell)}))^2\,d\omega dv,
\end{equation*}
and we can absorb the right-hand side into the left-hand side of \eqref{eq:maineqrpestfixedl2}. 

We note that 
\begin{equation}\label{dec:aux}
\int_{\s^2} (P_{\ell}\Phi_{(\ell)})^2(u,r,\theta,\varphi)\,d\omega \leq Cu^{-3+\epsilon}E_{\textnormal{aux},\ell},
\end{equation}
in $\{ r \geq R\}$ as we will show in the next section in Lemma \ref{lm:auxdecay}, and using the above estimate we have that
\begin{equation*}
 \int_{u_1}^{u_2} \int_{\mathcal{N}_u}(u+1)^{1+\eta} r^{p-6}\Phi_{(\ell)}^2\,d\omega du dv\leq CE_{\textnormal{aux},\ell} \int_{u_1}^{u_2}\int_{R}^{\infty}(u+1)^{-2+\epsilon+\eta}r^{p-6}\,dr du
 \end{equation*}
The integral on the right-hand side is bounded for $p<5$ and gives a factor $(p-5)^{-1}$.
\end{proof}

It remains to show estimate \eqref{dec:aux} which we do in Lemma \ref{lm:auxdecay} in section \ref{energy}.

\subsubsection{Aside: sharpness of the hierarchy for $P_{\ell} \Phi_{(\ell)}$}
In this Section we will show that Proposition \ref{prop:fixedlrpest} cannot hold true for $p \geq 3$, and that the final estimate of Proposition \ref{prop:fixedlrpest2} cannot hold true for $p \geq 5$.

\begin{proposition}\label{prop:blowup}
Fix some $\ell \geq 1$ and consider a solution $\psi$ of \eqref{eq:waveequation} such that
$$ I_{\ell} [\psi ] \neq 0 . $$
Then the range of $p$ in Proposition \ref{prop:fixedlrpest} is \textit{sharp}. Specifically for any $u_0 \leq u' < \infty$ we have that:
$$ \int_{u'}^{\infty} \int_{\mathcal{N}_u} r^2 ( L ( P_{\ell} \Phi_{(\ell)} ) )^2 \, d\omega dv du = \infty . $$ 

\end{proposition}
\begin{proof}
By the definition of the $\ell$-th Newman--Penrose charge, for some $\bar{R} = \bar{R}(u') > R$ and for some constant $c$ we have that:
\begin{equation}\label{as:sharp}
r^4 ( \partial_r ( P_{\ell} \Phi_{(\ell)} ) )^2 \geq C I_{\ell} [\psi ] \mbox{  for all $ r \geq \bar{R}$.} 
\end{equation}
We use the same computation as in Proposition \ref{prop:fixedlrpest} for $p=3$ but in the region $\{ u \geq u' , r \geq \bar{R} , v \leq v_{\ell} \}$ for some for some fixed $v_{\ell} < \infty$ and we have that:
\begin{align*}
\frac{1}{2}\int_{u'}^{\infty} & \int_{\mathcal{N}_u \cap \{v \leq v_{\ell} \}}[ (p+4\ell)r^{p-1}+ O (r^{p-2}) ](L(\chi\Phi_{(\ell)}))^2\,d\omega dv du - J_1 \\ &+ \sum_{|\alpha|\leq 1} \int_{u'}^{\infty} \int_{\mathcal{N}_u \cap \{ v \leq v_{\ell} \}}r^{p-2}\underline{L} (\chi\Phi_{(\ell)}) \cdot \mathcal{R}_{\chi}[\partial^{\alpha}\Phi_{(\ell)}]\,d\omega dvdu
= \int_{\mathcal{N}_{u'} \cap \{v \leq v_{\ell} \}} r^3 (L(\chi\Phi_{(\ell)}))^2\,d\omega dv  ,
\end{align*} 
where we recall that
\begin{align*}
J_1:=&\:\sum_{k=0}^{\ell}\int_{u_1}^{u_2} \int_{\mathcal{N}_u} O(r^{p-3}) \chi \Phi_{(k)}\cdot L(\chi \Phi_{(\ell)})\,d\omega du dv.
\end{align*}
By the computations in Proposition \ref{prop:fixedlrpest} we note that
$$ | J_1 | + \sum_{|\alpha|\leq 1} \int_{u'}^{\infty} \int_{\mathcal{N}_u \cap \{ v \leq v_{\ell} \}}r^{p-2}\underline{L} (\chi\Phi_{(\ell)}) \cdot \mathcal{R}_{\chi}[\partial^{\alpha}\Phi_{(\ell)}]\,d\omega dvdu \leq C (\psi , u' ) < \infty , $$
for finite $u'$. Hence we get the desired result noticing that due to \eqref{as:sharp} we have that:
$$ \lim_{v_{\ell} \rightarrow \infty} \int_{\mathcal{N}_{u'} \cap \{v \leq v_{\ell} \}} r^3 (L(\chi\Phi_{(\ell)}))^2\,d\omega dv = \infty . $$
\end{proof}
In the case $I_{\ell} = 0$ we can use the time-inversion construction of Section \ref{timeinverse} and show the following:
\begin{proposition}
Fix some $\ell \geq 1$ and consider a solution $\psi$ of \eqref{eq:waveequation} such that
$$ I_{\ell} [\psi ] = 0 . $$
Then the range of $p$ in the final estimate of Proposition \ref{prop:fixedlrpest} is \textit{sharp}. Specifically for any $u_0 \leq u' < \infty$ we have that:
$$ \int_{u'}^{\infty} \int_{\mathcal{N}_u} r^4 ( L ( P_{\ell} \Phi_{(\ell)} ) )^2 \, d\omega dv du = \infty . $$ 

\end{proposition}
The proof is similar to the one of Proposition \ref{prop:blowup} and will hence be omitted.

\subsection{$r^p$-weighted estimates for $T^m\Phi_{(\ell)}$}
In this subsection we present an extended hierarchy for the higher order radiation fields after we commute the equation with $T$.
\begin{proposition}
\label{prop:fixedlrpdrestdrg}
Fix $\ell\in \N$ and consider a smooth solution $\psi$ to \eqref{eq:waveequation} satisfying \eqref{basic:as}. Then there exists $R>0$ sufficiently large, such that for $p\in (-4\ell + 2m,2m+2)$ and for all $u_0\leq u_1\leq u_2$:
\begin{equation}\label{est:drt}
\begin{split}
\int_{\mathcal{N}_{u_2}}& r^p(L^{m+1} (P_{\geq \ell} \widetilde{\Phi}_{(\ell)}))^2\,d\omega dv+ \int_{u_1}^{u_2} \int_{\mathcal{N}_u} r^{p-1}(L^{m+1} (P_{\geq \ell}\widetilde{\Phi}_{(\ell)}))^2\,d\omega dv du \\ & + \int_{u_1}^{u_2} \int_{\mathcal{N}_u} r^{p-1} | \slashed{\nabla}_{\mathbb{S}^2} ( L^m ( P_{\geq \ell} \widetilde{\Phi}_{(\ell)} ) ) |^2 \, d\omega dv du \\
\leq&\: C\sum_{l \leq m} \int_{\mathcal{N}_{u_1}}r^{p-2l} (L^{m+1-l}(P_{\geq \ell}\widetilde{\Phi}_{(\ell)}))^2\,d\omega dv+ C\sum_{k\leq m}\int_{\Sigma_{u_1}} J^T[T^k P_{\geq \ell}\psi]\cdot n_{u_1}\,d\mu_{\Sigma_{u_1}},
\end{split}
\end{equation}
where $C=C(D,R,m,\ell)>0$.
\end{proposition}
The proof follows along the same lines as the proof of Proposition 7.6 of \cite{paper4} and will not be repeated here.

\begin{proposition}
\label{prop:fixedlrpdrestdr}
Fix $\ell \in \N$ and consider a smooth solution $\psi$ to \eqref{eq:waveequation} satisfying \eqref{basic:as}. Then there exists $R>0$ sufficiently large, such that for $p\in (-4\ell + 2m,2m+3)$ and for all $u_0\leq u_1\leq u_2$:
\begin{equation*}
\begin{split}
\int_{\mathcal{N}_{u_2}}& r^p(L^{m+1} (P_{\ell}\Phi_{(\ell)}))^2\,d\omega dv+ \int_{u_1}^{u_2} \int_{\mathcal{N}_u} r^{p-1}(L^{m+1} (P_{\ell}\Phi_{(\ell)}))^2\,d\omega dv du\\
\leq&\: C\sum_{l \leq m} \int_{\mathcal{N}_{u_1}}r^{p-2l} (L^{m+1-l}(P_{\ell}\Phi_{(\ell)}))^2\,d\omega dv+ C\sum_{k\leq m}\int_{\Sigma_{u_1}} J^T[T^k P_{\ell}\psi]\cdot n_{u_1}\,d\mu_{\Sigma_{u_1}},
\end{split}
\end{equation*}
where $C=C(D,R,m,\ell)>0$.
\end{proposition}

\begin{proof}
We commute equation \eqref{eq:maineqNpquantfixedlv2} with $L^m$ and we have that 
\begin{equation}\label{eq:maineqNpquantfixedlv3}
\begin{split}
2 \underline{L} L^{m+1} ( P_{\ell} \Phi_{(\ell)} ) = & [ -2\ell r^{-1} + O (r^{-2} ) ] L^{m+1} ( P_{\ell} \Phi_{(\ell)} ) \\ & + \sum_{m_1 + m_2 = m , m_2 < m} O (r^{-1-m_1} ) L^{m_2 + 1} ( P_{\ell} \Phi_{(\ell)} )  \\ & + \sum_{k=0}^{\ell-1} \sum_{m_1 + m_2 = m} O (r^{-3-m_1} ) L^{m_2} ( P_{\ell} \Phi_{(k)} ) .
\end{split}
\end{equation}
This implies that
\begin{equation}\label{eq:maineqNpquantfixedlv3rp}
\begin{split}
\underline{L} ( r^p L^{m+1} ( \chi P_{\ell} \Phi_{(\ell)} )^2 ) = & - \frac{p}{2} D r^{p-1} ( L^{m+1} ( \chi P_{\ell} \Phi_{(\ell)} )^2 ) + r^p ( L^{m+1} ( P_{\ell} \Phi_{(\ell)} ) ) 2 \underline{L} ( L^{m+1}  \cdot  \\ & + [ -2\ell r^{p-1} + O (r^{p-2} ) ] ( L^{m+1} ( \chi P_{\ell} \Phi_{(\ell)} ) )^2 \\ & + \sum_{m_1 + m_2 = m , m_2 < m} O (r^{p-1-m_1} ) ( L^{m_2 + 1} ( P_{\ell} \Phi_{(\ell)} ) ) \cdot ( L^{m+1} ( \chi P_{\ell} \Phi_{(\ell)} ) )   \\ & + \sum_{k=0}^{\ell-1} \sum_{m_1 + m_2 = m} O (r^{p-3-m_1} ) ( L^{m_2} ( \chi P_{\ell} \Phi_{(k)} ) ) \cdot ( L^{m+1} ( \chi P_{\ell} \Phi_{(\ell)} ) ) \\ & + \sum_{| \alpha | \leq m+1} ( L^{m+1} ( \chi P_{\ell} \Phi_{(\ell)} ) ) \cdot \mathcal{R}_{\chi} [ \partial^{\alpha} ( P_{\ell} \Phi_{(\ell)} ) ] ,
\end{split}
\end{equation}
where $\chi$ and $\mathcal{R}_{\chi}$ are defined as before ($\chi (r) = 0$ for $r \leq R$ and $\chi(r) = 1$ for $r \geq R$, and the terms of $\mathcal{R}_{\chi}$ are supported in $\{ R \leq r \leq R+1\}$). Using the identity above we obtain the following:
\begin{equation}\label{est:auxdr}
\begin{split}
\int_{\mathcal{N}_{u_2}}  & r^p ( L^{m+1} ( \chi P_{\ell} \Phi_{(\ell)} ) )^2 \, d\omega dv +   \frac{1}{2} \int_{u_1}^{u_2} \int_{\mathcal{N}_u}  [ ( p +4\ell ) r^{p-1} + O ( r^{p-2} ) ] ( L^{m+1} ( \chi P_{\ell} \Phi_{(\ell)} ) )^2 \, d\omega dvdu \\ = &   \int_{\mathcal{N}_{u_2}} r^p ( L^{m+1} ( \chi P_{\ell} \Phi_{(\ell)} ) )^2 \, d\omega dv \\ & + \sum_{m_1 + m_2 = m , m_2 < m}\int_{u_1}^{u_2} \int_{\mathcal{N}_u}  O (r^{p-1-m_1} ) ( L^{m_2 + 1} ( P_{\ell} \Phi_{(\ell)} ) ) \cdot ( L^{m+1} ( \chi P_{\ell} \Phi_{(\ell)} ) )  \, d\omega dvdu \\ & + \sum_{k=0}^{\ell-1} \sum_{m_1 + m_2 = m} \int_{u_1}^{u_2} \int_{\mathcal{N}_u}  O (r^{p-3-m_1} ) ( L^{m_2} ( \chi P_{\ell} \Phi_{(k)} ) ) \cdot ( L^{m+1} ( \chi P_{\ell} \Phi_{(\ell)} ) ) \, d\omega dvdu \\ & + \sum_{| \alpha | \leq m+1}\int_{u_1}^{u_2} \int_{\mathcal{N}_u}  ( L^{m+1} ( \chi P_{\ell} \Phi_{(\ell)} ) ) \cdot \mathcal{R}_{\chi} [ \partial^{\alpha} ( P_{\ell} \Phi_{(\ell)} ) ] \, d\omega dvdu \\ := & \int_{\mathcal{N}_{u_2}} r^p ( L^{m+1} ( \chi P_{\ell} \Phi_{(\ell)} ) )^2 \, d\omega dv + J_1 + J_2 \\ & + \sum_{| \alpha | \leq m+1}\int_{u_1}^{u_2} \int_{\mathcal{N}_u}  ( L^{m+1} ( \chi P_{\ell} \Phi_{(\ell)} ) ) \cdot \mathcal{R}_{\chi} [ \partial^{\alpha} ( P_{\ell} \Phi_{(\ell)} ) ] \, d\omega dvdu ,
\end{split}
\end{equation}
where
$$ J_1 := \sum_{m_1 + m_2 = m , m_2 < m}\int_{u_1}^{u_2} \int_{\mathcal{N}_u}  O (r^{p-1-m_1} ) ( L^{m_2 + 1} (\chi P_{\ell} \Phi_{(\ell)} ) ) \cdot ( L^{m+1} ( \chi P_{\ell} \Phi_{(\ell)} ) )  \, d\omega dvdu , $$
$$ J_2 := \sum_{k=0}^{\ell-1} \sum_{m_1 + m_2 = m} \int_{u_1}^{u_2} \int_{\mathcal{N}_u}  O (r^{p-3-m_1} ) ( L^{m_2} ( \chi P_{\ell} \Phi_{(k)} ) ) \cdot ( L^{m+1} ( \chi P_{\ell} \Phi_{(\ell)} ) ) \, d\omega dvdu . $$
For the last term of the last equality we have that
$$ \sum_{| \alpha | \leq m+1}\int_{u_1}^{u_2} \int_{\mathcal{N}_u}  ( L^{m+1} ( \chi P_{\ell} \Phi_{(\ell)} ) ) \cdot \mathcal{R}_{\chi} [ \partial^{\alpha} ( P_{\ell} \Phi_{(\ell)} ) ] \, d\omega dvdu \leq C \sum_{l \leq m} \int_{\Sigma_{u_1}} J^T [ T^l P_{\ell} \psi ] \cdot n_{u_1}  \, d\mu_{\Sigma_{u_1}} , $$
by using Young's inequality and the Morawetz estimate \eqref{morawetz}. 

From here on we argue by induction on $m$. We note that the desired estimate is true in the base case of $m=0$ (due to the result of Proposition \ref{prop:fixedlrpest}). For $J_1$ we have that:
\begin{align*}
 J_1 \leq & \epsilon \int_{u_1}^{u_2} \int_{\mathcal{N}_u} r^{p-1}  ( L^{m+1} ( \chi P_{\ell} \Phi_{(\ell)} ) )^2 \, d\omega dv du \\ & +  \frac{C}{\epsilon} \sum_{m_1 + m_2 = m , m_2 < m} \int_{u_1}^{u_2} \int_{\mathcal{N}_u} r^{p-1-2m_1} ( L^{m_2 + 1} (\chi P_{\ell} \Phi_{(\ell)} ) )^2 \, d\omega dvdu . 
\end{align*}
The first term of the inequality given above can be absorbed by the left-hand side of \eqref{est:auxdr}, while for the second term we use the induction hypothesis and it can be bounded by:
\begin{align*}
\frac{C}{\epsilon} \sum_{m_1 + m_2 = m , m_2 < m} \int_{u_1}^{u_2} & \int_{\mathcal{N}_u} r^{p-1-2m_1} ( L^{m_2 + 1} (\chi P_{\ell} \Phi_{(\ell)} ) )^2 \, d\omega dvdu \\ \leq & C_{\epsilon} \sum_{s=0}^{m-1} \left( \int_{\mathcal{N}_{u_1}} r^{p-2s} ( L^{m+1-s} ( P_{\ell} \Phi_{(\ell)} ) )^2 \, d\omega dv + \int_{\Sigma_{u_1}} J^T [ T^s P_{\ell} \psi ] \cdot n_{u_1} \, d\mu_{\Sigma_{u_1}} \right).
\end{align*}
 For $J_2$ we have that:
\begin{align*}
J_2 \leq & \epsilon \int_{u_1}^{u_2} \int_{\mathcal{N}_u} r^{p-1}  ( L^{m+1} ( \chi P_{\ell} \Phi_{(\ell)} ) )^2 \, d\omega dv du \\ & + \frac{1}{\epsilon} \sum_{k=0}^{\ell-1} \sum_{m_1 + m_2 = m}\int_{u_1}^{u_2} \int_{\mathcal{N}_u} r^{p-5-2m_1} ( L^{m_2} ( \chi P_{\ell} \Phi_{(k)} ) )^2 \, d\omega dvdu .
\end{align*}
The first term can be absorbed by the left-hand side of \eqref{est:auxdr} while the second term can be bounded by:
$$ C_{\epsilon} \sum_{k=0}^{\ell-1} \sum_{s=0}^{m-1} \left( \int_{\mathcal{N}_{u_1}} r^{p-2s} ( L^{m+1-s} ( P_{\ell} \Phi_{(k)} ) )^2 \, d\omega dv + \int_{\Sigma_{u_1}} J^T [ T^s P_{\ell} \psi ] \cdot n_{u_1} \, d\mu_{\Sigma_{u_1}} \right) , $$
by using Proposition \ref{prop:fixedlrpdrestdrg}. Note that here we are assuming that $p < 2m+3$ while the aforementioned bound also holds for $p < 2m+4$.

\end{proof}

Assuming that $\ell$-th Newman--Penrose constant is vanishing the range of $p$ in the previous Proposition can be extended to $p \in (-4\ell+2m, 2m+5)$.

\begin{proposition}
\label{prop:fixedlrpdrestdr0}
Fix $\ell\in \N$ and consider a smooth solution $\psi$ to \eqref{eq:waveequation} satisfying \eqref{basic:as}. Then there exists $R>0$ sufficiently large, such that for $p\in (-4\ell+2m ,2m+4)$ and for all $u_0\leq u_1\leq u_2$:
\begin{equation*}
\begin{split}
\int_{\mathcal{N}_{u_2}}& r^p(L^{m+1} (P_{\ell}\Phi_{(\ell)}))^2\,d\omega dv+ \int_{u_1}^{u_2} \int_{\mathcal{N}_u} r^{p-1}(L^{m+1} (P_{\ell}\Phi_{(\ell)}))^2\,d\omega dv du\\
\leq&\: C\sum_{l \leq m} \int_{\mathcal{N}_{u_1}}r^{p-2l} (L^{m+1-l}(P_{\ell}\Phi_{(\ell)}))^2\,d\omega dv+ C\sum_{k\leq m}\int_{\Sigma_{u_1}} J^T[T^k P_{\ell}\psi]\cdot n_{u_1}\,d\mu_{\Sigma_{u_1}} ,
\end{split}
\end{equation*}
and for $p \in [2m+4, 2m+5)$ and any $\delta > 0$ small enough:
\begin{equation*}
\begin{split}
\int_{\mathcal{N}_{u_2}}& r^p(L^{m+1} (P_{\ell}\Phi_{(\ell)}))^2\,d\omega dv+ \int_{u_1}^{u_2} \int_{\mathcal{N}_u} r^{p-1}(L^{m+1} (P_{\ell}\Phi_{(\ell)}))^2\,d\omega dv du\\
\leq&\: C\sum_{l \leq m} \int_{\mathcal{N}_{u_1}}r^{p-2l} (L^{m+1-l}(P_{\ell}\Phi_{(\ell)}))^2\,d\omega dv+ C\sum_{k\leq m}\int_{\Sigma_{u_1}} J^T[T^k P_{\ell}\psi]\cdot n_{u_1}\,d\mu_{\Sigma_{u_1}} + C \frac{E_{aux, \ell ,k ,s}}{u_1^{1-\delta}},
\end{split}
\end{equation*}
where $C = C(D,R,m,\ell)$.
\end{proposition}

\begin{proof}
The difference with the proof of Proposition \ref{prop:fixedlrpdrestdr} is the treatment of $J_2$. We apply a $u$-weighted Young's inequality and for any $\delta$ small enough (to be determined later) we have that:
\begin{equation}\label{est:auxauxdr}
\begin{split}
J_2 \leq & \int_{u_1}^{u_2} \int_{\mathcal{N}_u} \frac{1}{u^{1+\delta}} r^{p}  ( L^{m+1} ( \chi P_{\ell} \Phi_{(\ell)} ) )^2 \, d\omega dv du \\ & + \sum_{k=0}^{\ell-1} \sum_{m_1 + m_2 = m}\int_{u_1}^{u_2} \int_{\mathcal{N}_u} u^{1+\delta} r^{p-6-2m_1} ( L^{m_2} ( \chi P_{\ell} \Phi_{(k)} ) )^2 \, d\omega dvdu \\ = & \int_{u_1}^{u_2} \int_{\mathcal{N}_u} \frac{1}{u^{1+\delta}} r^{p}  ( L^{m+1} ( \chi P_{\ell} \Phi_{(\ell)} ) )^2 \, d\omega dv du \\ & + \sum_{k=0}^{\ell-1} \sum_{s=0}^m \int_{u_1}^{u_2} \int_{\mathcal{N}_u} u^{1+\delta} r^{p-6-2m+2s} ( L^{s} ( \chi P_{\ell} \Phi_{(k)} ) )^2 \, d\omega dvdu .
\end{split}
\end{equation}
Now we use \eqref{est:auxdr5p} for $m \geq 0$ from Lemma \ref{lm:auxdecay1} of section \ref{energy}, and in the end we get that:
$$ \sum_{k=0}^{\ell-1} \sum_{s=0}^m \int_{u_1}^{u_2} \int_{\mathcal{N}_u} u^{1+\delta} r^{p-6-2m+2s} ( L^{s} ( \chi P_{\ell} \Phi_{(k)} ) )^2 \, d\omega dvdu \leq C \int_{u_1}^{u_2} \frac{E_{aux , \ell , m , s}}{u^{2-\delta_2}} \, du , $$
where $E_{aux, \ell , m , s}$ is given in Lemma \ref{lm:auxdecay1}.

\end{proof}
Finally for this section we show how to use the extended range of the previous two Propositions in order to obtain additional hierarchies for the solution commuted with $T$.
\begin{proposition}
\label{prop:fixedlrpdrestT}
Fix $\ell\in \N$ and consider a smooth solution $\psi$ to \eqref{eq:waveequation} satisfying \eqref{basic:as}. Then there exists $R>0$ sufficiently large, such that for $p\in (-4\ell + 2m,2m+3)$ and for all $u_0\leq u_1\leq u_2$:
\begin{equation*}
\begin{split}
 \int_{u_1}^{u_2} & \int_{\mathcal{N}_u} r^{p-1}(L^{m} (P_{\ell} ( T\Phi_{(\ell)} ) ))^2\,d\omega dv du
\\ \leq & C  \int_{\mathcal{N}_{u_1}}r^{p} (L^{m+1} (P_{\ell}\Phi_{(\ell)}))^2\,d\omega dv
 \\
\leq&\: C\sum_{l\leq m}\int_{\mathcal{N}_{u_1}}r^{p-2l} (L^{m+1-l}  (P_{\ell}\Phi_{(\ell)}))^2\,d\omega dv+ C\sum_{k\leq m}\int_{\Sigma_{u_1}} J^T[T^k P_{\ell}\psi]\cdot n_{u_1}\,d\mu_{\Sigma_{u_1}},
\end{split}
\end{equation*}
where $C=C(D,R,m,\ell)>0$.

Under the assumption that 
$$ I_{\ell} [ \psi ] = 0 , $$
the above estimate holds in the extended range $p \in (-4\ell+2m , 2m+5)$.

Moreover we have that
\begin{equation*}
\begin{split}
 \int_{u_1}^{u_2} & \int_{\mathcal{N}_u} r^{p-1}(L^{m} (P_{\geq \ell+1} ( T\Phi_{(\ell)} ) ))^2\,d\omega dv du
\\ \leq & C  \int_{\mathcal{N}_{u_1}}r^{p} (L^{m+1} (P_{\geq \ell+1}\Phi_{(\ell)}))^2\,d\omega dv
 \\
\leq&\: \bar{C}\sum_{l\leq m , |\alpha | \leq 1}\int_{\mathcal{N}_{u_1}}r^{p-2l} (L^{m+1-l} \Omega^{\alpha} (P_{\geq \ell + 1}\Phi_{(\ell)}))^2\,d\omega dv+ C\sum_{k\leq m}\int_{\Sigma_{u_1}} J^T[T^k P_{\geq \ell+1}\psi]\cdot n_{u_1}\,d\mu_{\Sigma_{u_1}},
\end{split}
\end{equation*}
for some $\bar{C}=\bar{C}(D,R,m,\ell)>0$, for $p \in (-4\ell+2m , 2m+2)$.
\end{proposition}

\begin{proof}
The two estimates can be proven in the same way. The proof follows by using equation \eqref{eq:maincommeq} after commuting it with $L^m$ (something that results in the appearance of the angular derivatives in the right-hand side -- for fixed $\ell$ see equation \eqref{eq:maineqNpquantfixedlv3}) and then making use of Propositions \ref{prop:fixedlrpdrestdrg} and \ref{prop:fixedlrpdrestdr} and that $T = L + \underline{L}$. The proof follows the same lines as the proof of Proposition 4.7 of \cite{paper1}.

\end{proof}

\section{Energy decay estimates}\label{energy}
In this section we will use the $r^p$-weighted estimates of the previous Section to show decay for certain energy-type quantities. 

\subsection{Auxiliary energy decay estimates}
First we state and prove two auxiliary results that were needed in the proofs of Propositions \ref{prop:fixedlrpest2} and \ref{prop:fixedlrpdrestdr0}. 

\begin{lemma}
\label{lm:auxdecay}
Fix $\ell\in \N$ and consider a solution $\psi$ to \eqref{eq:waveequation}. Assume that there exists a constant $C_0>0$ such that
\begin{equation*}
\left( \int_{\s^2}(P_{\ell}\Phi_{(k)})^2\,d\omega\right)(u_0,r)\leq C_0,
\end{equation*}
for $0\leq k\leq \ell$.

Then there exists $C=C(D,R,\ell ,\epsilon)$ such that for all $0\leq k\leq \ell$
\begin{align}
\label{eq:auxnullenergy1}
\int_{\mathcal{N}_u} (L(\chi P_{\ell}\Phi_{(k)}))^2\,d\omega dv\leq Cu^{-4-2 (\ell - k ) + \epsilon}E_{\textnormal{aux},\ell},\\
\label{eq:auxnullenergy2}
\int_{\mathcal{N}_u} r^2(L(\chi P_{\ell}\Phi_{(k)}))^2\,d\omega dv\leq Cu^{-2-2 ( \ell - k ) + \epsilon}E_{\textnormal{aux},\ell},
\end{align}
with
\begin{equation*}
E_{\textnormal{aux},\ell}=\sum_{k=0}^{\ell} \int_{\mathcal{N}_0} r^{4-\epsilon}(L(P_{\ell}\Phi_{(k)}))^2\,d\omega dv+\sum_{j=0}^{\ell}E_1[T^jP_{\ell}\psi],
\end{equation*}
and $E_1 [ \psi ]$ as in Proposition \ref{prop:fixedlrpest2}.

In the region $\{r\geq R\}$ we can moreover estimate
\begin{equation*}
\int_{\s^2}\chi^2(P_{\ell}\Phi_{(k)})^2(u,r,\theta,\varphi)\,d\omega \leq Cu^{-3-2( \ell - k ) +\epsilon}E_{\textnormal{aux},\ell},
\end{equation*}
for all $0\leq k\leq \ell$, with $C=C(D, R , \ell ,\epsilon)>0$.
\end{lemma}
\begin{proof}
From \cite{paper1} it follows that there exists a constant $C=C(D,\epsilon,R)>0$ such that for any linear wave $\psi_{\geq 1}$ that is supported on angular frequencies $\geq 1$, we have the following:
\begin{equation}
\label{eq:edecayn1}
\int_{\Sigma_u} J^N[\psi_{\geq 1}]\cdot n_u\,d\mu_{\Sigma_u}\leq C u^{-4+\epsilon} E_1[\psi_{\geq 1}],
\end{equation}
where $E_1 [ \psi_{\geq 1}]$ is as in Proposition \ref{prop:fixedlrpest2}.

Let $\{u_j\}$ be a dyadic sequence. Then we can apply the mean-value theorem to the estimate of Proposition \ref{prop:fixedlrpest} with $p=4-\epsilon$ to obtain
\begin{equation*}
\int_{\mathcal{N}_{u_j'}} r^{3-\epsilon}(L\Phi_{(\ell)})^2\,d\omega dv\leq C u_{j-1}^{-1} \int_{\mathcal{N}_{u_0}} r^{4-\epsilon}(L\Phi_{(\ell)})^2\,d\omega dv+u_{j-1}^{-4+\epsilon}\sum_{j=0}^{\ell} E_1[T^j\psi],
\end{equation*}
for some $u_{j-1}\leq u_{j}' \leq u_{j-1}$, where we additionally applied \eqref{eq:edecayn1}. By the dyadicity of $\{u_j\}$, $\{u_{2j+1}'\}$ is a dyadic subsequence of $\{u_j'\}$. We can therefore apply the mean-value theorem successively together with the estimate of Proposition \ref{prop:fixedlrpest} with $p=3-\epsilon$,\ldots, $p=1-\epsilon$ and use the already established decay in the previous step to obtain:
\begin{align*}
\int_{\mathcal{N}_{\tilde{u}_j}} (L(\chi \Phi_{(\ell)}))^2\,d\omega dv\leq C{\tilde{u}_j}^{-4}E_{\textnormal{aux},\ell},\\
\int_{\mathcal{N}_{\tilde{u}_j}} r^{1-\epsilon}(L(\chi \Phi_{(\ell)}))^2\,d\omega dv\leq C{\tilde{u}_j}^{-3}E_{\textnormal{aux},\ell},\\
\int_{\mathcal{N}_{\tilde{u}_j}} r^{2-\epsilon}(L(\chi \Phi_{(\ell)}))^2\,d\omega dv\leq C{\tilde{u}_j}^{-2}E_{\textnormal{aux},\ell},\\
\int_{\mathcal{N}_{\tilde{u}_j}} r^{3-\epsilon}(L(\chi \Phi_{(\ell)}))^2\,d\omega dv\leq C{\tilde{u}_j}^{-1}E_{\textnormal{aux},\ell},
\end{align*}
for a potentially different dyadic sequence $\{ \tilde{u}_j \}$. We can use the interpolation estimates from Lemma \ref{ineq:interpol} to obtain
\begin{align*}
\int_{\mathcal{N}_{\tilde{u}_j}} (L(\chi \Phi_{(\ell)}))^2\,d\omega dv\leq C{\tilde{u}_j}^{-4+\epsilon}E_{\textnormal{aux},\ell},\\
\int_{\mathcal{N}_{\tilde{u}_j}} r^{2}(L(\chi \Phi_{(\ell)}))^2\,d\omega dv\leq C{\tilde{u}_j}^{-2+\epsilon}E_{\textnormal{aux},\ell}.
\end{align*}
By an application of the estimates of Proposition \ref{prop:fixedlrpest} with $p=0$ and $p=2$ respectively, we obtain
\begin{align*}
\int_{\mathcal{N}_u} (L(\chi \Phi_{(\ell)}))^2\,d\omega dv\leq Cu^{-4+\epsilon}E_{\textnormal{aux},\ell},\\
\int_{\mathcal{N}_u} r^2(L(\chi \Phi_{(\ell)}))^2\,d\omega dv\leq Cu^{-2+\epsilon}E_{\textnormal{aux},\ell},
\end{align*}
for all $u\geq u_0$, which gives us, so \eqref{eq:auxnullenergy1} and \eqref{eq:auxnullenergy2} hold for $k=\ell$.

We can moreover apply Hardy's inequality to estimate
\begin{align*}
\int_{\mathcal{N}_u} (L(\chi \Phi_{(k)}))^2\,d\omega dv\leq&\: C\int_{\mathcal{N}_u} r^{-2}(L(\chi \Phi_{(k+1)}))^2\,d\omega dv+ C\int_{\mathcal{N}_u}(\chi')^2 \Phi_{(k)}^2\,d\omega dv,\\
\int_{\mathcal{N}_u} r^2(L(\chi \Phi_{(k)}))^2\,d\omega dv\leq&\: C\int_{\mathcal{N}_u} (L(\chi \Phi_{(k+1)}))^2\,d\omega dv +C\int_{\mathcal{N}_u}(\chi')^2 \Phi_{(k)}^2\,d\omega dv.
\end{align*}
Hence, we can inductively show that \eqref{eq:auxnullenergy1} and \eqref{eq:auxnullenergy2} hold also for $0\leq k<\ell$.

We apply the fundamental theorem of calculus, Cauchy--Schwarz and Hardy's inequality \eqref{hardy1} to estimate
\begin{equation*}
\begin{split}
\int_{\s^2}\chi\Phi_{(k)}^2(u,r,\theta,\varphi)\,d\omega=&\:2\int_{\mathcal{N}_u} \chi\Phi_{(k)} \cdot L(\chi \Phi_{(k)})\,d\omega dv\\
\leq &\:2\sqrt{\int_{\mathcal{N}_u}r^{-2}\chi^2\Phi_{(k)}^2\,d\omega dv}\cdot \sqrt{\int_{\mathcal{N}_u} r^2(L(\chi \Phi_{(k)}))^2\,d\omega dv}\\
\leq &\:C\sqrt{\int_{\mathcal{N}_u}(L(\chi \Phi_{(k)}))^2\,d\omega dv}\cdot \sqrt{\int_{\mathcal{N}_u} r^2(L(\chi \Phi_{(k)}))^2\,d\omega dv}\\
\leq &\: u^{-3+(\ell - k ) +\epsilon}E _{\textnormal{aux},\ell}.
\end{split}
\end{equation*}
\end{proof}
In a similar way using the results of Proposition \ref{prop:fixedlrpdrestdr0} we can  show the following. 

\begin{lemma}
\label{lm:auxdecay1}
Fix $\ell , m \in \N$ and consider a solution $\psi$ to \eqref{eq:waveequation}. Assume that there exists a constant $C_0>0$ such that
\begin{equation*}
\sum_{l \leq m ,  k \leq \ell} \left( \int_{\s^2}(L^l (  P_{\ell}\Phi_{(k)}) )^2\,d\omega\right)(u_0,r)\leq C_0 .
\end{equation*}

Then we have the following energy estimates:
\begin{equation}\label{est:auxdr51}
\int_{\mathcal{N}_u}  ( L^{m+1} ( P_{\ell} \Phi_{(k)} ) )^2 \, d\omega dv \leq C \frac{E_{aux , \ell , k , l}}{u^{4+2(\ell-k) -\delta_1}} , 
\end{equation}
and
\begin{equation}\label{est:auxdr52}
\int_{\mathcal{N}_u}  r^2 ( L^{m+1} ( P_{\ell} \Phi_{(k)} ) )^2 \, d\omega dv \leq C \frac{E_{aux , \ell , k , l}}{u^{2+2(\ell-k) -\delta_1}} , 
\end{equation}
for any $\delta_1 > 0$ small enough and $k \in \{0 , \dots , \ell\}$, and the following pointwise estimates:
\begin{equation}\label{est:auxdr5p}
\int_{\mathbb{S}^2} ( L^m ( P_{\ell} \Phi_{(k)} )^2 \, d\omega \leq C \frac{E_{aux , \ell , k , 0}}{u^{3+ 2( \ell-k ) - \delta_1}} ,   
\end{equation}
for any $\delta_1 > 0$ small enough and $k \in \{0, \dots , \ell-1\}$ in the region $r \geq R$, where $C = C (D,R,\ell,\delta_1)$ and where
\begin{align*}
E_{aux , \ell , k , l} \doteq & \sum_{j_1 + j_2 \leq m+\ell} \int_{\mathcal{N}_u} r^{2m+3-2j-\delta} ( L^{m+1-j_1} T^{j_2} ( P_{\ell} \Phi_{(\ell)} ) )^2 \, d\omega dv \\ &+  \sum_{j_1 + j_2 \leq m+\ell}\sum_{s=k}^{\ell-1} \int_{\mathcal{N}_u} r^{2m+2-2j} ( L^{m+1-j_1} T^{j_2} ( P_{\ell} \Phi_{(s)} ) )^2 \, d\omega dv + \sum_{l \leq2 \ell+m} \int_{\Sigma_{u_0}} J^N [ T^l P_{\ell} \psi ] \cdot n_{u_0} \, d\mu_{\Sigma_{u_0}} .
\end{align*}

\end{lemma}

\subsection{Almost-sharp energy decay estimates}
In the remaining part of the section we now show some decay estimates for the $r^p$-weighted energies.
\begin{lemma}\label{dec:rp1}
Fix $\ell\in \N$ and consider a smooth solution $\psi$ to \eqref{eq:waveequation} satisfying \eqref{basic:as}. Then we have that for every $u \geq u_0$ that:
\begin{equation}
\int_{\mathcal{N}_u} r^p ( L ( P_{\ell} \Phi_{(\ell)} ) )^2 \, d\omega dv \leq C \frac{E_{aux , I_{\ell} \neq 0 , \ell}}{u^{3-p-\delta}} , 
\end{equation}
for $p \in (0,3)$, any $\delta \in (0,1)$ small enough, $C \doteq C ( D, R,p,\delta,\ell)$, and
\begin{align*}
E_{aux, I_{\ell} \neq 0 , m} \doteq & \int_{\mathcal{N}_u} r^{3-\delta} ( L ( P_{\ell} \Phi_{(m)} ) )^2 \, d\omega dv \\ &+ \sum_{s=0}^{m-1} \int_{\mathcal{N}_u} r^2 ( L ( P_{\ell} \Phi_{(s)} ) )^2 \, d\omega dv + \sum_{l \leq m} \int_{\Sigma_{u_0}} J^T [ T^l P_{\ell} \psi ] \cdot n_{u_0} \, d\mu_{\Sigma_{u_0}} . 
\end{align*}

If  we additionally assume that
$$ I_{\ell} [ \psi ] = 0 , $$
we have for every $u \geq u_0$ that:
\begin{equation}
\int_{\mathcal{N}_u} r^p ( L ( P_{\ell} \Phi_{(\ell)} ) )^2 \, d\omega dv \leq C \frac{E_{aux , I_{\ell} = 0  , n}}{u^{5-p-\delta}} , 
\end{equation}
for $p \in (0,5)$, any $\delta \in (0,1)$ small enough, $C \doteq C(D,R,p,\delta,\ell)$, and
\begin{align*}
E_{aux, I_{\ell} = 0 , m} \doteq & \int_{\mathcal{N}_u} r^{5-\delta} ( L ( P_{\ell} \Phi_{(m)} ) )^2 \, d\omega dv \\ & + \sum_{s=0}^{m-1} \int_{\mathcal{N}_u} r^2 ( L ( P_{\ell} \Phi_{(s)} ) )^2 \, d\omega dv+ \sum_{l \leq m} \int_{\Sigma_{u_0}} J^T [ T^l P_{\ell} \psi ] \cdot n_{u_0} \, d\mu_{\Sigma_{u_0}} . 
\end{align*}

\end{lemma}
The proof of the above Lemma is quite standard and can be done in the same way as the proof of Lemma \ref{lm:auxdecay} using the results of Propositions \ref{prop:fixedlrpest} and \ref{prop:fixedlrpest2}.

\begin{lemma}\label{dec:rp2}
Fix $\ell\in \N$ and consider a smooth solution $\psi$ to \eqref{eq:waveequation} satisfying \eqref{basic:as}. Then we have that for $k \in \{ 0 , \dots , \ell-1\}$ and for every $u \geq u_0$ that:
\begin{equation}
\int_{\mathcal{N}_u} r^p ( L ( P_{\ell} \Phi_{(k)} ) )^2 \, d\omega dv \leq C \frac{E_{aux , I_{\ell} \neq 0 , n}}{u^{3+2(\ell-k) -p-\delta}} , 
\end{equation}
for $p \in (0,2)$, any $\delta \in (0,1)$ small enough, and $C \doteq C ( D, R,p,\delta,\ell)$.

If additionally we assume that
$$ I_{\ell} [ \psi ] = 0 , $$
we have for $k \in \{ 0 , \dots , \ell-1\}$ and for every $u \geq u_0$ that:
\begin{equation}
\int_{\mathcal{N}_u} r^p ( L ( P_{\ell} \Phi_{(k)} ) )^2 \, d\omega dv \leq C \frac{E_{aux , I_{\ell} =0  , \ell}}{u^{5+ 2 (\ell-k) -p-\delta}} , 
\end{equation}
for $p \in (0,2)$, any $\delta \in (0,1)$ small enough, and $C \doteq C(D,R,p,\delta,\ell)$.

\end{lemma}
The proof of the above Lemma uses the previous one as well as its proof, along with Hardy's inequality \eqref{hardy1}. For details see Section 7 of \cite{paper1}. 

By making use of the results of Lemmas \ref{dec:rp1} and \ref{dec:rp2} (for $P_{\ell} \Phi_{(\ell)}$ and $P_{\ell+1} \Phi_{(\ell)}$ respectively, for the latter assuming that $I_{\ell+1} [ \psi ] \neq 0$) and by making use of the $r^p$-weighted estimates of Proposition \ref{prop:fixedlrpdrestdrg} (along with the same tools as in the proof of Lemma \ref{lm:auxdecay}) we have the following result:

\begin{lemma}
Fix $\ell\in \N$ and consider a smooth solution $\psi$ to \eqref{eq:waveequation} satisfying \eqref{basic:as}. Then for every $u \geq u_0$ we have that:
$$ \int_{\mathcal{N}_u} r^p ( L ( P_{\geq \ell+2} \Phi_{(\ell)} ) )^2 \, d\omega dv \leq C \frac{E_{aux, \ell+2}}{u^{6-p-\delta}} , $$
for $p \in (0,2)$, for some constant $C$, for any $\delta > 0$ and for an initial norm
\begin{align*}
E_{aux , \ell+2} \doteq & \int_{\mathcal{N}_{u_0}} r^2 ( L ( P_{\geq \ell+2} \Phi_{(\ell)} ) )^2 \, d\omega dv  + \int_{\mathcal{N}_{u_0}} r^2 ( L ( P_{\geq \ell+2} \Phi_{(\ell+1)} ) )^2 \, d\omega dv \\ & + \int_{\mathcal{N}_{u_0}} r^2 ( L ( P_{\geq \ell+2} \Phi_{(\ell+2)} ) )^2 \, d\omega dv , 
\end{align*}

\end{lemma}

Using the last two results we get the following:
\begin{corollary}\label{rem:n+2}
Fix $\ell\in \N$ and consider a smooth solution $\psi$ to \eqref{eq:waveequation} satisfying \eqref{basic:as}. Then we have that
\begin{equation}\label{dec:engeqn0}
\int_{\mathcal{N}_u} r^p ( L ( P_{\geq \ell} \Phi_{(\ell)} ) )^2 \, d\omega dv \leq C \frac{E_{aux-decomp-n0}}{u^{3-p-\delta}} ,
\end{equation}
for $p \in (0,2)$, $\delta > 0$, $C = C ( D,R,p,\delta,\ell)$, and 
$$ E_{aux-decomp-n0} \doteq E_{aux , I_n \neq 0 , n} + E_{aux , \ell+2} + \sum_{l \leq n+2} \int_{\Sigma_{u_0}} J^N [ T^l P_{\geq \ell} \psi ] \cdot n_{u_0} \, d\mu_{\Sigma_{u_0}} . $$

If additionally we assume that
$$ I_{\ell} [ \psi ] = 0 , $$
and we have that
\begin{equation}\label{dec:engeq0}
\int_{\mathcal{N}_u} r^p ( L ( P_{\geq \ell} \Phi_{(\ell)} ) )^2 \, d\omega dv \leq C \frac{E_{aux-decomp-n0}}{u^{5-p-\delta}} ,
\end{equation}
for $p \in (0,2)$, $\delta > 0$, $C = C(D,R,p,\delta,\ell)$, and 
$$ E_{aux-decomp-0} \doteq E_{aux , I_{\ell} = 0 , \ell} + E_{aux , \ell+2} + \sum_{l \leq \ell+2} \int_{\Sigma_{u_0}} J^N [ T^l P_{\geq \ell} \psi ] \cdot n_{u_0} \, d\mu_{\Sigma_{u_0}} .$$

\end{corollary}

Using the previous two results and the red-shift estimate \eqref{redshift:as}, we can obtain the following energy estimates (which can also be viewed as $r^p$-weighted estimates for $p=0$):
\begin{proposition}\label{prop:endecay}
Fix $\ell\in \N$ and consider a smooth solution $\psi$ to \eqref{eq:waveequation} satisfying \eqref{basic:as}. Then we have that for every $\tau \geq \tau_0$ that:
\begin{equation}
\int_{\Sigma_{\tau}} J^N [ P_{\ell} \psi ] \cdot n_{\tau} \, d\mu_{\Sigma_{\tau}} \leq C \frac{E_{aux , I_{\ell} \neq 0 , \ell}}{\tau^{2\ell+3-\delta}} , 
\end{equation}
for any $\delta \in (0,1)$ small enough, $C \doteq C ( D, R,\delta,\ell)$, and $ E_{aux, I_{\ell} \neq 0 , \ell}$ as defined in Lemma \ref{dec:rp2}.

If additionally we assume that
$$ I_{\ell} [ \psi ] = 0 , $$
we have for $k \in \{ 0 , \dots , \ell-1\}$ and for every $u \geq u_0$ that:
\begin{equation}
\int_{\Sigma_u} J^N [ P_{\ell} \psi ] \cdot n_u \, d\mu_{\Sigma_u} \leq C \frac{E_{aux , I_{\ell} =0  , \ell}}{u^{2\ell+5-\delta}} , 
\end{equation}
for any $\delta \in (0,1)$ small enough, $C \doteq C(D,R,\delta,\ell)$, and $E_{aux, I_{\ell} = 0 , \ell}$ as defined in Lemma \ref{dec:rp2}.
\end{proposition}
Arguing as in Corollary \ref{rem:n+2}, we get the following results.
\begin{corollary}\label{rem:enn+2}
Fix $\ell\in \N$ and consider a smooth solution $\psi$ to \eqref{eq:waveequation} satisfying \eqref{basic:as}. Then we have that:
\begin{equation}\label{dec:enn+2n0}
\int_{\Sigma_u} J^N [ P_{\geq \ell} \psi ] \cdot n_u \, d\mu_{\Sigma_u} \leq C \frac{E_{aux-decomp-n0}}{u^{2\ell + 3-\delta}} ,
\end{equation}
for all $u \geq u_0$, for any $\delta > 0$, for some $C = C ( D, R)$ and where $E_{aux-decomp-n0}$ was determined in Remark \ref{rem:n+2}.

If  we additionally assume that
$$ I_{\ell} [\psi ] = 0 , $$
we have that 
\begin{equation}\label{dec:enn+20}
\int_{\Sigma_u} J^N [ P_{\geq \ell} \psi ] \cdot n_u \, d\mu_{\Sigma_u} \leq C \frac{E_{aux-decomp-0}}{u^{2\ell + 5-\delta}} ,
\end{equation}
for all $u \geq u_0$, for any $\delta > 0$, for some $C = C( D, R)$ and where $E_{aux-decomp-0}$ was determined in Remark \ref{rem:n+2}.
\end{corollary}

Using Propositions \ref{prop:fixedlrpdrestdrg}, \ref{prop:fixedlrpdrestdr}, \ref{prop:fixedlrpdrestdr0}, and \ref{prop:fixedlrpdrestT} we get the following decay estimates for the $r^p$-weighted quantities and for the energies after commuting with $T$ (for details see the analogous estimates from section 7 of \cite{paper1}):
\begin{proposition}\label{dec:enT}
Fix $\ell\in \N$ and consider a smooth solution $\psi$ to \eqref{eq:waveequation} satisfying \eqref{basic:as}. Then we have for every $u \geq u_0$ that:
\begin{equation}
\int_{\mathcal{N}_u} r^p ( L ( P_{\ell} T^m \Phi_{(\ell)} ) )^2 \, d\omega dv \leq C \frac{E_{aux-T , I_{\ell} \neq 0 , \ell , m}}{u^{3+2m-p-\delta}} , 
\end{equation}
for $p \in (0,3)$, any $\delta \in (0,1)$ small enough, $C \doteq C ( D, R,m,p,\delta,\ell)$, and
\begin{align*}
E_{aux-T, I_{\ell} \neq 0 , q , m} \doteq & \sum_{j=0}^m \sum_{k \leq \ell} \int_{\mathcal{N}_u} r^{2m+3-2j-\delta} ( L^{m+1-j} ( P_{\ell} T^k  \Phi_{(q)} ) )^2 \, d\omega dv \\ &+ \sum_{s=0}^{q-1}  \sum_{k \leq \ell}\int_{\mathcal{N}_u} r^{2m+2-2j} ( L^{m+1-j} ( P_{\ell} T^k \Phi_{(s)} ) )^2 \, d\omega dv \\ & + \sum_{l \leq q+m+\ell+1} \int_{\Sigma_{u_0}} J^N [ T^l P_{\ell} \psi ] \cdot n_{u_0} \, d\mu_{\Sigma_{u_0}} . 
\end{align*} 

For $k \in \{ 0 , \dots , \ell\}$, $m \in \mathbb{N}$, and for every $u \geq u_0$ we have that:
\begin{equation}
\int_{\mathcal{N}_u} r^p ( L ( P_{\ell} T^m \Phi_{(k)} ) )^2 \, d\omega dv \leq C \frac{E_{aux-T , I_{\ell} \neq 0 , \ell  , m}}{u^{3+2(\ell-k)+2m -p-\delta}} , 
\end{equation}
for $p \in (0,2)$, any $\delta \in (0,1)$ small enough, and $C \doteq C ( D, R,m,p,\delta,\ell)$.

We also have that for every $u \geq u_0$:
\begin{equation}
\int_{\Sigma_u} J^N [ P_{\ell} T^m \psi ] \cdot n_u \, d\mu_{\Sigma_u} \leq C \frac{E_{aux-T , I_{\ell} \neq 0 , \ell}}{u^{2\ell+2m+3-\delta}} , 
\end{equation}
for any $\delta \in (0,1)$ small enough, $C \doteq C ( D, R,m,\delta,\ell)$.

In the case that we additionally assume
$$ I_{\ell} [ \psi ] = 0 , $$
for every $u \geq u_0$ we have that:
\begin{equation}
\int_{\mathcal{N}_u} r^p ( L ( P_{\ell} T^m \Phi_{(\ell)} ) )^2 \, d\omega dv \leq C \frac{E_{aux-T , I_{\ell} = 0 , \ell , m}}{u^{5+2m-p-\delta}} , 
\end{equation}
for $p \in (0,5)$, any $\delta \in (0,1)$ small enough, $C \doteq C ( D, R,m,\delta,\ell)$, and
\begin{align*}
E_{aux-T, I_{\ell} = 0 , q  , m} \doteq & \sum_{j=0}^m \sum_{k \leq \ell} \int_{\mathcal{N}_u} r^{2m+5-2j-\delta} ( L^{m+1-j} ( P_{\ell} T^k \Phi_{(q)} ) )^2 \, d\omega dv \\ &+ \sum_{s=0}^{q-1} \sum_{k \leq \ell} \int_{\mathcal{N}_u} r^{2m+2-2j} ( L^{m+1-j} ( P_{\ell} T^k \Phi_{(s)} ) )^2 \, d\omega dv \\ & + \sum_{l \leq q+m+\ell+1} \int_{\Sigma_{u_0}} J^N [ T^l P_{\ell} \psi ] \cdot n_{u_0} \, d\mu_{\Sigma_{u_0}} . 
\end{align*}
for $m \in \mathbb{N}$, and for every $u \geq u_0$ we have that:
\begin{equation}
\int_{\mathcal{N}_u} r^p ( L ( P_{\ell} T^m \Phi_{(k)} ) )^2 \, d\omega dv \leq C \frac{E_{aux-T , I_{\ell} = 0 , \ell , k , m}}{u^{5+2(\ell-k)+2m -p-\delta}} , 
\end{equation}
for $p \in (0,2)$, any $\delta \in (0,1)$ small enough, $C \doteq C ( D, R)$, and we also have that for every $u \geq u_0$:
\begin{equation}
\int_{\Sigma_u} J^N [ P_{\ell} T^m \psi ] \cdot n_u \, d\mu_{\Sigma_u} \leq C \frac{E_{aux-T , I_{\ell} = 0 , \ell}}{u^{2\ell+2m+5-\delta}} , 
\end{equation}
for any $\delta \in (0,1)$ small enough, $C \doteq C ( D, R,m,\delta,\ell)$.
\end{proposition}

By writing $P_{\geq \ell} T^m \Phi_{(\ell )}$ as a sum of $P_{\ell} T^m \Phi_{(\ell )}$, $P_{\ell+1} T^m \Phi_{(\ell )}$ and $P_{\geq \ell+2} T^m \Phi_{(\ell )}$, and applying the previous results to the first two terms, and working similarly for the last term using the results of the previous section, we can obtain almost sharp energy decay estimates for $P_{\geq \ell} T^m \Phi_{(\ell)}$.
\begin{corollary}\label{rem:n+2t}
Fix $\ell\in \N$ and consider a smooth solution $\psi$ to \eqref{eq:waveequation} satisfying \eqref{basic:as}. Then we have that
\begin{equation}\label{dec:enn+2tn0}
\int_{\mathcal{N}_u} r^p ( L ( P_{\geq \ell} T^m \Phi_{(\ell)} ) )^2 \,d\omega dv \leq C \frac{E_{aux-decomp-n0-m}}{u^{3+2m-p-\delta}} ,
\end{equation}
for all $u\geq u_0$, for any $\delta > 0$, where $C = C(D,R)$, and where
\begin{align*}
E_{aux-decomp-n0-m} \doteq & \sum_{k \leq \ell} \Big( \int_{\mathcal{N}_{u_0}} r^{3-\delta} ( L ( P_{\ell} T^{m+k} \Phi_{(\ell)} ) )^2 \,d\omega dv + \int_{\mathcal{N}_{u_0}} r^2 ( L ( P_{\ell+1} T^{m+k} \Phi_{(\ell)} ) )^2 \,d\omega dv  \\ & + \int_{\mathcal{N}_{u_0}} r^2 ( L ( P_{\geq \ell+2} T^{m+k} \Phi_{(\ell)} ) )^2 \,d\omega dv \\ & +  \sum_{|\alpha| \leq m , j \leq m+\ell} \int_{\mathcal{N}_{u_0}} r^{2-\delta} | \slashed{\nabla}^{\alpha}_{\mathbb{S}^2} ( L ( P_{\geq \ell+2} T^j \Phi_{(\ell)} ) ) |^2 \,d\omega dv \Big) \\ & + \sum_{l \leq m+\ell + 1} \int_{\Sigma_{u_0}}J^N [ P_{\geq \ell } T^l \psi ] \cdot n_{u_0} \, d\mu_{\Sigma_{u_0}} .
\end{align*}

If we additionally assume that 
$$ I_{\ell} [\psi ] = 0 , $$
then we have that
\begin{equation}\label{dec:enn+2t0}
\int_{\mathcal{N}_u} r^p ( L ( P_{\geq \ell} T^m \Phi_{(\ell)} ) )^2 \,d\omega dv \leq C \frac{E_{aux-decomp-0-m}}{u^{5+2m-p-\delta}} ,
\end{equation}
for all $u\geq u_0$, for any $\delta > 0$, where $C = C(D,R)$, and where
\begin{align*}
E_{aux-decomp-0-m} \doteq &  \sum_{k \leq \ell} \Big(\int_{\mathcal{N}_{u_0}} r^{5-\delta} ( L ( P_{\ell} T^m \Phi_{(\ell)} ) )^2 \,d\omega dv + \int_{\mathcal{N}_{u_0}} r^2 ( L ( P_{\ell+1} T^m \Phi_{(\ell)} ) )^2 \,d\omega dv  \\ & + \int_{\mathcal{N}_{u_0}} r^2 ( L ( P_{\geq \ell+2} T^m \Phi_{(\ell)} ) )^2 \,d\omega dv\\ &+ \sum_{|\alpha| \leq m, j \leq m+\ell} \int_{\mathcal{N}_{u_0}} r^{2-\delta} | \slashed{\nabla}^{\alpha}_{\mathbb{S}^2} ( L ( P_{\geq \ell+2} T^j \Phi_{(\ell)} ) ) |^2 \,d\omega dv\Big) \\ &+ \sum_{l \leq m+\ell + 1} \int_{\Sigma_{u_0}}J^N [ P_{\geq \ell } T^l \psi ] \cdot n_{u_0} \, d\mu_{\Sigma_{u_0}} .
\end{align*}
\end{corollary}

We note that the angular derivatives present in the norms of the aforementioned Corollary appear due to the use of equation \eqref{eq:maincommeq} in Proposition \ref{prop:fixedlrpdrestT}.

\section{Novel hierarchies of elliptic estimates}\label{elliptic}

In this section we will present bound $\partial_r$ derivatives of a linear wave supported on angular frequencies $\geq \ell$ for some $\ell \geq 1$, by $T$ derivatives, via $r^{-k}$-weighted energy estimates, where the range of $k$ depends on $\ell$. These estimates will allow us to obtain \emph{almost sharp upper bounds on fixed $r$ hypersurfaces} (something that will be done in the following sections). 

Our estimates have the form of a weighted hierarchy, where the weights depend on the lowest angular frequencies on which our linear wave is supported. Note that here we make the choice $\frac{2}{D} - h = \frac{c}{r^{1+\eta}} + O (r^{-2} )$ for some $c \neq 0$ and $\eta > 0$ (which implies that we work on the $\mathcal{S}$ hypersurfaces as noted before), for $\partial_{\rho} = \partial_r + h T$.

\begin{theorem}\label{thm:el}
Assume that
\begin{align*}
\sum_{k\leq 1}\lim_{\rho\to \infty}r T^k\psi<&\infty,\\
\lim_{\rho\to \infty}r^2\partial_r\psi=&0.
\end{align*}
Let $\psi$ be a solution to \eqref{eq:waveequation} and let $\ell\geq 1$. Then we can estimate
\begin{equation}
\label{eq:ellipticpsiell}
\begin{split}
\int_{r_{+}}^{\infty}&\int_{\s^2} \left[ r^{-2-k}  (\partial_{\rho}(Dr^2\partial_{\rho}\psi_{\geq \ell} ))^2+Dr^{-k}|\snabla_{\s^2} \partial_{\rho}\psi_{\geq \ell} |^2+r^{-2-k}(\slashed{\Delta}_{\s^2}\psi_{\geq \ell} )^2 \right] \,d\omega d\rho\\
\leq&\: C(D,\ell,k)\int_{r_{+}}^{\infty}\int_{\s^2} \left[ r^{2-k} (\partial_{\rho} T \psi_{\geq \ell} )^2+r^{-k-2\eta}(T^2\psi_{\geq \ell} )^2+r^{-k}(T\psi_{\geq \ell} )^2 \right] \,d\omega d\rho,
\end{split}
\end{equation}
for all 
\begin{equation*}
-3< k<2\ell-1,
\end{equation*}
and any $\eta > 0$.
\end{theorem}
\begin{proof}
Recall that in $(v,r,\theta , \varphi)$ ingoing Eddington--Finkelstein coordinates, and with $\partial_{\rho} = \partial_r + h T$, we have that the wave equation \eqref{eq:waveequation} becomes:
\begin{equation}
\label{eq:inhomelliptic}
\partial_{\rho}(Dr^2 \partial_{\rho} \psi)+\slashed{\Delta}_{\s^2}\psi= r^2F_T,
\end{equation}
where
\begin{equation*}
F_T:=(2+O(r^{-1-\eta}) ) \partial_{\rho}T\psi+O(r^{-1-\eta})T^2\psi+(2r^{-1}+O(r^{-2-\eta}))T\psi.
\end{equation*}
Therefore, by squaring and integrating both sides of \eqref{eq:inhomelliptic} and multiplying by $r^{-2-k}$, we arrive at
\begin{equation}
\label{eq:waveeqestlrplus}
\begin{split}
\int_{r_{+}}^{\infty}\int_{\s^2} &r^{-2-k}(\partial_{\rho}(Dr^2\partial_{\rho}\psi))^2r^2+r^{-2-k}(\slashed{\Delta}_{\s^2}\psi)^2+2r^{-2-k}\partial_{\rho}(Dr^2\partial_{\rho}\psi)\slashed{\Delta}_{\s^2}\psi\,d\omega d\rho\\
\leq&\: C\int_{r_{+}}^{\infty}\int_{\s^2}r^{2-k}F_T^2\,d\omega d\rho.
\end{split}
\end{equation}

We first consider the mixed derivative term on the left-hand side of \eqref{eq:waveeqestlrplus}. We integrate over $\s^2$ and integrate by parts in $\rho$ and the angular variables.
\begin{equation}
\begin{split}
\label{eq:ibpangularlrplus}
\int_{r_{+}}^{\infty} \int_{\s^2}2r^{-2-k}\partial_{\rho}(Dr^2\partial_{\rho}\psi)\slashed{\Delta}_{\s^2}\psi\,d\omega d\rho=&\:\int_{r_{+}}^{\infty} \int_{\s^2}2(2+k)r^{-1-k}D\partial_{\rho}\psi\slashed{\Delta}_{\s^2}\psi\\
&-2Dr^{-k}\partial_{\rho}\psi\slashed{\Delta}_{\s^2}\partial_{\rho}\psi\,d\omega d\rho\\
=&\:\int_{r_{+}}^{\infty} \int_{\s^2}2(2+k)r^{-1-k}D\partial_{\rho}\psi\slashed{\Delta}_{\s^2}\psi\\
&+2Dr^{-k}|\snabla_{\s^2}\partial_{\rho}\psi|^2\,d\omega d\rho,
\end{split}
\end{equation}
where we used that all resulting boundary terms vanish by $D(r_+)=0$ and the asymptotics on $\psi$ as $\rho\to \infty$ in the assumptions. 

We apply \eqref{eq:inhomelliptic} again to estimate
\begin{equation*}
\begin{split}
2(2+k)r^{-1-k}D\partial_{\rho}\psi\slashed{\Delta}_{\s^2}\psi=&-2(2+k)r^{-1-k}D\partial_{\rho}\psi\partial_{\rho}(Dr^2\partial_{\rho}\psi)\\
&+2(2+k)r^{-1-k}D\partial_{\rho}\psi\cdot r^2 F_T\\
=&-(2+k)r^{-3-k}\partial_{\rho}((Dr^2\partial_{\rho}\psi)^2)+2(2+k)r^{-1-k}D\partial_{\rho}\psi\cdot r^2F_T\\
\geq &-(2+k)r^{-3-k}\partial_{\rho}((Dr^2\partial_{\rho}\psi)^2)-\epsilon_1 (2+k)r^{-k}D^2(\partial_{\rho}\psi)^2\\
&-(2+k)\epsilon_1^{-1}r^{-2-k}r^4F_T^2,
\end{split}
\end{equation*}
where we will take $\epsilon_1>0$ sufficiently small later. 

We can further integrate by parts to estimate
\begin{equation*}
-\int_{r_{+}}^{\infty}(2+k)r^{-3-k}\partial_{\rho}((Dr^2\partial_{\rho}\psi)^2)\,d\rho=-\int_{r_{+}}^{\infty}(2+k)(3+k)r^{-4-k}(Dr^2\partial_{\rho}\psi)^2\,d\rho,
\end{equation*}
where we used that the boundary term vanishes at $r=r_+$ and moreover $\lim_{\rho \to\infty}r^{1-k}D(\partial_{\rho}\psi)^2=0$ if $k>-3$.
We furthermore apply \eqref{eq:inhomelliptic} once more estimate to express the $(\slashed{\Delta}_{\s^2}\psi)^2$ term on the left-hand side of \eqref{eq:waveeqestlrplus} in terms of $(\partial_{\rho}((Dr^2\partial_{\rho}\psi))^2$:
\begin{equation*}
r^{-2-k}(\slashed{\Delta}_{\s^2}\psi)^2\geq (1-\epsilon_2)r^{-2-k}(\partial_{\rho}((Dr^2\partial_{\rho}\psi))^2+\left(1-\frac{1}{\epsilon_2}\right)r^{-k}r^2F_T^2,
\end{equation*}
where we will take $\epsilon_2>0$ suitably small later.

Putting all the above estimates together and filling them into the inequality \eqref{eq:waveeqestlrplus}, we obtain
\begin{equation}
\label{eq:waveeqestlrplus2}
\begin{split}
\int_{r_{+}}^{\infty}&\int_{\s^2} (2-\epsilon_2)r^{-2-k}(\partial_{\rho}(Dr^2\partial_r\psi))^2+2Dr^{-k}|\snabla_{\s^2}\partial_r\psi|^2\\
&-(2+k)(3+k+\epsilon_1)r^{-4-k}(Dr^2\partial_{\rho}\psi)^2\,d{\omega}d\rho\\
\leq&\: C\int_{r_{+}}^{\infty}\int_{\s^2}r^{2-k}F_T^2\,d{\omega}d\rho
\end{split}
\end{equation}
We apply a Hardy inequality to absorb part of the third term on the left-hand side into the first term
\begin{equation*}
\begin{split}
\int_{r_{+}}^{\infty}&\frac{1}{2}(1-\epsilon_3)(3+k+\epsilon_1)^2r^{-4-k}(Dr^2\partial_{\rho}\psi)^2\,d\rho\\
\leq&\: 2(1-\epsilon_3)\left(\frac{3+k+\epsilon_1}{3+k}\right)^2\int_{r_{+}}^{\infty}r^{-2-k}(\partial_{\rho}(Dr^2\partial_{\rho}\psi))^2\,d\rho,
\end{split}
\end{equation*}
where we used that $\lim_{\rho\to \infty}r^{-3-k}(Dr^2\partial_{\rho}\psi)^2=0$ if $k>-3$ by assumption, and we take $\epsilon_i>0$ suitably small.

We are left with absorbing the integral of the term:
\begin{equation*}
\left[\left(2+k-\frac{1}{2}(1-\epsilon_3)(3+k+\epsilon_1)\right)(3+k+\epsilon_1)\right]r^{-4-k}(Dr^2\partial_{\rho}\psi)^2
\end{equation*}
into the left-hand side of \eqref{eq:waveeqestlrplus2}. At this point we localize in angular frequencies $\geq \ell$ and we therefore need
\begin{equation*}
\begin{split}
\int_{r_{+}}^{\infty}&\int_{\s^2}\left[\left(2+k-\frac{1}{2}(1-\epsilon_3)(3+k+\epsilon_1)\right)(3+k+\epsilon_1)\right]r^{-k}D^2(\partial_{\rho}\psi_{\geq \ell} )^2\,d\omega d\rho\\
\leq &\:2\ell(\ell+1)\int_{r_{+}}^{\infty}\int_{\s^2}r^{-k}D^2(\partial_{\rho}\psi_{\geq \ell} )^2\,d\omega d\rho\\
\leq&\:2\int_{r_{+}}^{\infty}\int_{\s^2}r^{-k}D^2|\snabla_{\s^2}\partial_{\rho}\psi_{\geq \ell} |^2\,d\omega d\rho,
\end{split}
\end{equation*}
where we applied a Poincar\'e inequality to $\psi_{\geq \ell}$ in the final inequality. 

We can absorb the very right-hand side of the above equation into the $|\snabla_{\s^2}\partial_r\psi_{\geq \ell} |^2$ term on the left-hand side of \eqref{eq:waveeqestlrplus2}, provided used that $|D|\leq 1$ and $k$ satisfies the inequality:
\begin{equation*}
(4+2k-(1-\epsilon_3)(3+k+\epsilon_1))(3+k+\epsilon_1)\leq 4\ell(\ell+1),
\end{equation*}
which we can rewrite as
\begin{equation}
\label{eq:mainineqk}
(1+k+\epsilon_3k+3\epsilon_3-\epsilon_1+\epsilon_1\epsilon_3)(3+k+\epsilon_1)\leq 4\ell(\ell+1).
\end{equation}

Equivalently, whenever
\begin{equation}
\label{eq:altineqk}
(1+k)(3+k)< 4\ell(\ell+1).
\end{equation}
we can find $\epsilon_1,\epsilon_3$ suitably small such that \eqref{eq:mainineqk} is satisfied. The inequality \eqref{eq:altineqk} holds when $-3-2\ell<k<2\ell-1$. Then we use the above estimates to obtain:
\begin{equation}
\label{eq:mainineqellipticell}
\begin{split}
\int_{r_{\rm min}}^{\infty}&\int_{\s^2} r^{-2-k}(\partial_{\rho}(Dr^2\partial_{\rho}\psi))^2+Dr^{-k}|\snabla_{\s^2}\partial_{\rho}\psi|^2\,d\rho\\
\leq&\: C(\ell,k)\int_{r_{\rm min}}^{\infty}r^{2-k}F_T^2\,d{\omega}d\rho.
\end{split}
\end{equation}

Hence, by applying Young's inequality to the terms in $F_T^2$, we are left with:
\begin{equation*}
\begin{split}
\int_{r_{\rm min}}^{\infty}&\int_{\s^2} r^{-2-k}(\partial_{\rho}(Dr^2\partial_{\rho}\psi_{\geq \ell} ))^2+Dr^{-k}|\snabla_{\s^2}\partial_{\rho}\psi_{\geq \ell} |^2+r^{-2-k}(\slashed{\Delta}_{\s^2}\psi_{\geq \ell} )^2\,d{\omega}d\rho\\
\leq&\: C_{k,\ell}\int_{r_{\rm min}}^{\infty}\int_{\s^2}r^{2-k}(\partial_{\rho} T \psi_{\geq \ell} )^2+r^{-k}(T\psi_{\geq \ell} )^2+r^{-k-2\eta}(T^2\psi_{\geq \ell})^2\,d{\omega}d\rho.
\end{split}
\end{equation*}
\end{proof}
We demonstrate now a localized version of the previous result, which can be seen as a ``spacelike" redshift estimate.

\begin{proposition}\label{cor:el}
Assume that
\begin{align*}
\sum_{k\leq 1}\lim_{\rho\to \infty}r T^k\psi<&\infty,\\
\lim_{\rho\to \infty}r^2\partial_r\psi=&0.
\end{align*}
Let $\psi$ be a solution to \eqref{eq:waveequation} and let $\ell\geq 1$. Then we have that
\begin{align*}
\int_{r_+}^{R } \int_{\mathbb{S}^2}  ( \partial_{\rho} \psi_{\geq \ell} )^2 \, d\omega d\rho + & \left. \int_{\mathbb{S}^2} | \snabla_{\mathbb{S}^2} \psi_{\geq \ell} |^2  \, d\omega \right|_{\rho = r_+} \\ \leq & C \int_{r_{+}}^{\infty}\int_{\s^2}\left[ r^{2-k} (\partial_{\rho} T \psi_{\geq \ell} )^2+r^{-k-2\eta}(T^2\psi_{\geq \ell} )^2+r^{-k}(T\psi_{\geq \ell} )^2 \right] \,d\omega d\rho,
\end{align*}
for any $-3 < k < 2\ell-1$, any $\eta > 0$, and where $C = C ( D , \ell , k)$.
\end{proposition}
\begin{proof}
We multiply equation \eqref{eq:inhomelliptic} with $\partial_{\rho} \psi$ and after integrating it in $\{ r_+ \leq r \leq R \}$ we get that:
\begin{align*}
\int_{r_+}^R & \int_{\mathbb{S}^2} \left[ ( \partial_{\rho} \psi_{\geq \ell} ) \left( \partial_{\rho} ( Dr^2 \partial_{\rho} \psi_{\geq \ell} ) \right) + ( \partial_{\rho} \psi_{\geq \ell} ) ( \slashed{\Delta}_{\mathbb{S}^2} \psi_{\geq \ell} ) \right] d\omega d\rho \\ = & \int_{r_+}^R  \int_{\mathbb{S}^2} (Dr^2 )' ( \partial_{\rho} \psi_{\geq \ell} )^2 \, d\omega d\rho + \frac{1}{2} \left. \int_{\mathbb{S}^2} ( D(R) R^2 ) ( \partial_{\rho} \psi_{\geq \ell} )^2 \, d\omega \right|_{\rho = R} \\ & + \frac{1}{2} \left. \int_{\mathbb{S}^2} | \snabla_{\mathbb{S}^2} \psi_{\geq \ell} |^2 \, d\omega \right|_{\rho = r_+} -  \frac{1}{2} \left. \int_{\mathbb{S}^2} | \snabla_{\mathbb{S}^2} \psi_{\geq \ell} |^2 \, d\omega \right|_{\rho = R} ,
\end{align*}
which implies that
\begin{align*}
\int_{r_+}^R & \int_{\mathbb{S}^2} (Dr^2 )' ( \partial_{\rho} \psi_{\geq \ell} )^2 \, d\omega d\rho +  \frac{1}{2} \left. \int_{\mathbb{S}^2} ( D(R) R^2 ) ( \partial_{\rho} \psi_{\geq \ell} )^2 \right|_{\rho = R} \\ & + \frac{1}{2} \left. \int_{\mathbb{S}^2} | \snabla_{\mathbb{S}^2} \psi_{\geq \ell} |^2 \right|_{\rho = r_+} \lesssim  \frac{1}{\epsilon} \int_{r_+}^R  \int_{\mathbb{S}^2}  |F_T |^2 +\epsilon \int_{r_+}^R  \int_{\mathbb{S}^2} ( \partial_{\rho} \psi_{\geq \ell} )^2+ \frac{1}{2} \left. \int_{\mathbb{S}^2} | \snabla_{\mathbb{S}^2} \psi_{\geq \ell} |^2 \right|_{\rho = R} ,
\end{align*}
and by absorbing the second term of the right hand side in the left hand side, and by using estimate \eqref{eq:ellipticpsiell} of Theorem \ref{thm:el} for the last term (which we can use for any $-3 < k < 2\ell -1$ as we work in a compact region) after applying an averaging argument and Poincar\'{e}'s inequality \eqref{poincare3}, we get the desired result. 
\end{proof}

A higher order version of the previous proposition is the following.

\begin{corollary}\label{cor:elb}
Assume that
\begin{align*}
\sum_{k\leq 1}\lim_{\rho\to \infty}r T^k\psi<&\infty,\\
\lim_{\rho\to \infty}r^2\partial_r\psi=&0.
\end{align*}
Let $\psi$ be a solution to \eqref{eq:waveequation} and let $\ell\geq 1$. Then we have that
\begin{align*}
\int_{r_+}^R & \int_{\mathbb{S}^2} \Big[ ( \partial_{\rho}^2 \psi_{\geq \ell} )^2 + | \slashed{\nabla}_{\mathbb{S}^2} ( \partial_{\rho} \psi_{\geq \ell} ) |^2 \Big] \, d\omega d\rho \\ \leq & C \int_{r_+}^R \int_{\mathbb{S}^2} | \partial_{\rho} F_T |^2 \, d\omega d\rho \\ & + C\int_{r_{+}}^{\infty}\int_{\s^2}\left[ r^{2-k} (\partial_{\rho} T \psi_{\geq \ell} )^2+r^{-k-2\eta}(T^2\psi_{\geq \ell} )^2+r^{-k}(T\psi_{\geq \ell} )^2 \right] \,d\omega d\rho,
\end{align*}
for any $-3 < k < 2\ell-1$, any $\eta > 0$, and where $C = C ( D , \ell , k)$.
\end{corollary}
\begin{proof}
We consider the equation
\begin{equation}\label{eq:waveaux1}
\partial_{\rho} ( (Dr^2 ) \partial_{\rho}^2 \psi_{\geq \ell} ) + (Dr^2 )' \partial_{\rho}^2 \psi_{\geq \ell} + 2 \partial_{\rho} \psi_{\geq \ell} = \partial_{\rho} \bar{F}_T ,
\end{equation}
where
$$ \bar{F}_T =  ( 2h Dr^2 - 2r^2 ) \partial_{\rho} T \psi_{\geq \ell} + ( 2hr^2 - h^2 Dr^2 + 2h^2 Dr^2 -2hr^2 ) T^2 \psi_{\geq \ell} + [ ( h D r^2 )' - 2 h Dr + 2hDr - 2r] T \psi_{\geq \ell} . $$
We multiply equation \eqref{eq:waveaux1} by $r \partial_{\rho}^2 \psi_{\geq \ell}$ and we integrate it over $\{ r_+ \leq r \leq R \}$ and we have that
\begin{align*}
\int_{r_+}^R & \int_{\mathbb{S}^2} \Big[ 2r (D r^2 )' ( \partial_{\rho}^2 \psi_{\geq \ell} )^2 - (Dr^2 ) ( \partial_{\rho}^2 \psi_{\geq \ell} )^2 \Big] \, d\omega d\rho \\ & + \int_{r_+}^R  \int_{\mathbb{S}^2} \Big[ \frac{1}{2} | \slashed{\nabla}_{\mathbb{S}^2}  ( \partial_{\rho} \psi_{\geq \ell} ) |^2 - ( \partial_{\rho} \psi_{\geq \ell } )^2 \Big] \, d\omega d\rho \\ \leq &  \int_{r_+}^R  \int_{\mathbb{S}^2} | r ( \partial_{\rho}^2 \psi_{\geq \ell} ) ( \partial_{\rho} \bar{F}_T ) | \, d\omega d\rho ,
\end{align*}
and we get the required estimate by applying Cauchy-Schwarz on the term in the right-hand side, absorbing the term involving $\partial_{\rho}^2 \psi_{\geq \ell}$ from the left-hand side (which has a term with $\partial_{\rho}^2 \psi_{\geq \ell}$ and no degeneracy in $D$ by the choice of $R$ as the second term of the left-hand side above can be absorbed by the first), and by applying Theorem \ref{thm:el} to the term with the minus sign of the left-hand side.
\end{proof}
Note that by the Corollary above and Theorem \ref{thm:el} we also get that
\begin{align*}
\int_{r_+}^R & \int_{\mathbb{S}^2} \Big[ r^{2-k} ( \partial_{\rho}^2 \psi_{\geq \ell} )^2 + r^{-k} | \slashed{\nabla}_{\mathbb{S}^2} ( \partial_{\rho} \psi_{\geq \ell} ) |^2 \Big] \, d\omega d\rho \\ \leq & C \int_{r_+}^R \int_{\mathbb{S}^2} | \partial_{\rho} F_T |^2 \, d\omega d\rho \\ & + C\int_{r_{+}}^{\infty}\int_{\s^2}\left[ r^{2-k} (\partial_{\rho} T \psi_{\geq \ell} )^2+r^{-k-2\eta}(T^2\psi_{\geq \ell} )^2+r^{-k}(T\psi_{\geq \ell} )^2 \right] \,d\omega d\rho,
\end{align*}
for $-3 < k < 2\ell-1$, any $\eta > 0$, and where $C = C ( D , \ell , k)$ is a constant.

Now we generalize Theorem \ref{thm:el} for higher $\partial_{\rho}$ derivatives.

\begin{theorem}\label{thm:elh}
Fix some $\ell \geq 1$. Assume that
\begin{align*}
\sum_{l\leq m+1 }\lim_{\rho\to \infty}r T^l\psi<&\infty,\\
\lim_{\rho\to \infty}r^{2+l}\partial_r^{l+1}\psi=&0 \mbox{  for $l \leq m$}.
\end{align*}
Then for all 
$$ -3 -2m < k < 2\ell -2m -1 \mbox{  and  } 0 \leq m \leq \ell , $$
we have that
\begin{align*}
\int_{r_+}^{\infty} & \int_{\mathbb{S}^2} \left[ r^{-2-k} \left( \partial_{\rho} \left( ( Dr^2 ) \partial_{\rho}^{m+1} \psi_{\geq \ell} \right) \right)^2 + r^{-k} ( \partial_{\rho}^{m+1} \psi_{\geq \ell} )^2 \right] \, d\omega d\rho \\ &+\int_{r_+}^{\infty}  \int_{\mathbb{S}^2} \left[ r^{-k} D | \snabla_{\mathbb{S}^2} (  \partial_{\rho}^{m+1} \psi_{\geq \ell} ) |^2 + r^{-2-k} ( \slashed{\Delta}_{\mathbb{S}^2} ( \partial_{\rho}^m \psi_{\geq \ell} ) )^2  \right] \, d\omega d\rho \\ \leq & C \sum_{s=1}^{m+1} \int_{r_+}^{\infty} \int_{\mathbb{S}^2} \left[ r^{2-k - 2(m-s+1)} ( \partial_{\rho}^{s} T \psi_{\geq \ell} )^2 + r^{-k-2\eta-2(m-s+1)} ( \partial_{\rho}^{s-1} T^2 \psi_{\geq \ell} )^2 \right]\, d\omega d\rho \\ & + \int_{r_+}^{\infty} r^{-k-2(m-s+1)} ( \partial_{\rho}^{s-1} T \psi_{\geq \ell} )^2  \, d\omega d\rho   ,
\end{align*}
for some $C = C (D, m , k , \ell)$ and for any $\eta > 0$.
\end{theorem}

\begin{proof}
We will argue by induction. The base case is $m=0$ which was proven through a combination of Theorem \ref{thm:el} and Corollary \ref{cor:el}.  We assume that the desired estimate is true for $n= m -1$, and we also make the following induction hypothesis which due to Corollary \ref{cor:elb} is true in the base case: we assume that for some small enough $r_c < \infty$ that
\begin{align}\label{el:ind}
\int_{r_+}^{\infty} & \int_{\mathbb{S}^2} \Big[ r^{2-k'} ( \partial_{\rho}^{m+1} \psi_{\geq \ell} )^2 + r^{-k'} | \slashed{\nabla}_{\mathbb{S}^2} ( \partial_{\rho}^m \psi_{\geq \ell} ) |^2 \, d\omega d\rho \\ \lesssim &\int_{r_+}^{r_c}  \int_{\mathbb{S}^2} | \partial_{\rho}^{m+1} F_T |^2 d\omega d\rho \\ & + \sum_{s=1}^{m} \int_{r_+}^{\infty} \int_{\mathbb{S}^2} \left[ r^{2-k' -2(s-1)} ( \partial_{\rho}^{s} T \psi_{\geq \ell} )^2 + r^{-k' -2\eta-2(s-1)} ( \partial_{\rho}^{s-1} T^2 \psi_{\geq \ell} )^2 + r^{-k' -2(s-1)} ( \partial_{\rho}^{s-1} T \psi_{\geq \ell} )^2 \right] \, d\omega d\rho  ,
\end{align}
for $-3 -2 (m-1) < k' < 2\ell -2 (m-1 ) -1$. Now we examine the case $n=m$. We have the equation
\begin{equation}\label{eq:waveauxl}
\partial_{\rho} \left( ( Dr^2 ) \partial_{\rho}^{m+1} \psi_{\geq \ell} \right) + m ( Dr^2 )' \partial_{\rho}^{m+1} \psi_{\geq \ell} + m ( m+1 ) \partial_{\rho}^m \psi_{\geq \ell} + \slashed{\Delta}_{\mathbb{S}^2} ( \partial_{\rho}^m \psi_{\geq \ell} ) = \partial_{\rho}^m \bar{F}_T , 
\end{equation}
with $F_T$ as given after equation \eqref{eq:waveaux1}. We rearrange it as
\begin{equation*}
\partial_{\rho} \left( ( Dr^2 ) \partial_{\rho}^{m+1} \psi_{\geq \ell} \right) + \slashed{\Delta}_{\mathbb{S}^2} ( \partial_{\rho}^m \psi_{\geq \ell} ) = \partial_{\rho}^m \bar{F}_T - m ( Dr^2 )' \partial_{\rho}^{m+1} \psi_{\geq \ell} - m ( m+1 ) \partial_{\rho}^m \psi_{\geq \ell} , 
\end{equation*}
and after multiplying it by $r^{-2-k}$ and integrating in $\rho$ we have that:
\begin{align*}
\int_{r_+}^{\infty} & \int_{\mathbb{S}^2} r^{-2-k} \left[ \left( \partial_{\rho} \left( ( Dr^2 ) \partial_{\rho}^{m+1} \psi_{\geq \ell} \right) \right)^2 + ( \slashed{\Delta}_{\mathbb{S}^2} ( \partial_{\rho}^m \psi_{\geq \ell} ) )^2 \right]   \, d\omega d\rho  \\ & +  2 \int_{r_+}^{\infty} \int_{\mathbb{S}^2} r^{-2-k} ( \partial_{\rho} \left( ( Dr^2 ) \partial_{\rho}^{m+1} \psi_{\geq \ell} \right) ) \cdot ( \slashed{\Delta}_{\mathbb{S}^2} ( \partial_{\rho}^m \psi_{\geq \ell} ) ) \, d\omega d\rho \\ \leq & 2\int_{r_+}^{\infty} \int_{\mathbb{S}^2} r^{-2-k} (\partial_{\rho}^m \bar{F}_T )^2 \, d\omega d\rho + 2 \int_{r_+}^{\infty} \int_{\mathbb{S}^2} r^{-2-k} \Big[ m^2 [ (Dr^2 )' ]^2 ( \partial_{\rho}^{m+1} \psi_{\geq \ell} )^2 + m^2 (m+1)^2 ( \partial_{\rho}^m \psi_{\geq \ell} )^2 \Big] \, d\omega d\rho . 
\end{align*}
For the cross term we have that
\begin{align*}
2 \int_{r_+}^{\infty} \int_{\mathbb{S}^2} & r^{-2-k}( \partial_{\rho} \left( ( Dr^2 ) \partial_{\rho}^{m+1} \psi \right) ) \cdot ( \slashed{\Delta}_{\mathbb{S}^2} ( \partial_{\rho}^m \psi ) ) \Big] \, d\omega d\rho =  \int_{r_+}^{\infty} \int_{\mathbb{S}^2} r^{-2-k} 2 (Dr^2 ) | \slashed{\nabla}_{\mathbb{S}^2} ( \partial_{\rho}^{m+1} \psi ) |^2 \, d\omega d\rho \\ & + \int_{r_+}^{\infty} \int_{\mathbb{S}^2} r^{-2-k} \Big[ 2 ( 2+k ) D' r - 2 (2+k) (1+k ) D \Big] | \slashed{\nabla}_{\mathbb{S}^2} ( \partial_{\rho}^{m} \psi ) |^2 \, d\omega d\rho .
\end{align*}
Using the last two computations we now have that
\begin{align*}
\int_{r_+}^{\infty} & \int_{\mathbb{S}^2} r^{-2-k} \left[ \left( \partial_{\rho} \left( ( Dr^2 ) \partial_{\rho}^{m+1} \psi_{\geq \ell} \right) \right)^2 + D r^{-k} | \slashed{\nabla}_{\mathbb{S}^2} ( \partial_{\rho}^{m+1} \psi_{\geq \ell} ) |^2 + ( \slashed{\Delta}_{\mathbb{S}^2} ( \partial_{\rho}^m \psi_{\geq \ell} ) )^2 \right]   \, d\omega d\rho \\ \lesssim & \int_{r_+}^{\infty} \int_{\mathbb{S}^2} r^{-2-k} (\partial_{\rho}^m \bar{F}_T )^2 \, d\omega d\rho + 2 \int_{r_+}^{\infty} \int_{\mathbb{S}^2} r^{-2-k} \Big[ m^2 [ (Dr^2 )' ]^2 ( \partial_{\rho}^{m+1} \psi_{\geq \ell} )^2 + m^2 (m+1)^2 ( \partial_{\rho}^m \psi_{\geq \ell} )^2 \Big] \, d\omega d\rho \\ & + \int_{r_+}^{\infty} \int_{\mathbb{S}^2} r^{-2-k} | \slashed{\nabla}_{\mathbb{S}^2} ( \partial_{\rho}^m \psi_{\geq \ell} ) |^2  \, d\omega d\rho . 
\end{align*}
For the last three terms of the right-hand side we now use the induction hypothesis, and specifically estimate \eqref{el:ind}. Note that we can do this because now as $-3-2m < k < 2\ell -2m -1$ we have that $-3 -2(m-1) < k+2 < 2\ell -2 (m-1) -1$. This finishes the proof of the desired estimate.
\end{proof}

\section{Higher-order redshift and energy estimates}\label{redshift}
Additionally to the elliptic estimates we will also need some higher order red-shift estimates. Recall that
$$ N := T - \partial_r \mbox{  for $r \in [r_+ , r_1 ]$,} $$
$$ N := T \mbox{  for $r \geq r_2$,} $$
for $r_+ \leq r_1 < r_2$, in ingoing $(v,r,\theta,\varphi)$ coordinates.

\begin{lemma}\label{red:aux}
Fix some $\ell \geq 1$, and assume that $\psi$ solves \eqref{eq:waveequation}. Then for some $r_1$, $r_2$ such that $r_+ < r_1 < r_2$ and $r_2 - r_1$ being small enough, $0 \leq m \leq \ell$ and any $\tau_0 \leq \tau_1 \leq \tau_2$, we have that:
\begin{equation}\label{est:red}
\begin{split}
\int_{S_{\tau_2}} & J^N [ N^{m+1} \psi_{\geq \ell} ] \cdot n_{\tau_2} \, d\mu_{S_{\tau_2}} +  \int_{\tau_1}^{\tau_2} \int_{S_{\tau}} J^N [ N^{m+1} \psi_{\geq \ell} ] \cdot n_{\tau} \, d\mu_{S_{\tau}} d\tau \\ \leq & C \int_{S_{\tau_1}} J^N [ N^{m+1} \psi_{\geq \ell} ] \cdot n_{\tau_1} \, d\mu_{S_{\tau_1}} + C \int_{S_{\tau_1}} J^N [ N^{m+1} T \psi_{\geq \ell} ] \cdot n_{\tau_1} \, d\mu_{S_{\tau_1}} \\ & + C\int_{\tau_1}^{\tau_2} \int_{S_{\tau} \cap \{r_1 \leq r \leq r_2 \}} J^T [ N^{m+1} \psi_{\geq \ell} ] \cdot n_{\tau} \, d\mu_{S_{\tau}} d\tau + C\sum_{k+j = m+1 , k \leq m} \int_{\tau_1}^{\tau_2} \int_{S_{\tau}} J^N [ N^k T^j \psi_{\geq \ell} ] \cdot n_{\tau} \, d\mu_{S_{\tau}} d\tau \\ & + C\int_{\tau_1}^{\tau_2} \int_{S_{\tau} \cap \{ r \leq r_2 \}} J^T [ N^{m} \psi_{\geq \ell} ] \cdot n_{\tau} \, d\mu_{S_{\tau}} d\tau ,
\end{split}
\end{equation}
for some $C = C( D,h) > 0$.
\end{lemma}
\begin{proof}
From the proof of the above Lemma in the case $m=0$ which was established in \cite{lecturesMD} we have the following estimate:
\begin{equation}
\begin{split}
\int_{S_{\tau_2}} & J^N [ N^{m+1} \psi_{\geq \ell} ] \cdot n_{\tau_2} \, d\mu_{S_{\tau_2}} +  \int_{\tau_1}^{\tau_2} \int_{S_{\tau}} J^N [ N^{m+1} \psi_{\geq \ell} ] \cdot n_{\tau} \, d\mu_{S_{\tau}} d\tau \\  \leq & C \int_{S_{\tau_1}} J^N [ N^{m+1} \psi_{\geq \ell} ] \cdot n_{\tau_1} \, d\mu_{S_{\tau_1}} + \int_{S_{\tau_1}} J^N [ N^{m+1} T \psi_{\geq \ell} ] \cdot n_{\tau_1} \, d\mu_{S_{\tau_1}}\\ & + C\int_{\tau_1}^{\tau_2} \int_{S_{\tau} \cap \{r_1 \leq r \leq r_2 \}} J^T [ N^{m+1} \psi_{\geq \ell} ] \cdot n_{\tau} \, d\mu_{S_{\tau}} d\tau + C\int_{\tau_1}^{\tau_2} \int_{S_{\tau}} J^T [ N^{m} T \psi_{\geq \ell} ] \cdot n_{\tau} \, d\mu_{S_{\tau}} d\tau \\ & + C\int_{\tau_1}^{\tau_2} \int_{S_{\tau} \cap \{ r \leq r_2 \}} J^T [ N^{m} \psi_{\geq \ell} ] \cdot n_{\tau} \, d\mu_{S_{\tau}} d\tau + C\int_{\tau_1}^{\tau_2} \int_{S_{\tau}} | G_{m, \ell} |^2 \, d\mu_{S_{\tau}} d\tau  ,
\end{split}
\end{equation}
for some constant $C$, where $G_{m , \ell}$ is defined as
$$ \Box_g ( N^m \psi_{\ell} ) = \kappa_{m} \partial_r^{m+1} \psi_{\ell} + G_{m,\ell} , $$
where
$$ G_{m , \ell} = \sum_{k=0}^{m-1} O(r^{-1-k} ) ( \partial_r^{m-1-k} T \psi_{\geq \ell} ) + \sum_{k=0}^{m-1} O (r^{-2-k} ) ( \partial_r^{m-1-k} N \psi_{\geq \ell} ) . $$  
Note that the previous estimate holds true as $\kappa_m > 0$ as noted in Proposition 3.3.2 of \cite{lecturesMD}. Now using Theorem \ref{thm:elh} repeatedly for both terms of $G_{m ,\ell}$ we get that:
$$ \int_{\tau_1}^{\tau_2} \int_{S_{\tau}} | G_{m , \ell} |^2 \, d\mu_{S_{\tau}} d\tau \leq C \sum_{k+j = m+1 , k\leq m} \int_{\tau_1}^{\tau_2} \int_{S_{\tau}} J^N [ N^k T^j \psi_{\geq \ell} ] \cdot n_{\tau} \, d\mu_{S_{\tau}} d\tau , $$
which finishes the proof of the desired estimate.

\end{proof}

Using the previous Lemma we can now show decay for the following energies:
\begin{proposition}\label{dec:red}
Fix some $\ell \geq 1$ and assume that $\psi$ solves \eqref{eq:waveequation}. Assume that 
$$ I_{\ell} [\psi] \neq 0 , $$
and that
\begin{equation}\label{as:red}
\sum_{j_1 + j_2 \leq m , j \leq K+1} \int_{\Sigma_{\tau_0}} J^N [ N^{j_1 +1} T^{j_2 + j} \psi_{\geq \ell} ] \cdot n_{\tau_0} \, d\mu_{\Sigma_{\tau_0}} < \infty , 
\end{equation}
for some $m$ such that $0 \leq m \leq \ell$, and for some $K \in \mathbb{N}$. 

Then for any $u \geq u_0$ we have that
\begin{equation}\label{dec:redn}
\int_{S_{\tau}} J^N [ N^{m+1} T^K \psi_{\geq \ell} ] \cdot n_u \, d\mu_{\mathcal{S}_{\tau}} \leq C \frac{E_{aux-decomp-n0-K}}{\tau^{3+2(m+1) +2K - \delta}} ,
\end{equation}
for any $\delta > 0$, some $C = C (D,R)$, and for $E_{aux-decomp-n0-K}$ as in Corollary \ref{rem:n+2t}.

If instead we assume that
$$ I_{\ell} [\psi ] = 0 , $$
and maintain the assumption \eqref{as:red}, we have for any $u \geq u_0$ that
\begin{equation}\label{dec:redn0}
\int_{S_u} J^N [ N^{m+1} T^K \psi_{\geq \ell} ] \cdot n_u \, d\mu_{\Sigma_{\tau}} \leq C \frac{E_{aux-decomp-0-K}}{u^{5+2(m+1) +2K - \delta}} ,
\end{equation}
for any $\delta > 0$, some $C > 0$, and for $E_{aux-decomp-0-K}$ as in Corollary \ref{rem:n+2t}.

\end{proposition}
\begin{proof}
First we assume that $I_\ell \neq 0$. We use estimate \eqref{est:red} and we examine one by one the terms in the right-hand side. For the first one we have for any $\tau_1 < \tau_2$ and $r_1$, $r_2$ as in Lemma \ref{red:aux} that:
\begin{align*}
\int_{\tau_1}^{\tau_2} \int_{S_{\tau} \cap \{r_1 \leq r \leq r_2\}} J^T [ N^{m+1} \psi_{\geq \ell} ] \cdot n_{\tau} \, d\mu_{S_{\tau}} d\tau \leq & C \int_{\tau_1}^{\tau_2} \int_{S_{\tau}} J^N [ N^m T\psi_{\geq \ell} ] \cdot n_{\tau} \, d\mu_{S_{\tau}} d\tau \\ & + \int_{\tau_1}^{\tau_2} \int_{S_{\tau}} ( \partial_{\rho}^m T \psi_{\geq \ell} )^2 \, d\mu_{S_{\tau}} d\tau ,
\end{align*}
using Theorem \ref{thm:elh}. We can also improve the above estimate by applying Hardy's inequality \eqref{hardy2} to the second term of the last estimate hence arriving at:
\begin{equation*}
\int_{\tau_1}^{\tau_2} \int_{S_{\tau} \cap \{r_1 \leq r \leq r_2\}} J^T [ N^{m+1} \psi_{\geq \ell} ] \cdot n_{\tau} \, d\mu_{S_{\tau}} d\tau \leq C \int_{\tau_1}^{\tau_2} \int_{S_{\tau}} J^N [ N^m T\psi_{\geq \ell} ] \cdot n_{\tau} \, d\mu_{S_{\tau}} d\tau .
\end{equation*}
Moreover using Hardy's inequality \eqref{hardy2} and Theorem \ref{thm:elh} we get that
$$ \int_{\tau_1}^{\tau_2} \int_{S_{\tau} \cap \{ r\leq r_2\}} J^T [ N^m \psi_{\geq \ell} ] \cdot n_{\tau} \, d\mu_{S_{\tau}} d\tau \leq C \int_{\tau_1}^{\tau_2} \int_{S_{\tau} } J^N [ N^m T \psi_{\geq \ell} ] \cdot n_{\tau} \, d\mu_{S_{\tau}} d\tau . $$
We have the following estimate in the end:
\begin{align*}
 \int_{S_{\tau_2}} & J^N [ N^{m+1} \psi_{\geq \ell} ] \cdot n_{\tau_2} \, d\mu_{S_{\tau_2}} +  \int_{\tau_1}^{\tau_2} \int_{S_{\tau}} J^N [ N^{m+1} \psi_{\geq \ell} ] \cdot n_{\tau} \, d\mu_{S_{\tau}} d\tau \\ \leq & C \int_{S_{\tau_1}} J^N [ N^{m+1} \psi_{\geq \ell} ] \cdot n_{\tau_1} \, d\mu_{S_{\tau_1}} + C \int_{S_{\tau_1}} J^N [ N^{m+1} T \psi_{\geq \ell} ] \cdot n_{\tau_1} \, d\mu_{S_{\tau_1}} \\ & + C  \sum_{k+j = m+1 , k\leq m} \int_{\tau_1}^{\tau_2} \int_{S_{\tau}} J^N [ N^k T^j \psi_{\geq \ell} ] \cdot n_{\tau} \, d\mu_{S_{\tau}} d\tau .
\end{align*}
We will use the last inequality to obtain the desired estimate in an inductive manner. We start with the case $m=0$ where we note that the previous estimate takes the following form:
\begin{align*}
 \int_{S_{\tau_2}} & J^N [ N \psi_{\geq \ell} ] \cdot n_{\tau_2} \, d\mu_{S_{\tau_2}} +  \int_{\tau_1}^{\tau_2} \int_{S_{\tau}} J^N [ N \psi_{\geq \ell} ] \cdot n_{\tau} \, d\mu_{S_{\tau}} d\tau \\ \leq & C \int_{S_{\tau_1}} J^N [ N \psi_{\geq \ell} ] \cdot n_{\tau_1} \, d\mu_{S_{\tau_1}} + C \int_{S_{\tau_1}} J^N [ N T \psi_{\geq \ell} ] \cdot n_{\tau_1} \, d\mu_{S_{\tau_1}} \\ & + C\int_{\tau_1}^{\tau_2} \int_{S_{\tau}} J^N [ T \psi_{\geq \ell} ] \cdot n_{\tau} \, d\mu_{S_{\tau}} d\tau ,
\end{align*}
and using the decay estimates from Proposition \ref{dec:enT} and Corollary \ref{rem:n+2} we get that:
\begin{align*}
 \int_{S_{\tau_2}} & J^N [ N \psi_{\geq \ell} ] \cdot n_{\tau_2} \, d\mu_{S_{\tau_2}} +  \int_{\tau_1}^{\tau_2} \int_{S_{\tau}} J^N [ N \psi_{\geq \ell} ] \cdot n_{\tau} \, d\mu_{S_{\tau}} d\tau \\ \leq & C \int_{S_{\tau_1}} J^N [ N \psi_{\geq \ell} ] \cdot n_{\tau_1} \, d\mu_{S_{\tau_1}} + C \int_{S_{\tau_1}} J^N [ N T \psi_{\geq \ell} ] \cdot n_{\tau_1} \, d\mu_{S_{\tau_1}} \\ & + C \frac{E_{aux-T , I_{\ell} \neq 0 , 1}}{u^{5-\epsilon}} ,
\end{align*}
for any $\epsilon >0$. The desired estimate now follows from the Gr\"{o}nwall-type estimate \eqref{gronwall}. For $m=1$ we follow the same process and we have that:
\begin{align*}
 \int_{S_{\tau_2}} & J^N [ N^{2} \psi_{\geq \ell} ] \cdot n_{\tau_2} \, d\mu_{S_{\tau_2}} +  \int_{\tau_1}^{\tau_2} \int_{S_{\tau}} J^N [ N^{2} \psi_{\geq \ell} ] \cdot n_{\tau} \, d\mu_{S_{\tau}} d\tau \\ \leq & C \int_{S_{\tau_1}} J^N [ N^{2} \psi_{\geq \ell} ] \cdot n_{\tau_1} \, d\mu_{S_{\tau_1}} + C \int_{S_{\tau_1}} J^N [ N^{2} T \psi_{\geq \ell} ] \cdot n_{\tau_1} \, d\mu_{S_{\tau_1}} \\ & + C  \sum_{k+j = 2 , k\leq 1} \int_{\tau_1}^{\tau_2} \int_{S_{\tau}} J^N [ N^k T^j \psi_{\geq \ell} ] \cdot n_{\tau} \, d\mu_{S_{\tau}} d\tau ,
\end{align*}
and by the decay estimates from Proposition \ref{dec:enT} and Corollary \ref{rem:n+2t} and the estimate we just proved for $m=0$ we have that:
\begin{align*}
 \int_{S_{\tau_2}} & J^N [ N \psi_{\geq \ell} ] \cdot n_{\tau_2} \, d\mu_{S_{\tau_2}} +  \int_{\tau_1}^{\tau_2} \int_{S_{\tau}} J^N [ N \psi_{\geq \ell} ] \cdot n_{\tau} \, d\mu_{S_{\tau}} d\tau \\ \leq & C \int_{S_{\tau_1}} J^N [ N \psi_{\geq \ell} ] \cdot n_{\tau_1} \, d\mu_{S_{\tau_1}} + C \int_{S_{\tau_1}} J^N [ N T \psi_{\geq \ell} ] \cdot n_{\tau_1} \, d\mu_{S_{\tau_1}} \\ & + C \frac{E_{aux-decomp-n0-0} + E_{aux-decomp-n0-1}}{u^{7-\epsilon}} ,
\end{align*}
for any $\epsilon >0$. The desired estimate now follows again from the Gr\"{o}nwall-type estimate \eqref{gronwall}. We work in the same way for all other $m \leq \ell$.

For $I_{\ell} = 0$ the proof is identical, we just use the relevant decay energy estimates from Proposition \ref{dec:enT}, Corollaries \ref{rem:n+2} and \ref{rem:n+2t}.
\end{proof}

 \section{Pointwise decay estimates}\label{decay}
\subsection{Almost-sharp decay for the radiation fields $\Phi_{(\ell )}$}
In this Section we demonstrate how to obtain \textit{almost sharp} decay estimates for the radiation fields $P_{\geq \ell} \Phi_{(k)}$, $k \in \{ 0 , \dots , \ell \}$. 

The method of proof is quite standard (and has been used in the Lemma \ref{lm:auxdecay}) so we only give an outline. For a function $f$ and a cut-off $\chi$ as in \ref{prop:generalrpest} ($\chi (r) = 0$ for $r \leq R$, $\chi (r) = 1$ for $r \geq R+1$, and $\chi$ is smooth) we have that:
\begin{align*}
\int_{\mathbb{S}^2} ( \chi f )^2 (u,v) \, d\omega = & 2 \int_{v_{R}}^v ( \chi f ) \cdot ( L (\chi f ) ) \, d\omega dv' \\ \leq & 2 \left( \int_{\mathcal{N}_u} \frac{1}{r^2} ( \chi f )^2 \, d\omega dv \right)^{1/2} \left( \int_{\mathcal{N}_u} r^2 ( L ( \chi f ) )^2 \, d\omega dv \right)^{1/2} \\ \leq & 4 \left( \int_{\mathcal{N}_u}  ( L ( \chi f ) )^2 \, d\omega dv \right)^{1/2} \left( \int_{\mathcal{N}_u} r^2 ( L ( \chi f ) )^2 \, d\omega dv \right)^{1/2} .
\end{align*}
We will fix an $\ell \geq 1$, and we will apply the above computation for $f = P_{\geq \ell} \Phi_{(k)}$ with $k \in \{0 , \dots , \ell \}$, and we will obtain decay for the radiation fields by using the energy decay estimates of Section \ref{energy} (specifically the decay of the energy flux, and the decay for the relevant $r^p$-weighted energy for $p=2$).

\begin{theorem}
Let $\psi$ be a smooth solution of \eqref{eq:waveeqestlrplus} satisfying \eqref{basic:as} and fix some $\ell \geq 1$. Then for $k \in \{0, \dots , \ell \}$ we have in the region $\{ r \geq R \}$ for all $u \geq u_0$ that:
\begin{equation}\label{dec:pwkl}
\int_{\mathbb{S}^2} ( P_{\geq \ell} \Phi_{(k)} )^2 \, d\omega \leq C \frac{E_{aux-decomp-n0-0}}{u^{2+2(\ell - k) - \delta}} ,
\end{equation}
for any $\delta > 0$, for $C = C (D , R)$ and where the quantity $E_{aux-decomp-n0-0}$ is given in Corollary  \ref{rem:n+2t}.

If instead we assume that
$$ I_{\ell} [\psi ] = 0 , $$
then for $k \in \{0, \dots , \ell \}$ we have in the region $\{ r \geq R \}$ for all $u \geq u_0$ that:
\begin{equation}\label{dec:pwkl0}
\int_{\mathbb{S}^2} ( P_{\geq \ell} \Phi_{(k)} )^2 \, d\omega \leq C \frac{E_{aux-decomp-0-0}}{u^{4+2(\ell - k) - \delta}} ,
\end{equation}
for any $\delta > 0$, for $C = C (D , R)$ and where the quantity $E_{aux-decomp-0-0}$ is given in Corollary \ref{rem:n+2t}.
\end{theorem}

We can also apply the same method to $P_{\geq \ell} T^m \Phi_{(k)}$ for any $m \in \mathbb{N}$ and $k \in \{0, \dots , \ell \}$ and in this case we get the following result.

\begin{theorem}
Let $\psi$ be a smooth solution of \eqref{eq:waveeqestlrplus} satisfying \eqref{basic:as}, fix some $\ell \geq 1$, and fix some $m \in \mathbb{N}$. Then for $k \in \{0, \dots , \ell \}$ we have in the region $\{ r \geq R \}$ for all $u \geq u_0$ that:
\begin{equation}\label{dec:pwklt}
\int_{\mathbb{S}^2} ( P_{\geq \ell} T^m \Phi_{(k)} )^2 \, d\omega \leq C \frac{E_{aux-decomp-n0-m}}{u^{2+2m+2(\ell - k) - \delta}} ,
\end{equation}
for any $\delta > 0$, for $C = C (D , R)$ and where the quantity $E_{aux-decomp-n0-m}$ is given in Corollary  \ref{rem:n+2t}.

If instead we assume that
$$ I_{\ell} [\psi ] = 0 , $$
then for $k \in \{0, \dots , \ell \}$ we have in the region $\{ r \geq R \}$ for all $u \geq u_0$ that:
\begin{equation}\label{dec:pwklt0}
\int_{\mathbb{S}^2} ( P_{\geq \ell} T^m \Phi_{(k)} )^2 \, d\omega \leq C \frac{E_{aux-decomp-0-m}}{u^{4+2m+2(\ell - k) - \delta}} ,
\end{equation}
for any $\delta > 0$, for $C = C (D , R)$ and where the quantity $E_{aux-decomp-0-m}$ is given in Corollary  \ref{rem:n+2t}.
\end{theorem}

\subsection{Almost-sharp decay for $P_{\geq \ell} \psi$ and its derivatives}
In this section we will use the elliptic estimates and the energy decay estimates from Section \ref{energy} to derive \textit{almost sharp} decay estimates for a frequency localized linear wave $P_{\geq \ell} \psi$ and its $\partial_r$ derivatives.

\begin{theorem}\label{dec:drk1}
Fix some $\ell \geq 1$ and let $\psi$ be a smooth solution of \eqref{eq:waveequation} satisfying \eqref{basic:as}, and assume that
$$ I_{\ell} [ \psi] \neq 0 . $$

Then for $0 \leq k \leq \ell$ we have that
\begin{equation}\label{dec:drkl1}
\int_{\mathbb{S}^2} r^{-2\ell+\eta+2k} ( \partial_r^k P_{\geq \ell}\psi)^2 \, d\omega \leq C \frac{E_{aux-decomp-n0-k+\ell+1}}{u^{4\ell+4-2\eta-\delta}} ,
\end{equation}
and
\begin{equation}\label{dec:drkl2}
\int_{\mathbb{S}^2} r^{-2\ell+\eta+2k+1} ( \partial_r^k P_{\geq \ell}\psi )^2 \, d\omega \leq C \frac{E_{aux-decomp-n0-k+\ell+1}}{u^{4\ell+3-2\eta-\delta}} ,
\end{equation}
where $C = C ( D , R , k ,\delta, \ell )$, and where $E_{aux-decomp-n0-k+\ell+1}$ is given in Corollary \ref{rem:n+2t}.

If instead we assume that
$$ I_{\ell} [\psi ] = 0 , $$
then 
 for $0 \leq k \leq \ell$ we have that
\begin{equation}\label{dec:drkl1n}
\int_{\mathbb{S}^2} r^{-2\ell+\eta+2k} ( \partial_r^k P_{\geq \ell}\psi )^2 \, d\omega \leq C \frac{E_{aux-decomp-0-k}}{u^{4\ell+6-2\eta-\delta}} ,
\end{equation}
and
\begin{equation}\label{dec:drkl2n}
\int_{\mathbb{S}^2} r^{-2\ell+\eta+2k+1} ( \partial_r^k P_{\geq \ell}\psi )^2 \, d\omega \leq C \frac{E_{aux-decomp-0-k}}{u^{4\ell+5-2\eta-\delta}} ,
\end{equation}
where $C = C ( D , R , k ,\delta, \ell )$, and where $E_{aux-decomp-0-k+\ell+1}$ is given in Corollary \ref{rem:n+2t}.
\end{theorem}
\begin{proof}
We will demonstrate the cases $k=0$ and $k=\ell$, the proof of the remaining cases is done in the same way. 

\textbf{The case $k=0$:} By the fundamental theorem of calculus we have for any $r_0 \geq r_+$ and some $\eta > 0$ that:
\begin{align*}
\int_{\mathbb{S}^2} & r_0^{-2\ell+\eta} (P_{\geq \ell}\psi )^2 \, d\omega =  - \int_{r_0}^{\infty}  \int_{\mathbb{S}^2} \partial_{\rho} \left( r^{-2\ell+\eta} (P_{\geq \ell}\psi )^2  \right)  \, d\omega d\rho \\ = & (2-\eta )  \int_{r_0}^{\infty}  \int_{\mathbb{S}^2} r^{-2\ell-1+\eta} (P_{\geq \ell}\psi )^2 \, d\omega d\rho - \int_{r_0}^{\infty}  \int_{\mathbb{S}^2} 2 r^{-2\ell+\eta} P_{\geq \ell} \psi \cdot ( \partial_{\rho} P_{\geq \ell} \psi ) \, d\omega d\rho \\ \leq &  (2-\eta )  \int_{r_0}^{\infty}  \int_{\mathbb{S}^2} r^{-2\ell-1+\eta} (P_{\geq \ell}\psi )^2 \, d\omega d\rho \\ & + \left( \int_{r_0}^{\infty}  \int_{\mathbb{S}^2} r^{-2\ell-1+\eta} (P_{\geq \ell}\psi )^2 \, d\omega d\rho \right)^{1/2} \left( \int_{r_0}^{\infty}  \int_{\mathbb{S}^2} r^{-2\ell+1+\eta} ( \partial_{\rho} P_{\geq \ell} \psi )^2 \, d\omega d\rho \right)^{1/2} \\ \doteq & I + II .
\end{align*}
We are going to use now the elliptic estimates \eqref{eq:ellipticpsiell} of Theorem \ref{thm:el}. For the term $I$ we apply estimate \eqref{eq:ellipticpsiell} for $k=2\ell-1-\eta$ using the lowest order term and we have that
\begin{align*}
I = &  (2-\eta )  \int_{r_0}^{\infty}  \int_{\mathbb{S}^2} r^{-2\ell-1+\eta} ( P_{\geq \ell} \psi )^2 \, d\omega d\rho \\ \leq & C  \int_{r_+}^{\infty} \int_{\mathbb{S}^2} \left[ r^{-2\ell + 3+\eta} ( \partial_{\rho} T P_{\geq \ell} \psi )^2 + r^{-2\ell +1 +\eta -2 \eta'} ( T^2 P_{\geq \ell} \psi )^2 + r^{-2\ell+1+\eta} (T P_{\geq \ell} \psi )^2 \right] \, d\omega d\rho .
\end{align*}
We repeat the same process $\ell$ times by using again the elliptic estimates of Lemma \ref{thm:el} always using the two lowest order terms without the $D$ weight, but for $T^k P_{\geq \ell} \psi$, $k \in \{1, \dots , \ell\}$ in the place of $P_{\geq \ell} \psi$ (in the last estimate we use the estimate for the lowest order term of Theorem \ref{thm:el} for $T\psi_{\geq \ell}$ and $T^2 \psi_{\geq \ell}$ in the place of $\psi_{\geq \ell}$). In the end we get for some $C = C (D,R,\ell)$ that
\begin{equation*}
 I \leq C  \int_{r_0}^{\infty}  \int_{\mathbb{S}^2} r^{3+\eta} ( \partial_{\rho} T^{\ell +1} P_{\geq \ell}\psi )^2 \, d\omega d\rho   . 
 \end{equation*}
The same process gives us for some $C = C(D,R,\ell)$ that
\begin{equation*}
II \leq C \int_{r_0}^{\infty}  \int_{\mathbb{S}^2} r^{3+\eta} ( \partial_{\rho} T^{\ell +1} P_{\geq \ell}\psi )^2 \, d\omega d\rho  . 
\end{equation*}
The two previous estimates give us now that:
\begin{align*}
\int_{\mathbb{S}^2} & r_0^{-2\ell+\eta} (P_{\geq \ell}\psi )^2 \, d\omega \leq   C \int_{r_0}^{\infty}  \int_{\mathbb{S}^2} r^{3+\eta} ( \partial_{\rho} T^{\ell +1} P_{\geq \ell}\psi )^2 \, d\omega d\rho \\ \leq & \int_{\mathcal{S}_{\tau}} J^N [ T^{\ell+1} P_{\geq \ell} \psi ]\cdot n_{\tau} \, d\mu_{\mathcal{S}_{\tau}}  + \int_{\mathcal{N}_{\tau}} r^{1+\eta} ( L ( P_{\geq \ell} T^{\ell+1} \phi ) )^2 \, d\omega dv    , 
\end{align*}
and the decay follows from the results of Corollary \ref{rem:n+2t}. 

\textbf{The case $k=\ell$:} By the fundamental theorem of calculus we have for any $r_0 \geq r_+$ that:
\begin{align*}
\int_{\mathbb{S}^2} & r_0^{\eta} ( \partial_{\rho}^{\ell} P_{\geq \ell} \psi )^2 \, d\omega  =  - \int_{r_0}^{\infty}  \int_{\mathbb{S}^2} \partial_{\rho} \left( r^{\eta} ( \partial_{\rho} ^{\ell} P_{\geq \ell} \psi )^2 \right) \, d\omega d\rho \\ = & -\eta \int_{r_0}^{\infty}\int_{\mathbb{S}^2} r^{\eta-1} ( \partial_{\rho}^{\ell} P_{\geq \ell} \psi )^2 \, d\omega d\rho - 2 \int_{r_0}^{\infty} \int_{\mathbb{S}^2}r^{\eta} ( \partial_{\rho}^{\ell} P_{\geq \ell} \psi ) \cdot ( \partial_{\rho}^{\ell+1} P_{\geq \ell} \psi ) \, d\omega d\rho \\ \leq &   - 2 \int_{r_0}^{\infty}\int_{\mathbb{S}^2} r^{\eta} ( \partial_{\rho}^{\ell} P_{\geq \ell} \psi ) \cdot ( \partial_{\rho}^{\ell+1} P_{\geq \ell} \psi ) \, d\omega d\rho \\ \leq & \left( \int_{r_0}^{\infty}\int_{\mathbb{S}^2} r^{1+\eta}( \partial_{\rho}^{\ell+1} P_{\geq \ell} \psi )^2 \, d\omega d\rho \right)^{1/2}  \cdot\left( \int_{r_0}^{\infty}\int_{\mathbb{S}^2} r^{-1+\eta}( \partial_{\rho}^{\ell} P_{\geq \ell} \psi )^2 \, d\omega d\rho \right)^{1/2} \\ \leq & C \left( \int_{r_0}^{\infty}\int_{\mathbb{S}^2} r^{1+\eta}( \partial_{\rho}^{\ell+1} P_{\geq \ell} \psi )^2 \, d\omega d\rho \right)^{1/2} \\ & \times \Bigg( \sum_{s=1}^{\ell+1} \Big[\int_{r_0}^{\infty}\int_{\mathbb{S}^2} r^{3+\eta-2(s-\ell+1)}( \partial_{\rho}^{s} T P_{\geq \ell} \psi )^2 + r^{1+\eta -2\eta' -2 (\ell-s+1)} ( \partial_{\rho}^{s-1} T^2 P_{\geq \ell} \psi )^2 \, d\omega d\rho \\ & + \int_{r_0}^{\infty} \int_{\mathbb{S}^2}r^{1+\eta  -2 (\ell-s+1)} ( \partial_{\rho}^{s-1} T P_{\geq \ell} \psi )^2  \, d\omega d\rho \Big] \Bigg)^{1/2} ,
\end{align*}
where in the last inequality we used the higher order elliptic estimates of Theorem \ref{thm:elh} with $m= \ell$ and $k = -1-\eta$. Let us keep only the highest order term (in $\partial_{\rho}$ derivatives) from the very last term as the rest can be turned into it by repeating the same process. We are left with the following:
\begin{align*}
\Bigg( \int_{r_0}^{\infty}\int_{\mathbb{S}^2} & r^{1+\eta}( \partial_{\rho}^{\ell+1} P_{\geq \ell} \psi )^2 \, d\omega dr \Bigg)^{1/2}  \cdot  \Bigg( \int_{r_0}^{\infty}\int_{\mathbb{S}^2} r^{3+\eta}( \partial_{\rho}^{\ell+1} T P_{\geq \ell} \psi )^2 \, d\omega dr \Bigg)^{1/2} \\  \leq & \Bigg( \int_{\mathcal{S}_{\tau}} J^N [ N^{\ell} P_{\geq \ell} \psi ] \cdot n_{\tau} \, d\mu_{\mathcal{S}_{\tau}} + \sum_{s=1}^{\ell} \Big[ \int_{r_0}^{\infty}\int_{\mathbb{S}^2} r^{1+\eta - 2 (\ell-s)} ( \partial_{\rho}^s T P_{\geq \ell} \psi )^2 \, d\omega d\rho \\ &+ \int_{r_0}^{\infty}\int_{\mathbb{S}^2} ( r^{-1+\eta-2\eta' - 2 ( \ell -s)} ( \partial_{\rho}^{s-1} T^2 P_{\geq \ell} \psi )^2 +  r^{-1+\eta - 2 ( \ell -s)} ( \partial_{\rho}^{s-1} T P_{\geq \ell} \psi )^2 ) \, d\omega d\rho \Big] \Bigg)^{1/2} \\ & \times \Bigg( \int_{\mathcal{S}_{\tau}} J^N [ N^{\ell} T P_{\geq \ell} \psi ] \cdot n_{\tau} \, d\mu_{\mathcal{S}_{\tau}} + \sum_{s=1}^{\ell} \Big[ \int_{r_0}^{\infty} r^{3+\eta - 2 (\ell-s)} ( \partial_{\rho}^s T^2 P_{\geq \ell} \psi )^2 \, d\omega d\rho \\ &+ \int_{r_0}^{\infty}\int_{\mathbb{S}^2} ( r^{1+\eta-2\eta' - 2 ( \ell -s)} ( \partial_{\rho}^{s-1} T^3 P_{\geq \ell} \psi )^2 +  r^{1+\eta - 2 ( \ell -s)} ( \partial_{\rho}^{s-1} T^2 P_{\geq \ell} \psi )^2 ) \, d\omega d\rho \Big] \Bigg)^{1/2} ,
\end{align*}
where we broke the integrals into parts close and away from the horizon, and close to the horizon we just bound them by the aforementioned $N$ energies, and away we use again the results of Theorem \ref{thm:elh} for $m = \ell-1$, $k = 1-\eta$ and $k = -1-\eta$. Repeating the same process (and applying Hardy's inequality \eqref{hardy1} once) we are left with the following:
\begin{align*}
\Bigg( & \sum_{m_1+m_2 = \ell}  \int_{\mathcal{S}_{\tau}} J^N [ N^{m_1} T^{m_2} P_{\geq \ell} \psi ] \cdot n_{\tau} \, d\mu_{\mathcal{S}_{\tau}} + \int_{r_0}^{\infty} \int_{\mathbb{S}^2} r^{3+\eta} ( \partial_{\rho} T^{\ell} P_{\geq \ell}\psi )^2 \, d\omega d\rho \Bigg)^{1/2} \\ & \times  \Bigg( \sum_{m_1+m_2 = \ell} \int_{\mathcal{S}_{\tau}} J^N [ N^{m_1} T^{m_2+1} P_{\geq \ell} \psi ] \cdot n_{\tau} \, d\mu_{\mathcal{S}_{\tau}} + \int_{r_0}^{\infty} \int_{\mathbb{S}^2} r^{3+\eta} ( \partial_{\rho}^2 T^{\ell+1} P_{\geq \ell}\psi )^2 \, d\omega d\rho \Bigg)^{1/2} \\ \leq & \Bigg(  \sum_{m_1+m_2 = \ell}  \int_{\mathcal{S}_{\tau}} J^N [ N^{m_1} T^{m_2} P_{\geq \ell} \psi ] \cdot n_{\tau} \, d\mu_{\mathcal{S}_{\tau}} + \int_{\mathcal{N}_{\tau}} r^{1+\eta} ( L (T^{\ell+1} P_{\geq \ell}\phi ) )^2 \, d\omega dv \Bigg)^{1/2} \\ & \times  \Bigg( \sum_{m_1+m_2 = \ell} \int_{\mathcal{S}_{\tau}} J^N [ N^{m_1} T^{m_2+1} P_{\geq \ell} \psi ] \cdot n_{\tau} \, d\mu_{\mathcal{S}_{\tau}} + \int_{\mathcal{N}_{\tau}} r^{1+\eta} ( L (T^{\ell+1} P_{\geq \ell}\phi ) )^2 \, d\omega dv \Bigg)^{1/2}  ,
\end{align*}
after applying again Theorem \ref{thm:el}, and now the result follows from Corollary \ref{rem:n+2t} and Proposition \ref{dec:red}.

\end{proof}

The same proof as the one of the aforementioned theorem, gives us the following as well:
\begin{theorem}
Fix some $\ell \geq 1$ and let $\psi$ be a smooth solution of \eqref{eq:waveequation} satisfying \eqref{basic:as}, and assume that
$$I_{\ell} [ \psi] \neq 0.$$
 
Then for $0 \leq k \leq \ell$ we have that
\begin{equation}\label{dec:drklt1}
\int_{\mathbb{S}^2} r^{-2\ell+\eta+2k} ( \partial_r^k TP_{\geq \ell} \psi )^2 \, d\omega \leq C \frac{E_{aux-decomp-n0-k+\ell+2}}{u^{4\ell+6-2\eta-\delta}} ,
\end{equation}
and
\begin{equation}\label{dec:drklt2}
\int_{\mathbb{S}^2} r^{-2\ell+\eta+2k+1} ( \partial_r^k TP_{\geq \ell} \psi_{\ell} )^2 \, d\omega \leq C \frac{E_{aux-decomp-n0-k+\ell+2}}{u^{4\ell+5-2\eta-\delta}} ,
\end{equation}
where $C = C ( D,R,k,\delta,\ell )$, and $E_{aux-decomp-n0-k+\ell+2}$ is given in Corollary \ref{rem:n+2t}.

\end{theorem}

Finally we record one more estimate (that is not almost sharp but one power off) that will be used later.
\begin{theorem}
Fix some $\ell \geq 1$ and let $\psi$ be a smooth solution of \eqref{eq:waveequation} satisfying \eqref{basic:as}, and assume that
$$I_{\ell} [ \psi] \neq 0.$$

Then we have that
\begin{equation}\label{dec:drklh1}
\int_{\mathbb{S}^2} r^{\eta} ( \partial_r^{\ell+1} P_{\geq \ell} \psi )^2 \, d\omega \leq C \frac{E_{aux-decomp-n0-\ell}}{u^{4\ell+4-2\eta-\delta}} ,
\end{equation}
and 
\begin{equation}\label{dec:drklht1}
\int_{\mathbb{S}^2} r^{\eta} ( \partial_r^{\ell+1} T P_{\geq \ell} \psi )^2 \, d\omega \leq C \frac{E_{aux-decomp-n0-\ell}}{u^{4\ell+6-2\eta-\delta}} ,
\end{equation}
while if we assume that $I_{\ell} [ \psi ] = 0$ we have that
\begin{equation}\label{dec:drklh2}
\int_{\mathbb{S}^2} r^{\eta} ( \partial_r^{\ell+1} P_{\geq \ell} \psi )^2 \, d\omega \leq C \frac{E_{aux-decomp-0-\ell}}{u^{4\ell+6-2\eta-\delta}} ,
\end{equation}
where $E_{aux-decomp-n0-\ell}$ and $E_{aux-decomp-0-\ell}$ are given in Corollary \ref{rem:n+2t}, and where $C = C(D,R,\delta,\ell)$. 
\end{theorem}
\begin{proof}
We will only show the case of $P_{\geq \ell} \psi$ as the case of $T P_{\geq \ell} \psi$ is similar. 

Using once again the fundamental theorem of calculus we have that:
\begin{align*}
\int_{\mathbb{S}^2} & r_0^{\eta} ( \partial_{\rho}^{\ell+1} P_{\geq \ell} \psi )^2 \, d\omega  =  - \int_{r_0}^{\infty} \int_{\mathbb{S}^2} \partial_{\rho} \left( r^{\eta} ( \partial_{\rho} ^{\ell+1} P_{\geq \ell} \psi )^2 \right) \, d\omega d\rho \\ = & -\eta \int_{r_0}^{\infty}\int_{\mathbb{S}^2} r^{\eta-1} ( \partial_{\rho}^{\ell+1} P_{\geq \ell} \psi )^2 \, d\omega d\rho - 2 \int_{r_0}^{\infty} \int_{\mathbb{S}^2}r^{\eta} ( \partial_{\rho}^{\ell+1} P_{\geq \ell} \psi ) \cdot ( \partial_{\rho}^{\ell+2} P_{\geq \ell} \psi ) \, d\omega d\rho \\ \leq &   - 2 \int_{r_0}^{\infty}\int_{\mathbb{S}^2} r^{\eta} ( \partial_{\rho}^{\ell} P_{\geq \ell+1} \psi ) \cdot ( \partial_{\rho}^{\ell+2} P_{\geq \ell} \psi ) \, d\omega d\rho \\ \leq & \left( \int_{r_0}^{\infty}\int_{\mathbb{S}^2} r^{-1+\eta} ( (\partial_{\rho}^{\ell+1} P_{\geq \ell} \psi )^2 \, d\omega d\rho \right)^{1/2} \cdot \left( \int_{r_0}^{\infty}\int_{\mathbb{S}^2} r^{1+\eta} ( (\partial_{\rho}^{\ell+2} P_{\geq \ell} \psi )^2 \, d\omega d\rho \right)^{1/2} ,
\end{align*}
and now we argue as in the proof of Theorem \ref{dec:drk1}, applying Theorem \ref{thm:elh} repeatedly for $m=\ell-1$ and $m=\ell$ using the highest order term (after separating the integrals in regions close and away from the horizon), and in the end we get that
\begin{align*}
\int_{\mathbb{S}^2} r_0^{\eta} ( \partial_{\rho}^{\ell+1} P_{\geq \ell} \psi )^2 \, d\omega \leq & C \left( \sum_{m_1 + m_2 = \ell} \int_{\mathcal{S}_{\tau}} J^N [ N^{m_1} T^{m_2} P_{\geq \ell} \psi ] \cdot n_{\tau} \, d\mu_{\mathcal{S}_{\tau}} \right)^{1/2} \\ & \times \left( \sum_{m_1 + m_2 = \ell+1} \int_{\mathcal{S}_{\tau}} J^N [ N^{m_1} T^{m_2} P_{\geq \ell} \psi ] \cdot n_{\tau} \, d\mu_{\mathcal{S}_{\tau}} \right)^{1/2} ,
\end{align*}
and the result follows from the decay estimates of Proposition \ref{dec:red}.

\end{proof}
 
\section{Precise late-time asymptotics when $I_{\ell}\neq 0$}\label{asymptotics}
In this section we will derive the \emph{precise late-time asymptotics} of a frequency localized solution $P_{\ell} \psi$ of \eqref{eq:waveequation}, for some $\ell \geq 1$. In order to achieve this we will make use of the $\ell$-th Newman--Penrose charge and the almost sharp upper bounds obtained in the previous sections.

Note that in this section we will always work with a solution $\psi$ of \eqref{eq:waveequation} that is localized at angular frequency $\ell \geq 1$, so by $\psi$ we always mean $P_{\ell} \psi$.

 \subsection{Asymptotics for the radiation field}\label{subsec:rf}
 Recall the definition of $\Phi_{(k)}$ given in Proposition \ref{prop:np}. The $\ell$-th Newman--Penrose charge for $\ell \geq 1$, as we have seen before, is given by
$$ I_{\ell} [ \psi ] ( \theta , \varphi ) = \lim_{r \rightarrow \infty} \Phi_{(\ell +1 )} (u,r,\theta , \varphi ) , $$
for $\psi = P_{\ell} \psi$ a frequency localized solution of \eqref{eq:waveequation}. Note that in this Section we will always assume that our solution of \eqref{eq:waveequation} is localized at frequency $\ell$ for some $\ell \geq 1$.

We will use the properties of the $\ell$-th Newman-Penrose charge and the decay estimates from the previous sections to derive first asymptotics for $L\Phi_{(\ell)}$, then for $\Phi_{(k)}$, $k \in \{ 1 , \dots , \ell-1\}$, and finally for $\phi$ (where $\phi = r\psi$) in the following region close to $\mathcal{I}^{+}$:
$$ \mathcal{B}_{\alpha_{\ell}} = \{ r \geq R \} \cap \{ (u,v)  :  u_0 \leq u \leq v -v^{\alpha_{\ell}} \} , $$
for some large $R$, and for $\frac{2\ell +2}{2\ell + 3} < \alpha_{\ell} < 1$. Note that the boundary of $\mathcal{B}_{\alpha_{\ell}}$ is the curve 
$$\gamma_{\alpha_{\ell}} = \{ ( u,v) :  v-u = v^{\alpha_{\ell}} \}$$
 on which we have in particular that $v\sim u$. Recall also that $v-u \sim r$ for $r \geq R$.
 
We need some auxiliary lemmas. The first one is a direct consequence of Corollary \ref{cor:cancel}.

\begin{lemma}
Let $\psi$ be a solution of \eqref{eq:waveequation} that is localized at angular frequency $\ell$. We have that:
\begin{equation}\label{eq:invw}
\underline{L} \left( \frac{1}{r^{2\ell}} L \Phi_{(\ell )} \right) = O ( r^{-2-2\ell} ) ( L \Phi_{(\ell )} ) + \sum_{k=0}^{\ell} O (r^{-3-2\ell} ) \Phi_{(k )} .
\end{equation}
\end{lemma} 

The second one has to do with a preliminary upper bound on the $\ell$-th Newman--Penrose charge in the region $\mathcal{B}_{\alpha_{\ell}}$.

\begin{lemma}\label{lm:np1v}
Let $\psi$ be a solution of \eqref{eq:waveequation} that is localized at angular frequency $\ell$. Assume that:
\begin{equation}\label{eq:np1as}
r^2 \left. \partial_r \Phi_{(\ell)} \right|_{u = u_0 } = I_{\ell} [ \psi ] (\theta , \varphi ) + O ( v^{-\epsilon} ) , 
\end{equation}
for some $\epsilon > 0$, and $I_{\ell} [\psi ] \neq 0$. Then for all $(u,v) \in \mathcal{B}_{\alpha_{\ell}}$ we have that:
\begin{equation}\label{eq:np1v}
| v^2 L \Phi_{(\ell )} (u,v)| \leq C E_{aux, I_{\ell} } \left( 1 + O ( v^{-\beta} ) \right) ,
\end{equation}
for some $\beta > 0$ and for a weighted initial energy $E_{aux , I_{\ell} }$ that depends on the sum of initial energies given in Lemma \ref{dec:rp1} and $I_{\ell} [ \psi ]$.
\end{lemma}
\begin{proof}
We integrate in $u$ the quantity $\underline{L} ( v^2 L \Phi_{(\ell )} )$ for $(u,v) \in \mathcal{B}_{\alpha_{\ell}}$ and we have that:
\begin{align*}
v^2 L \Phi_{(\ell )} (u,v) = & e^{- \int_{u_0}^u ( - 4\ell r^{-1} + O (r^{-2} ) ) \, du'} v^2 L \Phi_{(\ell )} (u_0 ,v) \\ & + e^{- \int_{u_0}^u ( - 4\ell r^{-1} + O (r^{-2} ) ) \, du'} \sum_{k=0}^{\ell} \int_{u_0}^u e^{ \int_{u_0}^{u'} ( - 4\ell r^{-1} + O (r^{-2} ) ) \, du''} v^2 O (r^{-3} ) \Phi_{(k)} \, du' .
\end{align*}
Note that
$$ e^{- \int_{u_0}^u ( - 4\ell r^{-1} + O (r^{-2} ) ) \, du'} = e^{O (r^{-1} (u,v)) + O (r^{-1} (u_0 , v))} \frac{r^{8\ell} (u,v)}{r^{8\ell} (u_0 , v )} , $$
and for $(u,v) \in \mathcal{B}_{\alpha_{\ell}}$:
\begin{align*}
v^2 L \Phi_{(\ell )} (u_0 ,v) = & \frac{v^2}{2r^2 (u,v) } ( r^2 \partial_r \Phi_{(\ell )} ) (u_0 , v) \\ & - \frac{M v^2}{r (u,v) } ( r^2 \partial_r \Phi_{(\ell )} ) (u_0 , v) + \frac{e^2 v^2}{r^4 (u,v) } ( r^2 \partial_r \Phi_{(\ell )} ) (u_0 , v) \\ = & 4  I_{\ell} [\psi ] + O (v^{-\beta} ) ,
\end{align*} 
for some $\beta > 0$, as $L = \frac{1}{2} D \partial_r$ and $v = u+2r$. Then we have that:
\begin{align*}
| v^2 L \Phi_{(\ell )} | (u,v) \leq &  2  I_{\ell} [\psi ] + O (v^{-\beta} ) \\ & + C \sum_{k=0}^{\ell} \int_{u_0}^u \frac{v^2}{r^3} | \Phi_{(k)} | \, du' \\ \leq &  2  I_{\ell} [\psi ] + O (v^{-\beta} ) \\ & + C v^{-\beta} \sum_{k=0}^{\ell} \int_{u_0}^u \frac{v^{2+\beta}}{r^3} \frac{E_{aux , I_{\ell}}}{(u' )^{1+(\ell -k ) - \delta}} \, du'\\ \leq &  2 I_{\ell} [\psi ] + O (v^{-\beta} ) \\ & + C v^{-\beta} \sum_{k=0}^{\ell} \int_{u_0}^u \frac{E_{aux , I_{\ell}}}{r^{3-\frac{2+\beta}{\alpha_{\ell}}}(u' )^{1+(\ell -k ) - \delta}} \, du' \\ \leq &  2  I_{\ell} [\psi ] + O (v^{-\beta} ) + C E_{aux , I_{\ell}} v^{-\beta} ,
\end{align*}
for $C$ a constant depending on $u_0$ and $\ell$. Note that we used that $v^{\alpha_{\ell}} \lesssim r$ and $u^{\alpha_{\ell}} \lesssim r$ in $\mathcal{B}_{\alpha_{\ell}}$ and estimate \eqref{dec:pwkl}.

\end{proof} 
 
Using the previous two lemmas we get the following result:
\begin{proposition}\label{prop:auxasn0}
Fix some $\ell \geq 1$, consider a solution $\psi$ of the wave equation \eqref{eq:waveequation} that is localized at frequency $\ell$, and also assume that \eqref{eq:np1as} holds true for some $\epsilon > 0$ and $I_{\ell} [\psi ] \neq 0$. Then there exists $\eta > 0$ small enough such that for all $( u , v ) \in \mathcal{B}_{\alpha_{\ell}}$ we have that:
\begin{equation}\label{asym:vphi1}
L \Phi_{(\ell)} (u,v , \theta , \varphi ) = 2^{2\ell+1}I_{\ell} [ \psi ] ( \theta , \varphi ) \frac{(v-u)^{2\ell}}{v^{2\ell+2}}+O(v^{-\eta})\frac{(v-u)^{2\ell}}{(v-u)^{2\ell+2}}.
\end{equation}
\end{proposition}

\begin{proof}
For $(u,v) \in \mathcal{B}_{\alpha_{\ell}}$ by integrating in $u$ equation \eqref{eq:invw} we get that:
\begin{equation*}
\begin{split}
v^{2\ell+2}r^{-2\ell} L \Phi_{(\ell)}(u,v)=&\:v^{2\ell+2}r^{-2\ell}(u_0 ,v) L \Phi_{(\ell)}(u_0 ,v)+\int_{u_0}^u v^{2\ell}O(r^{-2\ell-2})(v^2 L \Phi_{(\ell)})(u',v)\,du'\\
&+\sum_{k=0}^{\ell}\int_{u_0}^u v^{2\ell+2}O(r^{-2\ell-3})\Phi_{(k)}(u',v)\,du'.
\end{split}
\end{equation*}
By \ref{lm:np1v} we can estimate
\begin{equation*}
\begin{split}
\left|\int_{u_0}^u v^{2\ell}O(r^{-2\ell-2})(v^2 L \Phi_{(\ell)})(u',v)\,du'\right|\leq & C E_{aux, I_{\ell}} v^{-\eta}\int_{u_0}^u r^{-2\ell-2+\frac{2\ell+\eta}{\alpha_{\ell}}}\,du' \\ \leq & C E_{aux , I_{\ell}} v^{-\eta} r_{\gamma_{\alpha}}(v)^{-2\ell-1+\frac{2\ell+\eta}{\alpha_{\ell}}}=O(v^{-\eta}) ,
\end{split}
\end{equation*}
if $\alpha_{\ell} >\frac{2\ell+\eta}{2\ell+1}$ for some $\eta > 0$.

Furthermore,
\begin{equation*}
\begin{split}
\left|\sum_{k=0}^{\ell}\int_{u_0}^u v^{2\ell+2}O(r^{-2\ell-3})\widetilde{\Phi}_{(k)}(u',v)\,du'\right|\leq&\: C E_{aux , I_{\ell}} v^{-\eta} \sum_{k=0}^{\ell} \int_{u_0}^u r^{-2\ell-3+ \eta / \alpha_{\ell}} ( u' )^{1+ (\ell -k) - \delta} \, du' 
\\ \leq &\: Cv^{-\eta}\int_{u_0}^u r^{-2\ell-3+\frac{2\ell+2+\eta+2\epsilon}{\alpha_{\ell}}}( u' )^{-1-\epsilon}\,du'\\
\leq&\: Cv^{-\eta} r_{\gamma_{\alpha}}(v)^{-2\ell-3+\frac{2\ell+2+\eta+2\epsilon}{\alpha_{\ell}}}=O(v^{-\eta}) ,
\end{split}
\end{equation*}
if $\alpha_{\ell} >\frac{2\ell+2+\eta+2\epsilon}{2\ell+3}$, for some $\epsilon > 0$ and $\eta > 0$. Note that $\frac{2\ell+2+\eta+2\epsilon}{2\ell+3}>\frac{2\ell+\eta}{2\ell+1}$ and that we used the decay estimates \eqref{dec:pwkl}.

Using that $L=\frac{1}{2}D\partial_r$ and $v=u+2r$, we moreover have that
\begin{equation*}
\begin{split}
v^{2\ell+2}r^{-2\ell} & (u_0 ,v)L \Phi_{(\ell)}(u_0 ,v)= \frac{1}{2} v^{2\ell+2}r^{-2\ell}(u_0 ,v) \partial_r  \Phi_{(\ell)}(u_0 ,v) \\ = &\frac{1}{2} v^{2\ell+2}r^{-2\ell}(u_0 ,v)  ( r^2 \partial_r  \Phi_{(\ell)} ) (u_0 ,v) \\ & - M v^{2\ell+2}r^{-2\ell-3}(u_0 ,v) ( r^2 \partial_r  \Phi_{(\ell)} ) (u_0 ,v) \\ & + \frac{e^2}{2} v^{2\ell+2}r^{-2\ell-4}(u_0 ,v) (r^2 \partial_r  \Phi_{(\ell)} ) (u_0 ,v)  \\ = & 2^{2\ell+1} I_{\ell} [ \psi ] (\theta , \varphi ) +O(v^{-\eta}).
\end{split}
\end{equation*}
where we used the properties of $v$ and $r$ in $\mathcal{B}_{\alpha_{\ell}}$ and the assumption \eqref{eq:np1as}.
\end{proof}

Using the last three results and equation \eqref{eq:importantvintegral} from Lemma \ref{lm:integrals} we finally get the following Proposition which provides us with \textit{precise} asymptotics for all the radiation fields in $P_{\ell} \Phi_{(k)}$, $k \in \{0, \dots , \ell\}$ in $\mathcal{B}_{\alpha_{\ell}}$.
\begin{proposition}\label{prop:mainasn0}
Fix some $\ell \geq 1$. For a solution $\psi$ of \eqref{eq:waveequation} that is localized at angular frequency $\ell$, let $k \in \{2, \dots , \ell\}$, and also assume that \eqref{eq:np1as} holds true for some $\epsilon > 0$ and $I_{\ell} [\psi ] \neq 0$. Then there exist $0<\eta_k<1$ suitably small such that for all $(u,v)\in \mathcal{B}_{\alpha_{\ell}}$:
\begin{equation*}
\Phi_{(\ell)}(u,v,\theta,\varphi)=\frac{2^{2\ell}I_{\ell} [\psi ] (\theta,\varphi)}{(2\ell+1)} \left(\frac{v-u}{v}\right)^{2\ell+1}u^{-1}+O(u^{-1-\eta_1}),
\end{equation*}
\begin{equation*}
\widetilde{\Phi}_{(\ell-k)}(u,v,\theta,\varphi)=\frac{2^{2\ell+2k}I_{\ell} [\psi ] (\theta,\varphi)}{(2\ell+1)\cdot(2\ell)\cdot\ldots \cdot (2\ell+1-k)} \left(\frac{v-u}{v}\right)^{2\ell+1-k}u^{-1-k}+O(u^{-1-k-\eta_{k}}).
\end{equation*}
In particular,
\begin{equation*}
\begin{split}
\phi(u,v,\theta,\varphi)=&\:\frac{2^{4\ell}I_{\ell} [ \psi ] (\theta,\varphi)}{(2\ell+1)\cdot \ldots \cdot ( \ell +1 )} \left(\frac{v-u}{v}\right)^{\ell+1}u^{-1-\ell}+O(u^{-1-\ell-\eta_{\ell+1}})\\
=&\:\frac{2^{4\ell} I_{\ell} [ \psi ] (\theta,\varphi)}{(2\ell+1) \cdot \ldots \cdot ( \ell +1 )} \left(\frac{1}{u}-\frac{1}{v}\right)^{\ell+1}+O(u^{-1-\ell-\eta_{\ell+1}}).
\end{split}
\end{equation*}
\end{proposition}
\begin{proof}
For all $k\in \{0, \dots , \ell \}$ recall that
\begin{equation*}
\Phi_{(k)}=\sum_{m=0}^{k}(d_{k,m}+O(r^{-1}))r^{k+1+m}\partial_r^k\psi , 
\end{equation*}
for some constants $d_{k,m}$ (that depends on the constants $\alpha_{n,k}$ given in Proposition \ref{prop:np}). By the estimates of Theorem \ref{dec:drk1} we can therefore estimate:
\begin{equation*}
\begin{split}
|\Phi_{(k)}(u,v_{\gamma_{\alpha}}(u))|\leq & C_{\epsilon} E_{aux-decomp-n0-k} u^{-2\ell-3/2+\epsilon}r^{\ell+k+1/2}\\ \leq & C_{\epsilon}E_{aux-decomp-n0-k} u^{-1-(\ell-k)+\epsilon-(k+\ell+1/2)(1-\alpha_{\ell})},
\end{split}
\end{equation*}
where $E_{aux-decomp-n0-k}$ is as in the estimates of Theorem \ref{dec:drk1}. Note that for $k=\ell$ in particular we have that:
\begin{equation*}
|\Phi_{(\ell)}(u,v_{\gamma_{\alpha}}(u))|\leq C_{\epsilon} E_{aux-decomp-n0-\ell} u^{-2\ell-3/2+\epsilon}r^{2\ell+1/2}\leq C_{\epsilon} E_{aux , I_{\ell} \neq 0 , k , \ell } u^{-1+\epsilon-(2\ell+1/2)(1-\alpha_{\ell})}.
\end{equation*}
Choose $\epsilon>0$ suitably small and $\eta_1$ such that $0<\eta_1<(2\ell+1/2)(1-\alpha_{\ell} )-\epsilon$. Using \eqref{asym:vphi1} together with \eqref{eq:importantvintegral} we find that
\begin{equation*}
\begin{split}
\Phi_{(\ell)}(u,v)=&\:\Phi_{(\ell)}(u,v_{\gamma_{\alpha_{\ell}}}(u))+\int_{v_{\gamma_{\alpha_{\ell}}}(u)}^v 2^{2\ell+1 }I_{\ell} [ \psi ] ( \theta , \varphi ) \frac{(v'-u)^{2\ell}}{v'^{2\ell+2}}(1+O(v'^{-\eta}))\,dv'\\
=&\:\frac{2^{2\ell+1}I_{\ell} [\psi ] (\theta,\varphi)}{2\ell+1} \left(\frac{v-u}{v}\right)^{2\ell+1}u^{-1}(1+O(u^{-\eta'})+O(v^{-\eta'}))+O(u^{-1-\eta_1}).
\end{split}
\end{equation*}
Hence, using the definition of $\Phi_{(\ell)}$ from Proposition \ref{prop:np}, we get asymptotics as above for $\widetilde{\Phi}_{\ell}$ that imply 
\begin{equation*}
\begin{split}
\partial_r \widetilde{\Phi}_{(\ell-1)}(u,v)=& \frac{2^{2\ell+1}I_{\ell} [\psi ] (\theta,\varphi)}{2\ell+1} \frac{1}{r^2} \left(\frac{v-u}{v}\right)^{2\ell+1}u^{-1}(1+O(u^{-\eta'})+O(v^{-\eta'}))\\ & + \frac{1}{r^2}O(u^{-1-\eta_1}) ,
\end{split}
\end{equation*}
which in turn implies that
\begin{equation*}
\begin{split}
L \widetilde{\Phi}_{(\ell-1)}(u,v) =& \frac{2^{2\ell}I_{\ell} [\psi ] (\theta,\varphi)}{2\ell+1} \frac{2}{v-u} \left(\frac{v-u}{v}\right)^{2\ell+1}u^{-1}(1+O(u^{-\eta'})+O(v^{-\eta'}))\\ & + \frac{1}{(v-u)^2}O(u^{-1-\eta_1}) \\ & + O ((v-u)^{-1} ) \Big[ \frac{2^{2\ell}I_{\ell} [\psi ] (\theta,\varphi)}{2\ell+1} \frac{2}{v-u} \left(\frac{v-u}{v}\right)^{2\ell+1}u^{-1}(1+O(u^{-\eta'})+O(v^{-\eta'})) \\ & + \frac{2}{v-u}O(u^{-1-\eta_1}) \Big] .
\end{split}
\end{equation*}
This then implies that
$$ L \widetilde{\Phi}_{(\ell-1)}(u,v) = \frac{2^{2\ell+2}I_{\ell} [\psi ] (\theta,\varphi)}{2\ell+1} \frac{(v-u)^{2\ell-1}}{v^{2\ell+1}} u^{-1} ( 1+ O(u^{-\eta'} ) + O (v^{-\eta' } ) ) + O ( u^{-2-\eta'_2} ) , $$
for some $\eta'_2 > 0$ that depends on $\alpha_{\ell}$. Integrating the above equation in $v$ from the curve $\gamma_{\alpha_{\ell}}$ we get that:
\begin{align*}
\widetilde{\Phi}_{(\ell-1)}(u,v) = & \Phi_{(\ell-1)}(u,v_{\alpha_{\ell}}) \\ & + \frac{2^{2\ell+2} I_{\ell} [\psi ] (\theta , \varphi )}{ (2\ell +1 ) ( 2 \ell )} \left( \frac{v-u}{v} \right)^{2\ell} u^{-2} + O (u^{-2-\eta''_2 } ) , 
\end{align*}
for some $\eta''_2 > 0$. Using the estimates \eqref{dec:drk1} for $k=\ell-1$ we get that:
\begin{equation*}
|\widetilde{\Phi}_{(\ell-1)}(u,v_{\gamma_{\alpha}}(u))|\leq C_{\epsilon} E_{aux-dr , \ell }u^{-2\ell-3/2+\epsilon}r^{2\ell-1+1/2}\leq C_{\epsilon} E_{aux-dr , \ell }u^{-2+\epsilon-(2\ell-1/2)(1-\alpha)}
\end{equation*}
which then implies that:
$$ \widetilde{\Phi}_{(\ell-1)}(u,v) = \frac{2^{2\ell+2} I_{\ell} [\psi ] (\theta , \varphi )}{ (2\ell +1 ) ( 2 \ell )} \left( \frac{v-u}{v} \right)^{2\ell} u^{-2} + O (u^{-2-\eta_2 } ) , $$
for some $\eta_2 > 0$. We get asymptotics for all the remaining $\widetilde{\Phi}_{(\ell-k)}$'s by repeating the process all the way to $k=\ell$.

\end{proof}
An immediate consequence of the previous Proposition is to obtain asymptotics for $\psi$ in $\mathcal{B}_{\alpha_{\ell}}$, now starting from asymptotics for $\psi_{\ell} / r^{\ell}$ and using this to build the asymptotics of all higher derivatives (hence following the reverse path compared to Proposition 
\ref{prop:mainasn0}). In particular for any $m \in \mathbb{N}$ as $r \sim v-u$ in $\mathcal{B}_{\alpha_{\ell}}$, there exist $0<\bar{\eta}_k<1$, $k \in \{1, \dots , \ell\}$ suitably small such that for all $(u,v)\in \mathcal{B}_{\alpha_{\ell}}$:
\begin{equation}\label{prop:asb}
\frac{\psi}{r^{\ell}} (u,v,\theta,\varphi)=\:\frac{2^{5\ell+1}I_{\ell} [ \psi ] (\theta,\varphi)}{(2\ell+1)\cdot \ldots \cdot ( \ell +1 )} u^{-1-\ell} v^{-1-\ell} +O(u^{-2-2\ell-\bar{\eta}_{\ell}}),
\end{equation}
and
\begin{equation}\label{prop:asbh}
\frac{\partial_r^{\ell-k} \psi}{r^{k}} (u,v,\theta,\varphi)=\:\widetilde{A}_{\ell , k} I_{\ell} [ \psi ]  u^{-1-k} v^{-1-2\ell+k} +O(u^{-2-2\ell-\bar{\eta}_{\ell-k}}),
\end{equation}
for constants $\widetilde{A}_{\ell , k}$ depending on $\ell$ and $k$.

\subsection{Global asymptotics for $P_{\ell} \psi$}\label{global_as}
The goal of this section is to propagate the asymptotics of a linear wave $P_{\ell} \psi$ obtained in $\mathcal{B}_{{\alpha}_{\ell}}$ to the rest of the spacetime.
\begin{lemma}
Fix $\ell \geq 1$. Let $\psi_{\ell} \doteq P_{\ell} \psi$ be a frequency localized linear wave by assuming that $\psi$ solves \eqref{eq:waveequation}, then for all $k\geq 0$:
\begin{equation}
\label{eq:maincommwaveq}
\begin{split}
\partial_{\rho}((Dr^2)^{k+1}\partial_r^{k+1}\psi_{\ell} )=&(Dr^2)^k[-\slashed{\Delta}_{\s^2}-k(k+1)]\partial_r^{k} \psi_{\ell} \\ & +(Dr^2)^{k} r \cdot \sum_{j=0}^{k} O(r^{-k+j})\partial_r^jT\psi_{\ell}+(Dr^2)^{k} O(r^{1-\eta}) \partial_r^{k+1}T\psi_{\ell} .
\end{split}
\end{equation}
\end{lemma}
\begin{proof}
We will prove \eqref{eq:maincommwaveq} by induction. By \eqref{eq:inhomelliptic}, we have that \eqref{eq:maincommwaveq} holds for $k=0$.  Now suppose \eqref{eq:maincommwaveq} holds for $k\in \N_0$, then we will show that \eqref{eq:maincommwaveq} also holds for $k$ replaced by $k+1$.

We have by applying the Leibniz rule multiple times that
\begin{equation*}
\begin{split}
(Dr^2)^{-(k+1)}\partial_{\rho}((Dr^2)^{k+2}\partial_r^{k+2}\psi_{\ell} )= &\:(Dr^2)^{-k}\partial_{\rho}((Dr^2)^{k+1}\partial_r^{k+2}\psi_{\ell} )+(Dr^2)'\partial_r^{k+2}\psi_{\ell} \\
=&\: (Dr^2)^{-k}\partial_{\rho}\partial_r((Dr^2)^{k+1}\partial_r^{k+1}\psi_{\ell} )\\ & -(k+1)(Dr^2)^{-k}\partial_{\rho}((Dr^2)^{k}(Dr^2)'\partial_r^{k+1}\psi_{\ell} )+(Dr^2)'\partial_r^{k+2}\psi_{\ell} \\
=&\: (Dr^2)^{-k}\partial_r\partial_{\rho}((Dr^2)^{k+1}\partial_r^{k+1}\psi_{\ell} )-(Dr^2)h' \partial_r^{k+1}T\psi\\
&-(k+1)(Dr^2)'(Dr^2)^{-k}\partial_{\rho}((Dr^2)^{k}\partial_r^{k+1}\psi_{\ell} ) \\ & -(k+1)(Dr^2)'' \partial_r^{k+1}\psi_{\ell} +(Dr^2)'\partial_r^{k+2}\psi_{\ell} ,
\end{split}
\end{equation*}
where in the last inequality we used that $\partial_{\rho}\partial_r=\partial_r\partial_{\rho}-h'T$. We can further expand the right-hand side above to obtain
\begin{equation*}
\begin{split}
(Dr^2)^{-(k+1)}\partial_{\rho}((Dr^2)^{k+2}\partial_r^{k+2}\psi_{\ell} )= &\:\partial_r\left[(Dr^2)^{-k}\partial_{\rho}((Dr^2)^{k+1}\partial_r^{k+1}\psi_{\ell} )\right]\\ & +k(Dr^2)'(Dr^2)^{-(k+1)}\partial_{\rho}((Dr^2)^{k+1} \partial_r^{k+1}\psi_{\ell} )\\
&-(Dr^2)h' \partial_r^{k+1}T\psi-(k+1)(Dr^2)'(Dr^2)^{-k}\partial_{\rho}((Dr^2)^{k}\partial_r^{k+1}\psi_{\ell} )\\
&-(k+1)(Dr^2)'' \partial_r^{k+1}\psi+(Dr^2)'\partial_r^{k+2}\psi_{\ell} \\
=&\:\partial_r\left[(Dr^2)^{-k}\partial_{\rho}((Dr^2)^{k+1}\partial_r^{k+1}\psi_{\ell} )\right]+k(k+1)((Dr^2)')^2(Dr^2)^{-1}\partial_r^{k+1}\psi_{\ell} \\
&+k(Dr^2)'\partial_{\rho}\partial_r^{k+1}\psi-(Dr^2)h' \partial_r^{k+1}T\psi_{\ell} \\
&-k(k+1)((Dr^2)')^2(Dr^2)^{-1}\partial_r^{k+1}\psi_{\ell} -(k+1)(Dr^2)'\partial_{\rho}\partial_r^{k+1}\psi_{\ell} \\
&-(k+1)(Dr^2)'' \partial_r^{k+1}\psi+(Dr^2)'\partial_r^{k+2}\psi_{\ell} \\
=&\:\partial_r\left[(Dr^2)^{-k}\partial_{\rho}((Dr^2)^{k+1}\partial_r^{k+1}\psi_{\ell} )\right]-(k+1)(Dr^2)'' \partial_r^{k+1}\psi_{\ell} \\
&-[(Dr^2)h' +h(Dr^2)']\partial_r^{k+1}T\psi_{\ell} .
\end{split}
\end{equation*}
By applying $\partial_r$ to both sides of \eqref{eq:maincommwaveq} \textbf{and using that $(Dr^2)''=2$} and $\frac{2}{D}-h=\frac{c}(r^{1+\eta}) + O (r^{-2} )$, for $c \neq 0$ and some $\eta > 0$, we can rewrite the above equation as follows:
\begin{equation*}
\begin{split}
(Dr^2)^{-(k+1)}\partial_{\rho}((Dr^2)^{k+2}\partial_r^{k+2}\psi_{\ell} )= &\: \left[-\slashed{\Delta}_{\s^2}-k(k+1)-2(k+1)\right]\partial_r^{k+1}\psi_{\ell} \\
&+r \sum_{j=0}^{k+1} O(r^{-k-1+j})\partial_r^j T\psi_{\ell} +O(r^{1-\eta})\partial_r^{k+2}T\psi_{\ell} \\
= &\:\left[-\slashed{\Delta}_{\s^2}-(k+1)(k+2)\right]\partial_r^{k+1}\psi_{\ell} \\
&+r \sum_{j=0}^{k+1} O(r^{-(k+1)+j})\partial_r^j T\psi_{\ell} +O(r^{1-\eta})\partial_r^{k+2}T\psi_{\ell} .
\end{split}
\end{equation*}
\end{proof}
We use the above equation to obtain almost sharp decay for $\partial_r^{\ell+1} \psi_{\ell}$.

\begin{proposition}
Fix $\ell \geq 1$. Let $\psi_{\ell} \doteq P_{\ell} \psi$ be a frequency localized linear wave by assuming that $\psi$ solves \eqref{eq:waveequation}. Assume also that
$$ I_{\ell} [ \psi ] \neq 0 . $$
Then we have that
\begin{equation}
\label{eq:imporveddecayrder}
|\partial_r^{\ell+1}\psi_{\ell}|\leq C E_{aux-decomp-n0-\ell+1}  (1+\tau)^{-2\ell-3+\epsilon},
\end{equation}
where $C = C (D, R , \ell)$ and $E_{aux-decomp-n0-\ell+1}$ is as in \eqref{dec:drklt1}.
\end{proposition}
\begin{proof}
If we apply \eqref{eq:maincommwaveq} to $\psi_{\ell}$, with $k=\ell$, we obtain
\begin{equation}
\label{eq:commwaveqfixedl}
\partial_{\rho}((Dr^2)^{\ell+1}\partial_r^{\ell+1}\psi_{\ell})=(Dr^2)^{\ell} r \cdot \sum_{j=0}^{\ell} O(r^{-\ell+j})\partial_r^jT\psi_{\ell}+(Dr^2)^{\ell} O(r^{1-\eta}) \partial_r^{\ell+1}T\psi.
\end{equation}
Now, we integrate \eqref{eq:commwaveqfixedl} in $\rho$ from $\rho=r_+$ and use Theorem \ref{dec:drk1} to obtain
\begin{equation*}
\begin{split}
(Dr^2)^{\ell+1}|\partial_r^{\ell+1}\psi_{\ell}|\lesssim&\:  \sum_{j=0}^{\ell}  \int_{r_+}^{\rho}(Dr^2)^{\ell} r^{-\ell+j+1}|\partial_r^jT\psi_{\ell}|\,d\rho'\\
&+\int_{r_+}^{\rho}(Dr^2)^{\ell} r^{1-\eta}|\partial_r^{\ell+1}T\psi_{\ell}|\,d\rho'\\
\lesssim &\: \int_{r_+}^{\rho}r\cdot  (Dr^2)^{\ell} \,d\rho'\cdot \left(||\partial_r^{\ell+1}T\psi_{\ell}||_{L^{\infty}}+\sum_{j=0}^{\ell} ||r^{-\ell+j}\partial^k_rT\psi_{\ell}||_{L^{\infty}}\right)\\
\leq  C E_{aux-decomp-n0-\ell+1}  &\:(Dr^2)^{\ell+1}\cdot (1+\tau)^{-2\ell-3+\epsilon} ,
\end{split}
\end{equation*}
where we used estimates \eqref{dec:drklt1}. We can therefore conclude that
\begin{equation*}
|\partial_r^{\ell+1}\psi_{\ell}|\leq  C E_{aux-decomp-n0-\ell+1}  (1+\tau)^{-2\ell-3+\epsilon}.
\end{equation*}
\end{proof}

The above proposition is the key to propagating the asymptotics of $\psi_{\ell}$ from the curve $\rho=\rho_{\gamma_{\alpha}}(\tau)$ to the region $r_+\leq \rho \leq \rho_{\gamma_{\alpha_{\ell}}}(\tau)$, where the curve $\gamma_{\alpha_{\ell}}$ is defined as in the previous subsection. 

\begin{proposition}
Fix $\ell \geq 1$. Let $\psi_{\ell} \doteq P_{\ell} \psi$ be a frequency localized linear wave by assuming that $\psi$ solves \eqref{eq:waveequation}. Assume also that
$$ I_{\ell} [ \psi] \neq 0 . $$
For all $0\leq k\leq \ell$ and for all $\rho\leq \rho_{\gamma_{\alpha_{\ell}}}(\tau)$, where $\gamma_{\alpha_{\ell}}$ and $\alpha_{\ell}$ are defined in subsection \ref{subsec:rf}, we have that
\begin{equation*}
r^{k-\ell}\partial_{r}^{k}\psi_{\ell}(\tau,\rho)=A_{\ell,k}I_{\ell}(1+\tau)^{-2\ell-2}+O((1+\tau)^{-2\ell-2-\epsilon}),
\end{equation*}
for some $\epsilon > 0$, where the quantities $A_{\ell , k}$ are given by Proposition \ref{prop:asb}. 
\end{proposition}
\begin{proof}
By \eqref{prop:asbh}, we have that there exists a constant $A_{\ell , \ell}$ (depending on $\widetilde{A}_{\ell,\ell}$), such that
\begin{equation*}
\partial_{r}^{\ell}\psi_{\ell}(\tau,\rho_{\gamma_{\alpha}}(\tau))=A_{\ell , \ell}I_{\ell} [ \psi_{\ell} ] (1+\tau)^{-2\ell-2}+O((1+\tau)^{-2\ell-2-\epsilon}),
\end{equation*}
where we note that we suppress the dependence on $(\theta , \varphi )$ for $I_{\ell}$. We can then integrate in $\rho$ and use \eqref{eq:imporveddecayrder} together with estimates \eqref{dec:drklt1} to obtain for $r_+\leq \rho \leq \rho_{\gamma_{\alpha_{\ell}}}(\tau)$ and $\epsilon>0$ arbitrarily small the following:
\begin{equation*}
\begin{split}
\left|\partial_{r}^{\ell}\psi_{\ell}(\tau,\rho)-A_{\ell}I_{\ell}(1+\tau)^{-2\ell-2+\epsilon}\right|\leq&\: \int_{\rho}^{\rho_{\gamma_{\alpha_{\ell}}}(\tau)}|\partial_{\rho} \partial_r^{\ell}\psi_{\ell}|\,d\rho'\\
\lesssim&\:\tau^{-2\ell-3+\epsilon}\cdot \rho_{\gamma_{\alpha}}(\tau)\\
\lesssim &\:\tau^{-2\ell-3+\alpha_{\ell}+\epsilon}=\tau^{-2\ell-2-\epsilon'},
\end{split}
\end{equation*}
where $\epsilon'=1-\epsilon-\alpha_{\ell}>0$. 

For $\partial_r^{\ell-1} \psi_{\ell}$ we integrate in $r$ from $\rho$ to $\rho_{\gamma_{\alpha_{\ell}}}$ and we have that:
\begin{align*}
\partial_r^{\ell-1} \psi_{\ell} ( \tau , \rho ) = & \partial_r^{\ell-1} \psi_{\ell} ( \tau , \rho_{\gamma_{\alpha_{\ell}}} ) - \int_{\rho}^{\rho_{\gamma_{\alpha_{\ell}}}} \partial_{\rho'} ( \partial_r^{\ell-1} \psi_{\ell} ) \, d\rho'  \\ = & \partial_r^{\ell-1} \psi_{\ell} ( \tau ,  \rho_{\gamma_{\alpha_{\ell}}} ) - \int_{\rho}^{\rho_{\gamma_{\alpha_{\ell}}}} \partial_r^{\ell} \psi_{\ell}  \, dr \\ = & \rho_{\gamma_{\alpha_{\ell}}} \frac{\partial_r^{\ell-1} \psi_{\ell} ( \tau , \rho_{\gamma_{\alpha_{\ell}}} )}{\rho_{\gamma_{\alpha_{\ell}}}} - \rho_{\gamma_{\alpha_{\ell}}} [ A_{\ell, \ell} I_{\ell} [ \psi_{\ell} ] (1+\tau)^{-2\ell-2}+O((1+\tau)^{-2\ell-2-\epsilon}) ] \\ & + \rho  [ A_{\ell , \ell}I_{\ell} [ \psi_{\ell} ] (1+\tau)^{-2\ell-2}+O((1+\tau)^{-2\ell-2-\epsilon}) ]  ,
\end{align*} 
and notice that there is an \textit{exact} cancellation in the term
$$ \rho_{\gamma_{\alpha_{\ell}}} \frac{\partial_r^{\ell-1} \psi_{\ell} ( \rho_{\gamma_{\alpha_{\ell}}} )}{\rho_{\gamma_{\alpha_{\ell}}}} - \rho_{\gamma_{\alpha_{\ell}}} [ A_{\ell, \ell}I_{\ell} [ \psi_{\ell} ] (1+\tau)^{-2\ell-2}+O((1+\tau)^{-2\ell-2-\epsilon}) ] = 0 . $$
For the remaining derivatives we work in the same way and we notice that there is a cancellation of the form:
$$ \frac{A_{\ell , k+1}}{\ell-k} = A_{\ell , k} \mbox{  for $0 \leq k \leq \ell-1$, } $$
by the relations between the quantities $\widetilde{A}_{\ell,k}$ from \eqref{prop:asbh}.

\end{proof}

\section{Construction of time integrals}
\label{timeinverse}
For $\psi$ a solution of \eqref{eq:waveequation}, let again
$$ \psi_{\ell} \doteq P_{\ell} \psi , $$ for a fixed $\ell \geq 1$. In this section we will invert the time translation operator $T$.

\begin{theorem}
Fix some $\ell \geq 1$ and consider a frequency localized solution $\psi_{\ell} \doteq P_{\ell} \psi$ of \eqref{eq:waveequation}, where $\psi$ solves \eqref{eq:waveequation}. Assume that $\left. ( \psi_{\ell} , n_{\Sigma_0} \psi_{\ell} ) \right|_{\Sigma_0} \in ( C^{\infty}_c (\Sigma_0 ) )^2$.

Then there exists a unique frequency localized (at frequency $\ell$) smooth solution of \eqref{eq:waveequation} $\widetilde{\psi}_{\ell}$ such that
$$ T \widetilde{\psi}_{\ell} = \psi_{\ell} , $$
satisfying the boundary conditions $\lim_{r\rightarrow \infty} r \partial_{\rho}^{\ell+1} \widetilde{\psi}_{\ell} = (Dr^2 )^{\ell+1} \partial_{\rho}^{\ell+1} \widetilde{\psi}_{\ell} |_{\mathcal{H}^+} = 0$.

Moreover we have that
$$ I_{\ell}^{(1)} [ \psi ] (\theta , \varphi ) \doteq I_{\ell} [ \widetilde{\psi} ] ( \theta , \varphi ) = \lim_{r \rightarrow \infty}  \left[ - r^2 \partial_r \mathcal{L}_{(\ell)} ( \psi_{\ell} ) \right] , $$
for $\mathcal{L}_{(\ell)} ( \psi_{\ell} )$ the solution of the system
\begin{align*}
\mathcal{L}_{(0)} (\psi_{\ell} ) :=& r \cdot \int_r^{\infty} ( r' -r )^{\ell} \frac{1}{(D(r' ) r'^2)^{\ell+1}} F [ \psi_{\ell} ] (r' ) \, dr' \:,\\
\mathcal{L}_{(n)} (\psi_{\ell} ) :=&\:\widetilde{\mathcal{L}}_{(n)} ( \psi_{\ell} ) + \sum_{k=1}^{n} \alpha_{n,k}\mathcal{L}_{(n-1)} ( \psi_{\ell} )  ,
\end{align*}
where the $\alpha_{n,k}$'s are as in Proposition \ref{prop:np}, where
$$ F [ \psi_{\ell} ] (r) \doteq \int_{r_{+}}^r (D(r') r'^2 )^{\ell} f_1 (r' ) \partial_{r'}^{\ell+1} \psi_{\ell} \, dr' + \sum_{j=0}^{\ell} \int_{r_{+}}^r  (D (r' )r'^2 )^{\ell} r' f_2^j (r' ) \partial_{r'}^j \psi_{\ell} \, dr' , $$
for $f_1$ and $f_2^j$ are given by the equation:
\begin{equation}\label{eq:auxinv}
 \partial_{\rho} ( (Dr^2 )^{\ell +1} \partial_{\rho}^{\ell+1} \widetilde{\psi}_{\ell} ) =  (Dr^2 )^{\ell} f_1 (r) \partial_{\rho}^{\ell+1} T \widetilde{\psi}_{\ell}  + \sum_{j=0}^{\ell} f_2^j (r) (Dr^2 )^{\ell} r \partial_{\rho}^j T \widetilde{\psi}_{\ell}  ,
 \end{equation}
and where
$$ \widetilde{\mathcal{L}}_{(k)} (\psi_{\ell} ) \doteq (-1)^{k} (r^2 \partial_r )^{k} \left(  \ell ! r \cdot \int_r^{\infty} ( r' -r )^{\ell} \frac{1}{(D(r' ) r'^2)^{\ell+1}} F[\psi_{\ell} ](r' ) \, dr' \right)  .$$

\end{theorem}
\begin{remark}
Note that the constants in the iterative system for $\widetilde{\mathcal{L}}_k$ are the same as in the system for $\Phi_{(k)}$ in Proposition \ref{prop:np}.
\end{remark}
\begin{proof}
The existence and uniqueness of the time-inverse $\psi_{\ell}$ follows by solving equation \eqref{eq:auxinv} with appropriate conditions at infinity and at the horizon (taking $\lim_{r\rightarrow \infty} r \partial_{\rho}^{\ell+1} \widetilde{\psi}_{\ell} = (Dr^2 )^{\ell+1} \partial_{\rho}^{\ell+1} \widetilde{\psi}_{\ell} |_{\mathcal{H}^+} = 0$ is enough) on the initial hypersurface $\Sigma_{u_0}$. With the initial data obtained from the solutions of the previous ODE problem, we can now have a global solution $\psi_{\ell}$ in all of the domain of outer communications up to and including the event horizon.

We integrate equation \eqref{eq:auxinv} (which in schematic form is \eqref{eq:commwaveqfixedl}) in $\rho$ and we have that for any $r \geq r_{+}$:
\begin{align*}
\partial_{\rho}^{\ell+1} \widetilde{\psi}_{\ell} ( r ) = & \frac{1}{ (D(r) r^2 )^{\ell+1}} \int_{r_{+}}^r (D(r') r'^2 )^{\ell} f_1 (r')  \partial_{r'}^{\ell+1} T \widetilde{\psi}_{\ell} \, d\rho \\ & + \frac{1}{ (D(r) r^2 )^{\ell+1}}  \sum_{j=0}^{\ell} \int_{r_{+}}^r f_2^j (r') (D (r' )r'^2 )^{\ell} r' \partial_{r'}^j T \widetilde{\psi}_{\ell} \, d\rho.
\end{align*} 
For $F$ as in the statement of the theorem, we have that
$$ f_1 (r) \sim O ( r^{1-\eta} ) , \quad \quad f_2^j (r) \sim O (r'^{-\ell+j} ) , $$
hence
$$ F [\psi_{\ell} ] (r) \sim \int_{r_{+}}^r (D(r') r'^2 )^{\ell} O (r'^{1-\eta} ) \partial_{r'}^{\ell+1} T \widetilde{\psi}_{\ell} \, d\rho + \sum_{j=0}^{\ell} \int_{r_{+}}^r O (r'^{-\ell +j} ) (D (r' )r'^2 )^{\ell} r' \partial_{r'}^j T \widetilde{\psi}_{\ell} d\rho, $$
For $\partial_{\rho}^{\ell+1}$ we now have that
\begin{equation}\label{eq:auxdrell}
\partial_{\rho}^{\ell+1} \psi_{\ell} ( r ) = \frac{1}{ (D(r) r^2 )^{\ell+1}} F[\widetilde{\psi}_{\ell} ](r) . 
\end{equation} 
Note that generically we have that
$$ F[\psi_{\ell} ](r) \sim r^{\ell+1} , $$
we have that
$$ \frac{1}{ ( D(r) r^2 )^{\ell+1}} F[\widetilde{\psi}_{\ell} ](r) \sim r^{-\ell-1} , $$
so after integrating equation \eqref{eq:auxdrell} $\ell$ times  and applying Lemma \ref{lm:integralsandweights} we get that:
$$ \partial_{\rho} \widetilde{\psi}_{\ell} (r_{\ell-1} ) = (-1)^{\ell} (\ell-1)! \int_{r_{\ell-1}}^{\infty} ( r - r_{\ell-1} )^{\ell-1} \frac{1}{ ( D(r) r^2 )^{\ell+1}} F[\psi_{\ell} ](r) \, d\rho , $$
for any $r_{\ell-1} \geq r_{+}$. We use the assumption of compact at this point to make sure the large integral is finite. Then integrating one more time we get that:
\begin{align}\label{aux:ti}
\widetilde{\psi}_{\ell} (r_{\ell} ) = & (-1)^{\ell+1} \ell ! \int_{r_{\ell}}^{\infty} ( r - r_{\ell} )^{\ell} \frac{1}{( D(r) r^2 )^{\ell+1}} F[\widetilde{\psi}_{\ell} ](r) \, d\rho \\ & + (-1)^{\ell} ( \ell -1) \lim_{r \rightarrow \infty} r^{\ell+1} \left( \frac{1}{( D(r) r^2 )^{\ell+1}} F[\psi_{\ell} ](r) \right) .
\end{align}
Note also that $F$ depends on $T\widetilde{\psi}_{\ell}$. At this point we use the assumption that 
$$ T\widetilde{ \psi}_{\ell} = \psi_{\ell} , $$
and the decay assumption on $\psi_{\ell}$. The rest follows by an application of Proposition \ref{prop:np} on $\psi_{\ell}$ defined as above.

Note also that from the formulas above we have that
$$ \bar{\Phi}_{(k)} = O ( r^{-\ell+k+\eta'} ) \mbox{  for any $\eta' >0$ and $k \in \{ 0 , \ldots , \ell - 1\}$} , $$
for $\bar{\Phi}$ defined as $\Phi_{(k)}$ but with $\widetilde{\psi}$ in the place of $\psi$. This implies that
$$ \lim_{r \rightarrow \infty}\bar{ \Phi}_{(k)} = 0 \mbox{  for $k \in \{0, \ldots , \ell - 1 \}$} . $$

\end{proof}
\begin{remark}
Note that for
\begin{align*}
F[\psi_{\ell}](r)& = \sum_{j=0}^{\ell} \int_{r_{+}}^r ( D(r') r'^2 )^{\ell} r' O (r^{-\ell+j} ) \partial_{r'}^j \psi_{\ell} \, d\rho + \int_{r_{+}}^r  (D(r' ) r'^2 )^{\ell} O (r'^{1-\eta} ) \partial_{r'}^{\ell+1} \psi_{\ell} \, d\rho\\ & \doteq F_1 [\psi_{\ell} ] (r) + F_2 [\psi_{\ell} ](r) , 
\end{align*}
we have that
$$ F_1 [\psi_{\ell} ](r) = O ( \log r ), \quad F_2 [\psi_{\ell} ](r) = O ( r^{-\eta} ) . $$
We note that in the case of non-compactly data for $\psi_{\ell}$ with vanishing $\ell$-th Newman--Penrose charget $I_{\ell} [\widetilde{\psi} ] = 0$ satisfying $\lim_{r \rightarrow \infty} r^3 \partial_r [ (r^2 \partial_r )^{\ell} ( r \psi_{\ell} ) ] < \infty$, we need to use the precise form of $F$ in order to express $I_{\ell} [\widetilde{\psi} ]$ in terms of $\psi_{\ell}$.
\end{remark}
\begin{remark}
Let us note that due to the last observation in the proof above, once we take the limits in the iterative expression for the $\mathcal{L}_n$'s to compute the constants at infinity, only two terms will be non-zero (the ones at the top two highest derivative levels). One can show that the $\alpha_{\ell , \ell-1}$'s depends only on $M$, we note that at least in the case of compactly supported data, the time-inverted Newman--Penrose charge $I_{\ell}^{(1)}$ depend only on the mass $M$ and do not depend on the charge $e$ of the Reissner--Nordstr\"{o}m spacetime that we are working in.

Moreover, starting from the formula \eqref{aux:ti} and computing 
$$\lim_{r \rightarrow \infty} \mathcal{L}_{(\ell)} (\widetilde{ \psi}_{\ell} ) = \lim_{r \rightarrow \infty} [ T^{-1} ( \mathcal{L}_{(\ell)} ( \psi_{\ell} ) ] ,$$ 
we note that
$$ I^{(1)}_{\ell} [ \widetilde{\psi}_{\ell} ] = c(\ell , M) \lim_{r \rightarrow \infty} [ T^{-1} ( \mathcal{L}_{(\ell)} ( \psi_{\ell} ) ) ] , $$
for a function $c$ that depends only on the mass $M$ and the frequency $\ell$, and we can then express the time-inverted Newman--Penrose charge as an integral over future null infinity:
$$  I^{(1)}_{\ell} [ \widetilde{\psi}_{\ell} ] (\theta , \varphi ) = c(\ell , M) \int_{u_0}^{\infty}  \left. \mathcal{L}_{(\ell)} (\psi_{\ell} ) (\tau , \theta , \varphi ) \right|_{\mathcal{I}^+} \, d\tau  . $$

\end{remark}

\section{Precise late-time asymptotics when $I_{\ell}=0$}
Using the time-inversion construction of the previous Section we can now prove asymptotics for a frequency localized linear wave. First we derive asymptotics for $T$ derivatives. We start with a basic Lemma that is a generalization of Lemma \ref{lm:np1v}.

Note that again in this section we will always work with a solution $\psi$ of \eqref{eq:waveequation} that is localized at angular frequency $\ell \geq 1$, so by $\psi$ we always mean $P_{\ell} \psi$.

\begin{lemma}\label{lm:np1vm}
Let $\psi$ be a solution of \eqref{eq:waveequation} that is localized at angular frequency $\ell \geq 1$. Under the assumption \eqref{eq:np1as} (which we make for some $\epsilon > 0$) and $I_{\ell} [\psi ] \neq 0$ for all $(u,v) \in \mathcal{B}_{\alpha_{\ell}}$ we have that:
\begin{equation}\label{eq:np1vm}
\sum_{k=0}^m | v^{2+k} L^{k+1} \Phi_{(\ell )} (u,v)| \leq C E_{aux-decomp-n0-k} \left( 1 + O ( v^{-\beta} ) \right) ,
\end{equation}
for some $\beta > 0$, for $C = C ( D,R)$, and for $E_{aux-decomp-n0-k}$ as in Theorem \ref{dec:drk1} .
\end{lemma}
The proof is similar to that of Lemma \ref{lm:np1v} where we use equation \eqref{eq:maineqNpquantfixedlv3}. Next we prove a generalization of Proposition \ref{prop:auxasn0}.

\begin{proposition}\label{prop:auxasnt}
Let $\psi$ be a solution of \eqref{eq:waveequation} that is localized at angular frequency $\ell \geq 1$. Assume that \eqref{eq:np1as} holds true for some $\epsilon > 0$ and $I_{\ell} [\psi ] \neq 0$. Then for $m \in \mathbb{N}$ there exists $\eta > 0$ small enough such that for all $( u , v ) \in \mathcal{B}_{\alpha_{\ell}}$ we have that:
\begin{equation}\label{asym:vphim}
L^{m+1} \Phi_{(\ell)} (u,v , \theta , \varphi ) = 2^{2\ell+1}I_{\ell} [ \psi ] ( \theta , \varphi )(v-u)^{2\ell} L^m \left( \frac{1}{v^{2\ell+2}} \right) +O(v^{-\eta}) (v-u)^{2\ell} L^m \left( \frac{1}{(v-u)^{2\ell+2}} \right) .
\end{equation}
\end{proposition}
The proof follows the one of Proposition \ref{prop:auxasn0} using the equation 
\begin{align*}
\underline{L} \left( \frac{1}{r^{2\ell}} L^{m+1} \Phi_{(\ell )} \right) = & O ( r^{-2-2\ell} ) ( L^{m+1} \Phi_{(\ell )} ) \\ & + \sum_{m_1 + m_2 = m , m_2 < m} O (r^{-1-m_1-2\ell} ) L^{m_2 + 1} ( P_n \Phi_{(n)} )  \\ &  + \sum_{k=0}^{n-1} \sum_{m_1 + m_2 = m} O (r^{-3-m_1-2\ell} ) L^{m_2} ( P_n \Phi_{(k)} ) , 
\end{align*}
and Lemma \ref{lm:np1vm}. From the previous Proposition, and as 
$$ L  ( T^m \Phi_{(\ell)} ) = L^{m+1} \Phi_{(\ell )} + \sum O (r^{-1-m_1} ) L^{m_2} ( T^k \Phi_{(\ell)} ) , $$
working as in the proof of Proposition \ref{prop:mainasn0}, we get that:
\begin{proposition}\label{prop:mainasnt}
Let $\psi$ be a solution of \eqref{eq:waveequation} that is localized at angular frequency $\ell \geq 1$. Assume that \eqref{eq:np1as} holds true for some $\epsilon > 0$ and $I_{\ell} [\psi ] \neq 0$. Then for any $m \in \mathbb{N}$ there exist $0<\eta_k<1$ suitably small such that for all $(u,v)\in B_{\alpha_{\ell}}$:
\begin{equation*}
T^m \phi(u,v,\theta,\varphi)=\:\frac{2^{4\ell}I_{\ell} [ \psi ] (\theta,\varphi)}{(2\ell+1)\cdot \ldots \cdot ( \ell +1 )} ( v-u)^{\ell+1} T^m ( u^{-1-\ell} v^{-1-\ell} )+O(u^{-1-\ell-m-\eta_{\ell+1}}).
\end{equation*}
\end{proposition}
As in the case of $I_{\ell} \neq 0$, an immediate consequence of the previous Proposition is to obtain asymptotics for $\psi$ in $\mathcal{B}_{\alpha_{\ell}}$.
\begin{proposition}
Let $\psi$ be a solution of \eqref{eq:waveequation} that is localized at angular frequency $\ell \geq 1$. Assume that \eqref{eq:np1as} holds true for some $\epsilon > 0$ and $I_{\ell} [\psi ] \neq 0$. Then for any $m \in \mathbb{N}$ there exist $0<\eta_k<1$ suitably small such that for all $(u,v)\in B_{\alpha_{\ell}}$:
\begin{equation*}
T^m \left( \frac{\partial_r^{\ell-k} \psi (\tau , r , \theta , \varphi )}{r^{k}} \right) =\:\frac{2^{5\ell}I_{\ell} [ \psi ] (\theta,\varphi)}{(2\ell+1)\cdot \ldots \cdot ( \ell +1 )} T^m ( u^{-1-k} v^{-1-\ell+k} )+O(u^{-2-2\ell-m-\eta_{\ell+1}}).
\end{equation*}
\end{proposition}

Moreover we can also work as in Section \ref{global_as} and obtain precise asymptotics for $T^m \psi_{\ell}$ everywhere in $\mathcal{R}$.

\begin{proposition}
Let $\psi$ be a solution of \eqref{eq:waveequation} that is localized at angular frequency $\ell \geq 1$. Assume also that $I_{\ell} [\psi_{\ell} ] \neq 0$. Then we have for all $\tau \geq u_0$ and any $r_+ \leq r < \infty$ that
\begin{equation}\label{asym:tpsi}
T^m \left( \frac{\partial_r^k \psi_{\ell} ( \tau , r )}{r^{\ell-k}} \right) = A_{\ell , k , m} \tau^{-2\ell-2-m} + O ( \tau^{-2\ell-2-m-\epsilon} ) ,
\end{equation}
for some $\epsilon > 0$.
\end{proposition}
 
Using the construction of the time-inverted Newman--Penrose charges of Section \ref{timeinverse} and the previous Proposition of the current section for $m=1$, we can obtain precise asymptotics for frequency localized waves $P_{\ell} \psi$ with vanishing $\ell$-th Newman--Penrose constant. We have the following result:

\begin{proposition}\label{prop:vannp}
Let $\psi$ be a solution of \eqref{eq:waveequation} that is localized at angular frequency $\ell \geq 1$. Assume that \eqref{eq:np1as} holds true for some $\epsilon > 0$ and 
$$I_{\ell} [\psi ] = 0 , \quad I_{\ell}^{(1)} [\psi ] \neq 0 , $$
for $I_{\ell}^{(1)} [\psi ] $ the time inverted Newman--Penrose charge from Section \ref{timeinverse}. Then we have that:
\begin{equation}\label{asym:tirf}
\begin{split}
\phi (u,v,\theta , \varphi ) = & - \frac{2^{4\ell} I_{\ell}^{(1)} [\psi ] ( \theta , \varphi )}{ (2\ell+1 ) \cdot \dots \cdot ( \ell +2 )} ( v-u)^{\ell+1} u^{-2-\ell} v^{-1-\ell} + O ( u^{-2-\ell-\eta} v^{-1-\ell} ) \\ & \mbox{  for $(u,v) \in \mathcal{B}_{\alpha_{\ell}}$ and for some $\eta > 0$,} 
\end{split}
\end{equation}
and
\begin{equation}\label{asym:tiw}
\begin{split}
\frac{\psi (\tau , r , \theta , \varphi )}{r^{\ell}} =&  A_{\ell , 0}^{(1)} \tau^{-2\ell -3} + O ( \tau^{-2\ell-3-\eta} )  \\ & \mbox{  for any $\tau \geq u_0$, $r_+ \leq r < \infty$, and some $\eta > 0$,}
\end{split}
\end{equation} 
where $A_{\ell , 0}^{(1)}$ is the quantity $A_{\ell , 0 , 1}$ given in \eqref{asym:tpsi} but depending on $I_{\ell}^{(1)}$ instead of $I_{\ell}$.
\end{proposition}

\appendix
\section{Basic inequalities}
We gather here several auxiliary results that we use throughout the paper.
\begin{lemma}
\label{hardy}
Let $q\in \R \setminus \{-1\}$ and $\zeta :[r_0,\infty)\to\R$ a $C^1$ function. Then,
\begin{align}
\label{hardy1}
\int_{r_0}^{\infty} r^q \zeta^2(r)\,dr&\leq \frac{4}{(q+1)^2}\int_{r_0}^{\infty} r^{q+2}(\partial_r \zeta)^2(r)\,dr\\
& \textnormal{if} \: \zeta(r_0)=0\:\textnormal{and}\:\lim_{r\to \infty} r^{q+1}\zeta^2(r)=0, \nonumber\\
\label{hardy2}
\int_{r_{0}}^{\infty} \zeta^2(r)\,dr&\leq 4\int_{r_{0}}^{\infty} (r-r_{0})^2 (\partial_r \zeta)^2(r)\,dr \quad \textnormal{if}\: \lim_{r\to \infty} r \zeta^2(r)=0.
\end{align}
\end{lemma}

\begin{lemma}[Poincar\'e inequality on $\s^2$]
\label{poincare}
Fix $\ell \geq 1$. Then
\begin{equation}
\label{poincare1}
\int_{\s^2}\zeta_{\geq \ell}^2\,d\omega  \leq \frac{1}{\ell (\ell +1)}r^{-2}\int_{\s^2}|\snabla \zeta_{\geq \ell}|^2\,d\omega.
\end{equation}
In the case $\zeta$ is supported on a single angular mode the inequality becomes an equality:
\begin{equation}
\label{poincare2}
\int_{\s^2}\psi_{\ell}^2\,d\omega=\frac{1}{\ell (\ell +1)}r^{-2}\int_{\s^2}|\snabla \psi_{\ell}|^2\,d\omega.
\end{equation}
Furthermore,
\begin{equation}
\label{poincare3}
\int_{\s^2}r^2|\snabla \zeta|^2\,d\omega  \leq \int_{\s^2}(\slashed{\Delta}_{\s^2} \zeta )^2\,d\omega.
\end{equation}
\end{lemma}

\begin{lemma}
\label{sobolev}
There exists a numerical constant $C>0$ such that
\begin{equation}
\label{sobolevs1}
\sup_{\omega \in \s^2} \zeta^2(\omega)\leq C \sum_{|k|\leq 2}\int_{\s^2}(\Omega^k \zeta)^2(\omega)\,d\omega.
\end{equation}
\end{lemma}

\begin{lemma}
\label{ineq:interpol}
Let $\zeta: \R_+ \times [R,\infty)\to \R$ be a function such that the following inequalities hold:
\begin{align}
\label{interpol1}
\int_R^{\infty}r^{p-\epsilon} \zeta^2(\tau,r)\,dr\leq & C_1(1+\tau)^{-q},\\
\label{interpol2}
\int_R^{\infty}r^{p+1-\epsilon} \zeta^2(\tau,r)\,dr\leq & C_2(1+\tau)^{-q+1},
\end{align}
for some $\tau$-independent constants $C_1,C_2>0$, $q\in \R$ and $\epsilon\in (0,1)$.

Then
\begin{equation}
\label{interpol3}
\int_R^{\infty}r^{p}f^2(\tau,r)\,dr\lesssim \max\{C_1,C_2\}(1+\tau)^{-q+\epsilon}.
\end{equation}
\end{lemma}
For the proofs of all the above inequalities see Section 2.5 of \cite{paper1}.

\begin{lemma}
\label{lm:integralsandweights}
Let $f\in C^{0}([x_0,\infty))$. Let $ n\in \N$ such that $\lim_{x\to \infty} x^{n+1} |f(x)|=0$. Then 
\begin{equation}
\label{eq:xweightid}
\int_{x_0}^{\infty} (x-x_0)^n f(x)\,dx=n! \int_{x_0}^{\infty} \int_{x_1}^{\infty}\ldots \int_{x_n}^{\infty} f(x_{n+1})\,dx_{n+1}dx_n\ldots dx_1.
\end{equation}
\end{lemma}
For a proof of the aforementioned lemma see Lemma 6.9 of \cite{aagscat}.

\begin{lemma}
\label{gronwall}
Let $\zeta : [0,\infty)\to \R$ be a continuous, positive function. Assume that for all $0\leq \tau_1\leq \tau_2<\infty$,
\begin{equation}
\label{eq:intineqf}
\zeta (\tau_2)+\beta \int_{\tau_1}^{\tau_2}\zeta (s)\,ds\leq \zeta(\tau_1)+E_0(\tau_2-\tau_1)(\tau_1+1)^{-p},
\end{equation}
with $E_0,\beta ,p>0$ being constants and moreover, for all $0\leq \tau_0\leq \tau_1\leq \tau_2$
\begin{equation}
\label{eq:intineqf2}
\zeta (\tau_2)+\beta \int_{\tau_1}^{\tau_2}\zeta (s)\,ds\leq \zeta (\tau_1)+C_0 (\tau_2-\tau_1) \zeta (\tau_0).
\end{equation}
Then we have that
\begin{equation}
\label{eqboundf}
\zeta (\tau)\leq \left(1+C_0 \beta^{-1}\right) \zeta (\tau_0)
\end{equation}
for all $\tau\geq \tau_0$ and there exists a constant $C=C(C_0,E_0,\beta,p)>0$, such that
\begin{equation}
\label{eqdecayf}
\zeta (\tau)\leq C( \zeta (\tau_0)+E_0)(1+\tau)^{-p},
\end{equation}
for all $\tau\geq \tau_0$.
\end{lemma}
For a proof of the lemma above see Lemma 7.4 of \cite{paper2}. Finally we need the following basic computation.

\begin{lemma}\label{lm:integrals}
For any $\ell \geq 1$ there exists $\eta'=\eta'(\eta,\alpha)>0$ such that for any $v \geq v_{\gamma_{\alpha_{\ell}}}$ 
\begin{equation}
\label{eq:importantvintegral}
\int_{v_{\gamma_{\alpha_{\ell}}}(u)}^v \frac{(v'-u)^k}{(v')^{k+2}}(1+O((v')^{-\eta}))\,dv'=\left[\frac{(v-u)^{k+1}}{(k+1)uv^{k+1}}\right](1+O(v^{-\eta'})+O(u^{-\eta'})) ,
\end{equation}
where $\gamma_{\alpha_{\ell}} = \{ (u,v) : v=u=v^{\alpha_{\ell}} \}$ for $\frac{2\ell +2}{2\ell +3} < \alpha_{\ell} < 1$. Moreover for any $m \in \mathbb{N}$ we have that:
\begin{equation}
\label{eq:importantvintegralt}
\begin{split}
\int_{v_{\gamma_{\alpha_{\ell}}}(u)}^v & (v'-u)^k L^m \left( \frac{1}{(v')^{k+2}} \right) (1+O((v')^{-\eta}))\,du'\\ = & \left[(v-u)^{k+1}T^m \left( \frac{1}{(k+1)uv^{k+1}}\right) \right](1+O(v^{-\eta'})+O(u^{-\eta'})) .
\end{split}
\end{equation}
\end{lemma}
\begin{proof}
We first note that
$$ \frac{1}{k+1} L \left[ \frac{1}{u} \left( \frac{v-u}{v} \right)^{k} \right]  = \frac{(v-u)^{k}}{v^{k+2}} , $$
and estimate \eqref{eq:importantvintegral} follows from the above computation. On the other hand we note that
$$ (v-u)^k L^m \left( \frac{1}{v^{k+2}} \right) = T^m \left( \frac{(v-u)^k}{v^{k+2}} \right) , $$
and using this and the previous computation we also get \eqref{eq:importantvintegralt}.
\end{proof}

\bibliographystyle{alpha}

\end{document}